\documentclass[10pt,aps,superscriptaddress,twocolumn]{revtex4-2}
\usepackage[utf8]{inputenc}
\usepackage{braket}
\usepackage{graphicx}
\usepackage{amssymb}
\usepackage{tabularx}
\usepackage{bm}
\usepackage{amsmath}
\usepackage{float}
\usepackage{bbm}
\usepackage{setspace}
\usepackage{dsfont}
\usepackage{amsthm}
\usepackage{multirow}
\usepackage{pifont}
\usepackage{tikz}
\usepackage{mathtools}
\usepackage{booktabs}
\usepackage{makecell}  
\usepackage{multirow}

\usepackage[caption=false]{subfig}

\usepackage[ruled,vlined]{algorithm2e}
\usepackage{hyperref}
\usepackage{cleveref}

\newtheorem{lem}{Lemma}
\newcommand{\rk}{\text{rk}}

\begin{document}

\def\ucl{Department of Physics \& Astronomy, University College London, London, WC1E 6BT, United Kingdom}

\title{Single-Shot Decoding and Fault-tolerant Gates with Trivariate Tricycle Codes}

\author{Abraham Jacob}
\thanks{These authors contributed equally.}
\affiliation{\ucl}

\author{Campbell McLauchlan}
\thanks{These authors contributed equally.}
\affiliation{\ucl}
\affiliation{Centre for Engineered Quantum Systems, School of Physics, University of Sydney, Sydney, NSW 2006, Australia}

\author{Dan E. Browne}
\affiliation{\ucl}

\begin{abstract}
While quantum low-density parity check (qLDPC) codes are a low-overhead means of quantum information storage, it is valuable for quantum codes to possess fault-tolerant features beyond this resource efficiency.
In this work, we introduce trivariate tricycle (TT) codes, qLDPC codes that combine several desirable features: high thresholds under a circuit-level noise model, partial single-shot decodability for low-time-overhead decoding, a large set of transversal Clifford gates and automorphisms within and between code blocks, and (for several sub-constructions) constant-depth implementations of a (non-Clifford) $CCZ$ gate. 
TT codes are CSS codes based on a length-3 chain complex, and are defined from three trivariate polynomials, with the 3D toric code (3DTC) belonging to this construction. We numerically search for TT codes and find several candidates with improved parameters relative to the 3DTC, using up to 48$\times$ fewer data qubits as equivalent 3DTC encodings. %
We construct syndrome-extraction circuits for these codes and numerically demonstrate single-shot decoding in $Z$ memory, in both phenomenological and circuit-level noise models. We also demonstrate that the codes have good performance under single-shot decoding even in $X$ memory. Under circuit-level noise, TT codes have a crossing point of $0.3\%$ in $X$ memory and $1\%$ in $Z$ memory (with single-shot decoding).
All TT codes possess several transversal $CZ$ gates that can partially address logical qubits between two code blocks. Additionally, the codes possess a large set of automorphisms that can perform Clifford gates within a code block.
Finally, we establish several TT code polynomial constructions that allow for a constant-depth implementation of logical $CCZ$ gates. We find codes with $CCZ$ gates that use up to $5.7\times$ fewer data qubits as equivalent-distance 3DTC encodings, and others with encoding rates as high as $1/4$. 
\end{abstract}

\maketitle

\section{Introduction}

Quantum error-correction (QEC) is a necessary feature of any large-scale quantum computing architecture. The surface code~\cite{bravyi1998quantumcodeslatticeboundary,dennis2002topological,Kitaev_2003,Fowler_2012} is a very promising candidate QEC code, with mature fault-tolerant schemes for performing logic on encoded qubits~\cite{Litinski_2019_GoSCs}, including lattice surgery~\cite{Horsman_2012}, low overhead measurements in all bases~\cite{Gidney_2024_Y_mmt}, and twist and hole braiding~\cite{Brown_2017}. However, even with recent theoretical results on resource savings with the surface code~\cite{gidney2024magicstatecultivationgrowing,gidney2023yokedsurfacecodes}, the very large resource estimates for achieving useful quantum advantage~\cite{beverland2022assessingrequirementsscalepractical} mean that we should consider the use of higher-rate quantum low-density parity check (qLDPC) codes.

The bivariate bicycle (BB) codes~\cite{bravyi2024high} are a class of code with very good encoding rates and relative distances. They perform well under a circuit-level noise model, exhibiting high thresholds comparable to the surface code. As such, they are a promising candidate for a qLDPC code that can help lower the overheads of fault-tolerant quantum computing~\cite{yoder2025tourgrossmodularquantum}. There has been a great deal of investigation into the classes to which these codes belong, with the aim of finding codes with improved parameters and properties~\cite{Kovalev_Pryadko_2013,voss2025multivariatebicyclecodes,wang2024coprimebivariatebicyclecodes,lin2023quantumtwoblockgroupalgebra,liang2025generalizedtoriccodestwisted,liang2025planarquantumlowdensityparitycheck}, and implementing them in hardware~\cite{ye2025quantumerrorcorrectionlong,wang2025demonstrationlowoverheadquantumerror}.

Meanwhile, there has been interest in developing codes that display other useful features. To address the large time overheads and associated classical compute resources required for decoding qLDPC codes~\cite{beverland2022assessingrequirementsscalepractical}, much effort has gone into finding codes with single-shot decodability~\cite{Campbell_2019,quintavalle2021single,Higgott_2023,scruby2024radial_codes,lin2025GB_codes} and other single-shot properties~\cite{hong2024singleshotpreparationhypergraphproduct,hillmann2024singleshot}. Single-shot codes can allow for decoding to be performed on small amounts of syndrome data, collected over a constant number of QEC cycles, without incurring a significant drop in accuracy. The presence of meta-checks, which help determine when syndrome measurement errors have been made, can assist with this type of strategy, as well as allowing for operations such as lattice surgery to be performed with constant time-overhead~\cite{hillmann2024singleshot}.

A large set of transversal gates and automorphisms is useful for fault-tolerant computing with qLDPC codes, and recent work has focused on finding codes with \mbox{(fold\mbox{-})}transversal Clifford gates~\cite{Kubica_2015,breuckmann2024fold,Quintavalle_2023,malcolm2025computingefficientlyqldpccodes,eberhardt2024logicaloperatorsfoldtransversalgates}, automorphisms~\cite{bravyi2024high,sayginel2025faulttolerantlogicalcliffordgates}, and low-depth non-Clifford gates~\cite{breuckmann2024cupsgatesicohomology,lin2024transversalnoncliffordgatesquantum,golowich2024quantumldpccodestransversal,nguyen2024goodbinaryquantumcodes}. The former two provide logical primitives that can speed up an array of fault-tolerant computing schemes~\cite{malcolm2025computingefficientlyqldpccodes}, even those based on Pauli-based computation (involving the measurement of high-weight logical operators within and between code blocks)~\cite{he2025extractorsqldpcarchitecturesefficient,yoder2025tourgrossmodularquantum}. The latter can provide new avenues towards low-overhead magic-state preparation~\cite{Daguerre_2025} and distillation~\cite{wills2024constantoverheadmagicstatedistillation}.

Despite this recent intense activity, a thorough investigation of codes that \textit{combine} several of these desirable properties has been lacking. To that end we introduce the trivariate tricycle (TT) codes, generalisations of the BB codes, that exhibit several useful features. These codes are based on a length-3 chain complex and hence possess meta-checks for one of the bases, which we take to be the $Z$ check basis. We introduce  a quantity which we call  ``meta-check distance", useful for considerations such as generalised lattice surgery~\cite{hillmann2024singleshot}, and show that it is equal to the $Z$-distance of the code. The 3D toric code (3DTC) belongs to this code family, and indeed, we can lay out the TT codes on a cubic lattice without boundary, with additional long-range connections. We numerically find several medium-size examples of TT codes that exhibit parameters that far exceed those of the surface code and 3DTC, for an equivalent number of logical qubits and code distance. %
We identify a minimum-depth syndrome measurement circuit for the TT codes and numerically evaluate the codes' performance under phenomenological noise and circuit-level noise, finding that TT codes possess high thresholds despite the large check weights. We evaluate the single-shot performance of the codes (under both noise models) using an overlapping window strategy~\cite{scruby2024radial_codes}. We numerically observe that $Z$ memory experiment performance displays strong evidence of single-shot decodability, while even for $X$ memory, small TT codes perform optimally when using an overlapping window strategy. 

We find that several logical $\overline{CZ}$ gates can be performed transversally on the $TT$ code logical qubits between two blocks of code, while shift automorphisms can perform certain Clifford gates within a code block. Up to in-block automorphisms, there are three non-equivalent logical $\overline{CZ}$ gates. For many of the TT codes we discuss, the logical qubits can be divided into one of three sets. Two of the three logical $\overline{CZ}$ gates have a simple form on these logical sets, preserving one and entangling the other two. 

Finally, we introduce several TT code constructions that allow a constant-depth circuit of $CCZ$ gates to preserve the code space. We thereby find codes with non-Clifford gates whose parameters greatly outperform those of the 3DTC. As an example, our $[[114, 3, 6]]$ code with $Z$-distance 6 and $X$-distance 19 uses $17.6\%$ of the data qubits of the regular 3DTC (with $Z$-distance 6), while maintaining the same check weights. We introduce different code constructions that produce codes with significantly better encoding rates, up to $1/4$, that also admit non-trivial $CCZ$ gates. These latter codes, however, have small distances ($\geq 3$) relative to their maximum check weights, which are between $12$ and $22$ for the codes presented, and their distances cannot be grown without growing their check weights. 

This paper is structured as follows. In Section~\ref{sec:TT_codes_intro}, we introduce the code construction and prove several important Lemmas relating to the code parameters and code equivalences. We also introduce a layout for these codes on a 3D cubic grid with additional long-range connections. Section \ref{sec: single-shot} investigates the performance of TT codes under phenomenological noise, showing evidence of a threshold for small-to-medium instances of TT codes in memory experiments. We investigate their performance under overlapping window decoding, and also investigate their confinement properties~\cite{quintavalle2021single,Bombin_2015}. In Section \ref{sec:CLN_Sims} we  present circuit-level noise simulations of TT code performance, demonstrating evidence for thresholds in both $X$ and $Z$ memory experiments and single-shot capabilities for $Z$ memory using an overlapping window decoding method. We also present data indicating good performance of TT codes under overlapping window decoding even for $X$ memory. In Section~\ref{sec:logic_gates}, we introduce the decomposition of logical qubits into the three logical sets, along with the shift automorphisms and transversal $CZ$ gates that are possible in any TT code. We then discuss the sub-constructions that allow for constant-depth $CCZ$ circuits to be defined. We prove that a large class of such constructions exist, and present codes with logical $\overline{CCZ}$ gates found through exhaustive searches over small code sizes. In Section~\ref{sec:conclusion} we conclude and present several directions for future work. In the Appendices, we present further examples of TT codes found through numerical search and prove certain TT code properties, provide further numerical data, provide details on simulations and prove that our syndrome measurement circuits function as claimed, and provide further details on methods relating to codes with constant-depth $CCZ$ circuits.
\raggedbottom

\section{3-Block and Trivariate Tricycle Codes}\label{sec:TT_codes_intro}

We first define arbitrary 3-block group algebra codes (for Abelian group $G$), before introducing the trivariate tricycle codes more explicitly. For a general introduction on group algebra code constructions, see Ref.~\cite{lin2023quantumtwoblockgroupalgebra}.
We will be considering the group algebra $\mathbb{F}_2[G]$, consisting of formal sums of the form $\sum_{g\in G} a_g g$, with $a_g\in \mathbb{F}_2$. The sum and product of two elements from $\mathbb{F}_2[G]$ are defined in the natural way.
For a 3-block code, we introduce three length-1 chain complexes -- $\mathcal{A}_1 \xrightarrow{a}\mathcal{A}_0$, $\mathcal{B}_1 \xrightarrow{b}\mathcal{B}_0$, and $\mathcal{C}_1 \xrightarrow{c}\mathcal{C}_0$ -- where $\mathcal{A}_i,\mathcal{B}_i, \mathcal{C}_i = \mathbb{F}_2[G]$ and $a,b,c\in\mathbb{F}_2 [G]$. We will then be considering the tensor/balanced product complex formed from these three length-1 complexes. For general treatments of balanced products, see Refs.~\cite{panteleev2022asymptotically,Breuckmann_2021_BP_codes,breuckmann2024cupsgatesicohomology}. We need not define the tensor product of $\mathbb{F}_2[G]$-modules in general, since we can simply note that $\mathbb{F}_2[G] \otimes_G \mathbb{F}_2[G] \simeq \mathbb{F}_2[G]$. The tensor product complex is then:
\begin{align}
    &\mathcal{A}_1\otimes_G \mathcal{B}_1\otimes_G \mathcal{C}_1 \xrightarrow{\partial_3}\bigoplus_{i+j+k = 2}\mathcal{A}_i\otimes_G \mathcal{B}_j\otimes_G \mathcal{C}_k \xrightarrow{\partial_2} \nonumber\\
    &\bigoplus_{i+j+k = 1}\mathcal{A}_i\otimes_G \mathcal{B}_j\otimes_G \mathcal{C}_k \xrightarrow{\partial_1} \mathcal{A}_0\otimes_G \mathcal{B}_0\otimes_G \mathcal{C}_0.
    \label{eqn:chain_complex}
\end{align}
The boundary maps can be written in terms of $a$, $b$, and $c$ as:
\begin{align}
    \partial_3 &= \begin{bmatrix}
        a\\
        b\\
        c
    \end{bmatrix}\label{eqn:del_3}\\
    \partial_2 &= \begin{bmatrix}
        0 & c & b\\
        c & 0 & a\\
        b & a & 0
    \end{bmatrix}\\
    \partial_1 &= \begin{bmatrix} a & b & c\end{bmatrix}.\label{eqn:del_1}
\end{align}

We will now construct the trivariate tricycle codes from this abstract chain complex. 
We first set $G = \mathbb{Z}_\ell \times \mathbb{Z}_m \times \mathbb{Z}_p$, for integers $\ell$, $m$, and $p$. Define the following matrices which together generate this group:
\begin{align}
    x &= S_\ell\otimes \mathds{1}_m\otimes \mathds{1}_p ,\\
    y &= \mathds{1}_\ell\otimes S_m \otimes \mathds{1}_p,\\
    z &= \mathds{1}_l \otimes \mathds{1}_m\otimes S_p ,
\end{align}
where $S_\ell$ is the cyclic shift matrix, whose elements are $\delta_{i,i\oplus 1}$, where $\oplus$ here refers to addition mod $\ell$. We then introduce matrices $A$, $B$ and $C$, which will be polynomials in these three variables. These will take the positions of $a$, $b$, and $c$ in Equations~\ref{eqn:del_3}--\ref{eqn:del_1}. For now, we will consider each polynomial to contain three terms:
\begin{align}
    A &= A_1 + A_2 + A_3 \label{eqn: A poly}\\
    B &= B_1 + B_2 + B_3 \label{eqn: B poly} \\
    C &= C_1 + C_2 + C_3 \label{eqn: C poly}
\end{align}
though we will later relax this restriction. We then let $H_X \equiv \partial_1$, $H_Z = \partial_2^\top$ and $M_Z = \partial_3^\top$, where $H_X$ and $H_Z$ represent the usual parity check matrices of a CSS code, and $M_Z$ is a matrix of \textit{meta-checks}, encoding the redundancy in the $Z$ checks. It can be verified that $H_X H_Z^\top = M_Z H_Z = 0$.
To summarise, this results in the following chain complex structure:
\begin{align}
    \mathcal{C}_M \xleftarrow{M_Z} \mathcal{C}_Z \xleftarrow{H_Z}\mathcal{Q} \xrightarrow{H_X}\mathcal{C}_X
\end{align}
with all spaces defined as vector spaces over $\mathbb{F}_2$: $\mathcal{Q}, \mathcal{C}_X, \mathcal{C}_Z$, and $\mathcal{C}_M$ correspond to qubits, $X$-checks, $Z$-checks and $Z$-meta-checks, respectively. The matrices are given by:
\begin{align}
    M_Z &= \begin{bmatrix}
        A^\top & B^\top & C^\top
    \end{bmatrix}\label{eqn:PCM_MZ}\\
    H_Z &= \begin{bmatrix}
        0 & C^\top & B^\top\\
        C^\top & 0 & A^\top\\
        B^\top & A^\top & 0
    \end{bmatrix} \label{eqn:PCM_Z}\\
    H_X &= \begin{bmatrix} A & B &C\end{bmatrix}.\label{eqn:PCM_X}
\end{align}

\subsection{Code parameters and explicit examples}\label{sec:code_params}

\begin{table*}
    \begin{center}
        \def\arraystretch{1.3}
        \[\begin{array}{|c|c|c|c|c|c|c|c|}
        \hline
        \mathbf{[[n, k, d]]} & \begin{array}{c}
        \textbf{Encoding } \\
        \textbf{ Rate }
        \end{array} & \mathbf{\ell, m, p} & \mathbf{A} & \mathbf{B} & \mathbf{C} & 
        \mathbf{d_X} & \mathbf{d_Z} \\
        \hline \hline [[72,6,6]] & 1 / 12 & 4,3,2 & 1+ y + xy^2 & 1 + yz  + x^2y^2 & 1 + xy^2z + x^2y & 
        12 & 6\\
        \hline [[180,12, 8]] & 1 / 15 & 5,4,3 & 1+ x^2y^3z + x^4y & 1 + x^3  +x^4z^2 & 1 + x^3y^3 + x^4yz^2 & 
        \leq 20 & 8 \\
        \hline [[432,12, 12]] & 1 / 36 & 6,6,4 & 1+ xyz^3 + x^3y^4z^2 & 1 + x^3yz^2 + x^3y^2z^3 & 1 + x^4y^3z^3 + x^5z^2 & 
        \leq 36 & 12 \\
        \hline
        \end{array}\]
    \end{center}
    \caption{Example trivariate tricycle codes and their parameters, along with the polynomials $A$, $B$ and $C$, and dimensions $\ell$, $m$ and $p$, that define them. We report the encoding rates $r = k/n$ along with the $Z$ and $X$ distances, which are inequivalent for TT codes.   \label{tab:mycodes}}
\end{table*}

Here, we discuss the parameters of and equivalences among TT codes, before introducing explicit examples. 
A trivariate tricycle (TT) code is defined on $n= 3\ell m p$ data qubits. We begin by discussing the number of logical qubits. We have: 
\begin{lem}\label{lem:Num_logical_qubits}
For a TT code with $H_Z$ and $H_X$ defined in Equations~\ref{eqn:PCM_Z} and \ref{eqn:PCM_X}, respectively, the number of logical qubits is given by:
\begin{align}
\begin{split}
    k = -\frac{n}{3} + \dim(\ker A \cap \ker B\cap  \ker C) + \\
    \dim(\ker[0\; C\; B]\cap \ker[C\;0\;A]\cap \ker[B\; A\; 0]).
\end{split}
\end{align}
\end{lem}
\begin{proof}
As usual for a CSS code, $k = n - \text{rk}(H_X) - \text{rk}(H_Z)$, where we use rk$(M)$ to refer to the rank of matrix $M$. Now note that:
\begin{align}
    \rk(H_X) = \frac{n}{3} - \dim (\ker \, H_X^\top).
\end{align}
Let $\bar{C}_\ell$ be the $\ell \times \ell$ permutation matrix with elements $\delta_{-i,i}$, where the minus is taken mod $\ell$. This has the property that $\bar{C} S_\ell \bar{C} = S_\ell^\top$. Similarly define $\bar{C}_m$ and $\bar{C}_p$, and let $\bar{C} = \bar{C}_\ell \otimes \bar{C}_m\otimes \bar{C}_p$. Then we have:
\begin{align}
    H_X^\top = \begin{bmatrix}
        A^\top\\
        B^\top\\
        C^\top
    \end{bmatrix} = \begin{bmatrix}
        \bar{C} & 0 & 0\\
        0 & \bar{C} & 0\\
        0 & 0 & \bar{C}
    \end{bmatrix} \begin{bmatrix}
        A\\
        B\\
        C
    \end{bmatrix} \bar{C}.
\end{align}
Hence, $H_X^\top$ differs from $[A^\top B^\top C^\top]^\top$ by multiplication on the left and right by invertible matrices (since $\bar{C}$ is self-inverse). Therefore, the dimensions of their kernels are equivalent. Finally,
\begin{align}
    \dim\left(\ker \begin{bmatrix}A \\ B \\ C\end{bmatrix}\right) = \dim(\ker{A}\cap\ker{B}\cap\ker{C}).
\end{align}
Similarly, we can see that rk$(H_Z) = \text{rk}(H_Z^\top) = n - \dim(\ker[0\; C\; B]\cap \ker[C\;0\;A]\cap \ker[B\; A\; 0])$. 
\end{proof}

The code distances $d_X$ and $d_Z$ are defined as:
\begin{align}
    d_X &= \min \lbrace |v| \, : \, v\in \ker (H_Z) \backslash \text{rs}(H_X)\rbrace\\
    d_Z &= \min \lbrace |v| \, : \, v\in \ker (H_X) \backslash \text{rs}(H_Z)\rbrace
\end{align}
where $\text{rs}(M)$ denotes the row space of matrix $M$. In TT codes, $d_X \neq d_Z$ in general, and we observe that in all cases tested, $d_X \geq d_Z$. We can also introduce a quantity we call the \textit{meta-check distance} similarly:
\begin{align}
    d_M = \min\lbrace |v| \, : \, v \in \ker(M_Z)\backslash \text{cs}(H_Z)\rbrace
\end{align}
where $\text{cs}(M)$ denotes the column space of $M$. Because the meta-checks encode redundancies in the $Z$-checks, they can be used to repair faulty syndromes resulting from classical measurement errors. The interpretation of the meta-check distance is that it is the minimum number of classical measurement errors that can be made when measuring the $Z$-checks without triggering any meta-checks. A large meta-check distance ensures good performance under single-shot decoding, since we are likely to repair the syndrome accurately. A large meta-check distance also allows for a reduced time-overhead during logical operations like lattice surgery or state preparation~\cite{hillmann2024singleshot,hong2024singleshotpreparationhypergraphproduct}. We have the following:

\begin{lem}\label{lem:distances}
    For a TT code defined as above, $d_M = d_Z$, and $d_X = \min\lbrace |v|\, : \, v\in \ker(H_Z^\top) \backslash \text{rs}(M_Z)\rbrace$.
\end{lem}
\begin{proof}
    Suppose that $v$ is a minimum-weight ``check logical", that is, it is in $\ker(M_Z)\backslash \text{cs}(H_Z)$. Let $v = [\alpha^\top \; \beta^\top \; \gamma^\top ]^\top$. Then define the following:
    \begin{align}
        \bar{v} = \begin{bmatrix}
            \bar{C} \alpha\\
            \bar{C} \beta\\
            \bar{C} \gamma
        \end{bmatrix}.
    \end{align}
    This vector is in the kernel of $H_X$, since:
    \begin{align}
        H_X \bar{v} = [A \; B \; C] \begin{bmatrix}
            \bar{C} \alpha\\
            \bar{C} \beta\\
            \bar{C} \gamma
        \end{bmatrix} = \bar{C}(A^\top \alpha + B^\top \beta + C^\top \gamma) = 0.
    \end{align}
    Therefore $\bar{v}$ is a logical operator unless it belongs to the row space of $H_Z$. In this case, $\bar{v} = H_Z^\top w$, for some vector $w$. But then:
    \begin{align}
        \begin{bmatrix}
            \bar{C} & 0 & 0\\
            0 & \bar{C} & 0 \\
            0 & 0 & \bar{C}
        \end{bmatrix} v = H_Z^\top w \implies v &= \begin{bmatrix}
            \bar{C} & 0 & 0\\
            0 & \bar{C} & 0 \\
            0 & 0 & \bar{C}
        \end{bmatrix} H_Z^\top w \\
        &= H_Z \begin{bmatrix}
            \bar{C} & 0 & 0\\
            0 & \bar{C} & 0 \\
            0 & 0 & \bar{C}
        \end{bmatrix} w,
    \end{align}
    and so $v \in \text{cs}(H_Z)$, which contradicts the assumption that it was a check logical. Hence, $\bar{v}$ is a logical $Z$ operator. It has the same weight as $v$ since $\bar{C}$ does not change the weight of the vectors $\alpha$, $\beta$ and $\gamma$. Since we assumed that $v$ was a minimum-weight check logical, we see that $d_Z \leq d_M$. One can, in exactly the same way, show that $d_M \leq d_Z$.

    One can show the second part of the Lemma in the same way, by considering a minimum weight $X$-logical, $v$, and the corresponding $\bar{v}\in \ker(H_Z^\top) \backslash \text{rs}(M_Z)$, which has the same weight. 
\end{proof}

Certain different polynomials produce codes with equivalent parameters. Observing these equivalences can speed up code searches. We first observe that, trivially, if we permute $A$, $B$ and $C$, we obtain an equivalent code. Furthermore, taking the transpose of $A$, $B$ and $C$ also results in an equivalent code: the ranks of $H_X$ and $H_Z$ can easily be shown to be unchanged, and the distances $d_X$, $d_Z$, and $d_M$ are unchanged as a result of Lemma~\ref{lem:distances}. We finally show that, if one of the polynomials has a common monomial factor, this can be removed without affecting the code parameters. Without loss of generality, we treat $A$ as the polynomial that can be simplified in this way.

\begin{lem}\label{lem:code_equivalence}
    Consider the $[[n,k,d]]$ TT code defined by polynomials $A$, $B$ and $C$, with parity check matrices $H_X$ and $H_Z$. Suppose that $A = \alpha^\top \tilde{A}$ for some monomial $\alpha$ and polynomial $\tilde{A}$. Then define the TT code with parity check matrices $\tilde{H}_X$ and $\tilde{H}_Z$, based on polynomials $\tilde{A}$, $B$ and $C$. This defines an $[[n,k,d]]$ code. Additionally, the distances $d_X$ and $d_Z=d_M$ are separately preserved.
\end{lem}
\begin{proof}
    The fact that $n$ is unchanged is trivial. We show that $k$ is unchanged. Notice that:
    \begin{align}
        \tilde{H}_X = [\alpha A | B | C] = [A | B | C]\begin{bmatrix}
            \alpha & 0 & 0 \\
            0 & 1 & 0\\
            0 & 0 & 1
        \end{bmatrix},
    \end{align}
    and so, owing to the fact that $\alpha$ is invertible, $\tilde{H}_X$ is related to $H_X$ by an invertible matrix, meaning that rk$\tilde{H}_X = \text{rk}H_X$. Now consider some vector $v\in \ker H_Z$. It can be easily seen that:
    \begin{align}
        \begin{bmatrix}
            1 & 0 & 0\\
            0 & \alpha & 0\\
            0 & 0 & \alpha 
        \end{bmatrix} v \in \ker \tilde{H}_Z.
    \end{align}
    Hence, once again owing to the invertibility of $\alpha$, we have an isomorphism between $\ker H_Z$ and $\ker \tilde{H}_Z$. Therefore, rk$H_Z = \text{rk}\tilde{H}_Z$. Hence, $k = n - \ker H_X - \ker H_Z = n - \ker \tilde{H}_X - \ker \tilde{H}_Z$ is unchanged.

    Now consider the distances $d_X$ and $d_Z$. Suppose that $v_Z = [v_1^\top,v_2^\top,v_3^\top]^\top\in \ker H_X$ is a minimum-weight logical operator. Then clearly $v_Z' = [v_1^\top \alpha,v_2^\top,v_3^\top]^\top \in \ker \tilde{H}_X$. Hence, $v_Z'$ represents either a logical or a $Z$-stabilizer of the code $(\tilde{H}_X,\tilde{H}_Z)$. If it is a stabilizer then it is in the image of $\tilde{H}_Z^\top$, i.e., $v_Z' = \tilde{H}_Z^\top w$, for some $w = [w_1^\top, w_2^\top, w_3^\top]^\top$. But $v_Z' = \text{diag}[\alpha^\top , 1, 1] v_Z$. Hence:
    \begin{align}
        v_Z &= \begin{bmatrix}
            \alpha & 0 & 0\\
            0 & 1 & 0\\
            0 & 0 & 1
        \end{bmatrix} \tilde{H}_Z^\top  \begin{bmatrix}
            1 & 0 & 0\\
            0 & \alpha^\top & 0\\
            0 & 0 & \alpha^\top
        \end{bmatrix} \left( \begin{bmatrix}
            1 & 0 & 0\\
            0 & \alpha & 0\\
            0 & 0 & \alpha
        \end{bmatrix} \begin{bmatrix}
            w_1\\
            w_2\\
            w_3
        \end{bmatrix}\right)\\
        &= H_Z^\top w'
    \end{align}
    where $w ' = [w_1^\top , w_2^\top \alpha^\top , w_3^\top \alpha^\top ]^\top$. Hence, $v_Z \in \text{im} H_Z^\top$, contradicting the assumption that it is a logical operator of the code $(H_X, H_Z)$. Hence, we conclude that $v_Z'$ is a logical operator of $(\tilde{H}_X,\tilde{H}_Z)$. Therefore:
    \begin{align}
        \tilde{d}_Z \leq |v_Z'| = |\alpha v_1| + |v_2| + |v_3| = |v_Z| = d_Z,
    \end{align}
    where $\tilde{d}_Z$ and $d_Z$ are $Z$-distances of the modified and original codes, respectively. We can similarly show that $d_Z\leq \tilde{d}_Z$,
    and so the $Z$-distances are equivalent. By Lemma~\ref{lem:distances}, we find that $\tilde{d}_M = d_M$ as well. An analogous argument shows that $\tilde{d}_X = d_X$. Hence all distances are equivalent.
\end{proof}

We finish this section by listing some TT codes, based on weight-3 polynomials (i.e., those with three distinct terms), that display parameters that outperform those of common topological codes in 2D and 3D, such as the surface code and 3D toric code. In Table~\ref{tab:mycodes}, we list some examples of TT codes found through a random numerical search, along with their parameters and the matrices $A$, $B$ and $C$ that define each code. We list both the $X$ and $Z$ distances (we indicate where distances found are only upper bounds), which were found using a combination of techniques (DistM4RI connected cluster algorithm~\cite{pryadko2024distm4ri}, a BP+OSD decoder as outlined in Ref.~\cite{bravyi2024high}, and QDistRnd~\cite{Pryadko2022} using software from the CSSLO repository \cite{mark_webster_csslo_2024}). We list further TT code examples in Appendix~\ref{app:TT_code_tables}.

\subsection{3D Layout of TT codes}\label{sec:3D_layout}

\begin{figure*}
\begin{tikzpicture}
\node[inner sep=0pt] at (0,0)     {\includegraphics[width=0.8\linewidth]{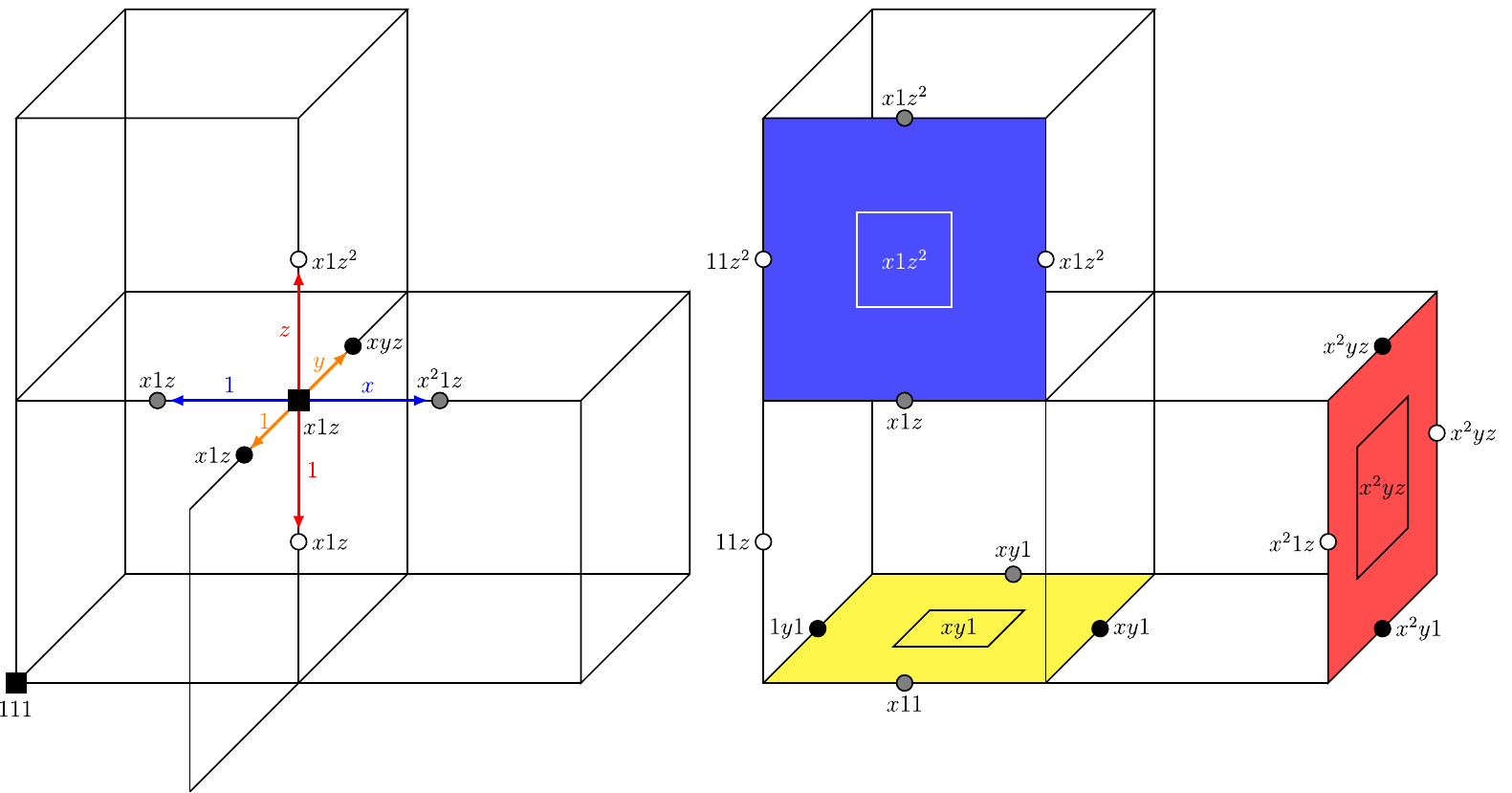}};
\node[inner sep=0pt] at (-7,4) {(a)};
\node[inner sep=0pt] at (0,4) {(b)};
\end{tikzpicture}
    \caption{Layout of the 3D toric code on a cubic lattice. The code's defining polynomials are $A = 1+x$, $B=1+y$, $C=1+z$. Qubits are placed on edges, $X$-checks on vertices and $Z$-checks on faces. (a) $X$-checks are connected to adjacent data qubits by terms in $A$ (for Left qubits, shaded grey), $B$ (for Centre qubits, shaded black), or $C$ (for Right qubits, shaded white). Example $X$-check $x1z$ and its connected qubits is shown. (b) $Z$-checks fall into three families, coloured red ($Z_a$), blue ($Z_b$), and yellow $(Z_c)$, respectively. Each family corresponds to an orthogonal orientation of a plane in the lattice. $Z$-checks are connected to data qubits on the edges surrounding the face by terms from (transposes of) two of the three polynomials. As an example, the red $Z$-checks are connected to Centre qubits by terms in $C^\top$ and Right qubits by terms in $B^\top$. Three example $Z$-checks and their connected qubits are shown. TT codes can be locally laid out on the same 3D cubic lattice, with additional long-range connections. For TT codes built from weight-3 polynomials, the $X$ checks gain 3 long-range connections, and the $Z$ checks gain 2 long-range connections.}
    \label{fig:checks}
\end{figure*}

A conceptual aid exists for bivariate bicycle codes upon laying them out on a 2D grid, with additional long-range connections and with periodic boundary conditions. Similarly, we can locally lay out the Tanner graphs of the TT codes on a 3D cubic lattice with long-range connections. This is most readily illustrated using the 3DTC, which is a special case of a TT code with  $A = 1+x$, $B=1+y$ and $C=1+z$. We lay the code out on a lattice with qubits on the edges, $X$-checks on the vertices and $Z$-checks on the faces. This structure is shown in Fig.~\ref{fig:checks}. We also can assign meta-checks to the 3-cells (cubes) of the lattice (see Fig.~\ref{fig:CCZ_cube} of Section~\ref{sec:CCZ_gates}). %

Also similar to the case of bivariate bicycle codes, the 3D layout elucidates a useful monomial labelling. For a general TT code, we define three blocks of data qubits, labelled L, C and R which correspond to the Left, Center and Right blocks of columns in $H_X$ and $H_Z$. In Fig.~\ref{fig:checks}, the L, C and R qubits inhabit edges pointing in the $x$, $y$ and $z$ directions, respectively (they are coloured grey, black and white). The number of data qubits in each block is given by $\ell m p$, and as such, we can label each qubit in one of the blocks by a distinct monomial $\alpha\in \mathcal{M} \equiv \lbrace x^i y^j z^k : i=1,\ldots,\ell, j=1,\ldots, m,k=1,\ldots, p\rbrace$.

There are also $\ell m p$ $X$-checks (rows of $H_X$), and we also label each of these with a monomial in $\mathcal{M}$. We do so in such a way that $X$-check $\alpha \in\mathcal{M}$ is connected to L qubits $A_i \alpha$ (for all terms $A_i$ in $A$), C qubits $B_i\alpha$, and R qubits $C_i\alpha$. For the 3DTC, this is shown in Fig.~\ref{fig:checks}(a). We explicitly consider an $X$-check labelled by $xz\in\mathcal{M}$. Meanwhile, we have three sets of $Z$ checks, which correspond to the three blocks of rows from $H_Z$. We will refer to these checks, and colour them in Fig.~\ref{fig:checks}(b), as $Z_a$ (red), $Z_b$ (blue), and $Z_c$ (yellow) checks. Geometrically, in the 3DTC, these checks correspond to square faces oriented in three orthogonal directions. The checks in each of these three sets are labelled with monomials from $\mathcal{M}$. The $Z_a$, $Z_b$ and $Z_c$ checks are connected to CR, LR and LC qubits, respectively via terms from $A^\top$, $B^\top$ or $C^\top$. In Fig.~\ref{fig:checks}(b) we consider explicit examples for checks from each set. There are, finally, $\ell m p$ meta-checks resulting from rows in $M_Z$, which we once again label by monomials from $\mathcal{M}$. Meta-check $\alpha$ is such that it is ``connected to" $Z_a$-checks $A^\top \alpha$, $Z_b$-checks $B^\top \alpha$, and $Z_c$ checks $C^\top \alpha$ (see Fig.~\ref{fig:CCZ_cube}). By this we mean that the product of all the $Z$ checks connected to $\alpha$ is the identity. 

A TT code with weight-3 polynomials has 3 additional long-range connections for each $X$-check relative to the 3DTC, and 2 extra long-range connections for each $Z$-check. However, the code's Tanner graph may still possess a 3D toric layout: it does so whenever it possesses a spanning subgraph isomorphic to the Cayley graph of $\mathbb{Z}_\mu\times \mathbb{Z}_\lambda\times \mathbb{Z}_\nu$ for integers $\mu, \lambda,\nu$. Its Tanner graph may also be composed of a single or multiple connected components, though in this paper we focus on those with only a single connected component. In Appendix~\ref{app:TT_code_tables} we report conditions for codes to possess a 3D toric layout and a connected Tanner graph.

\section{Single-shot decoding of TT codes under phenomenological noise} \label{sec: single-shot}

\begin{figure}
\subfloat[\label{fig1}]{\includegraphics[width=\linewidth]{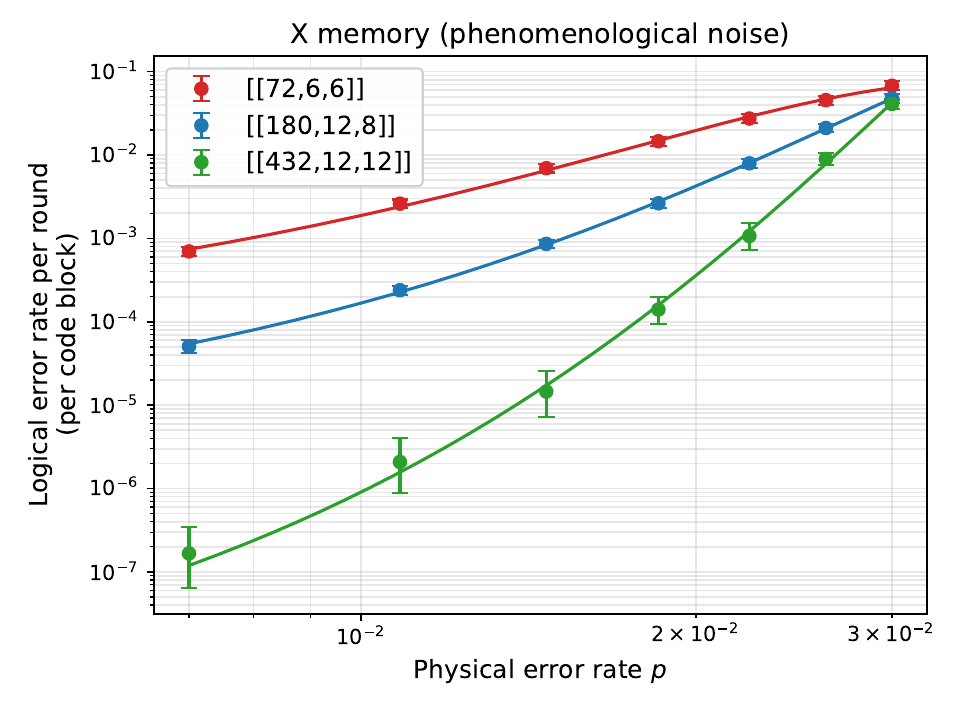}}\\
\subfloat[\label{fig2}]{\includegraphics[width=\linewidth]{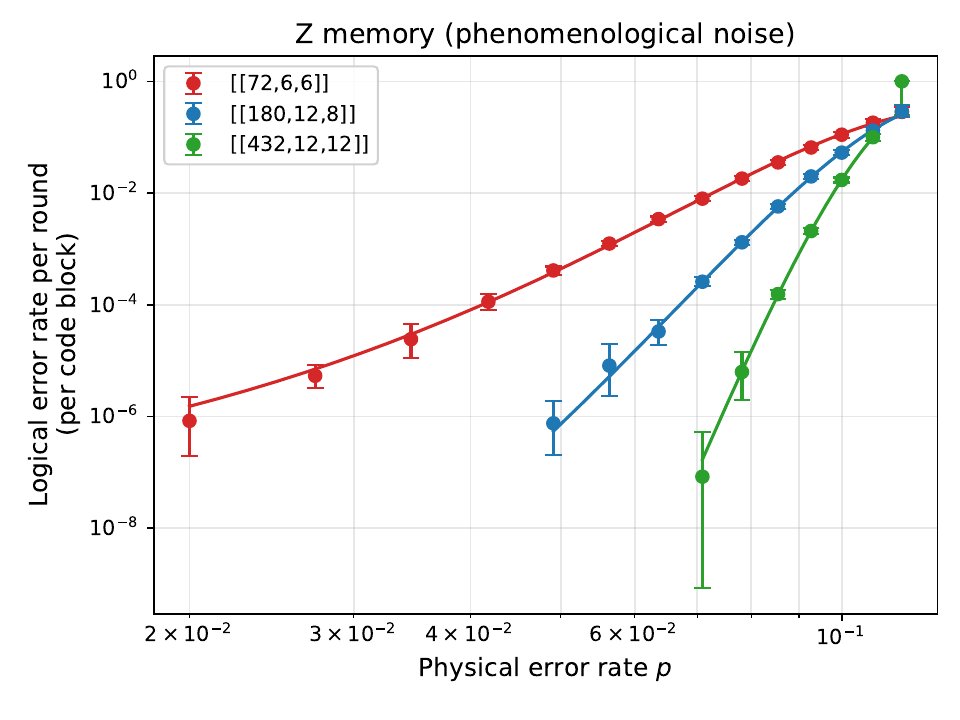}}
\caption{Performance of the TT codes under a phenomenological noise model. We present logical error rates (LERs), without rescaling for the number of logical qubits, in memory experiments performed over $2d_Z$ rounds of syndrome measurements. (a) TT code performance in $X$ memory experiments, for small instances of TT codes from Table~\ref{tab:mycodes}. Decoding was performed with BP+OSD-CS30 over the full syndrome history (100 BP iterations, min-sum scaling factor $0.625$). (b) TT code LERs in $Z$ memory experiments with a (2,1) windowing strategy. Within each window we use BP+OSD-CS0 (150 BP iterations, min-sum scaling factor $1.0$). The asymmetric distances of the TT codes lead to strong error suppression in $Z$ memory, resulting in order of magnitude reductions in the LER, compared to $X$ memory.}
\label{fig:phenom_data}
\end{figure}

\begin{figure}
    \centering
    \includegraphics[width=\linewidth]{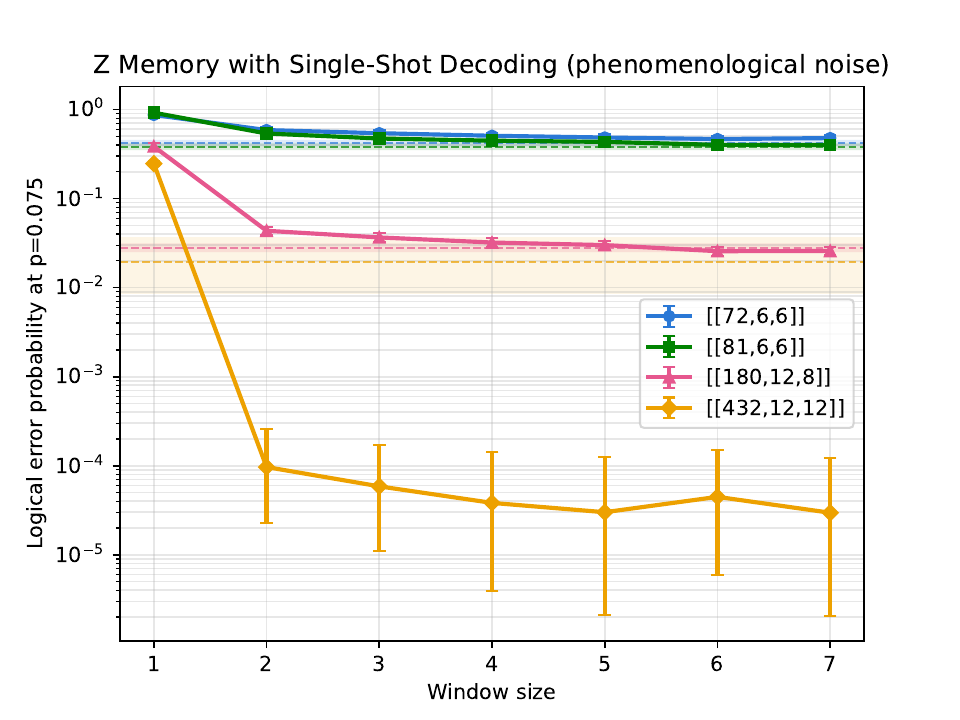}
    \caption{Total logical error probability for $Z$ memory vs window size used in the overlapping window strategy, for several TT codes, listed in \Cref{tab:codes_2}. We perform memory experiments under phenomenological noise at physical noise rate $p=0.075$, lasting $N=60$ rounds. We use BP+OSD0 within each window and choose a commit region size of $w-1$ for window size $w$. We perform 100 BP iterations, using a parallel schedule and the min-sum method, with a min-sum scaling factor of $0.625$. We include the results of decoding the entire syndrome history in a single batch, using BP+OSD0 (dashed horizontal lines, with shading indicating error bars). %
    }
    \label{fig:window_size_comp}
\end{figure}

We wish to investigate TT code performance numerically, and we begin by performing simulations under a simplified, ``phenomenological" noise model. We separately perform $Z$ and $X$ memory experiment simulations: we prepare all logical qubits in the $\ket{0}$ state (resp. the $\ket{+}$ state), perform stabilizer measurement for $r$ rounds, then measure out all data qubits in the $Z$ basis (resp. $X$ basis) to determine the state of the logical qubits at the end of the experiment. In between rounds, single-qubit depolarising noise channels are applied to data qubits with strength $p$ and, independently, each stabilizer measurement suffers a classical readout error with probability $p$. The final readout of data qubits is noiseless, as is standard in the literature~\cite{Higgott_2023,radebold2025explicitinstancesquantumtanner,tsubouchi2026degeneracycuttinglocalefficient}. Our simulations were carried out using the Python packages LDPC \cite{Roffe_LDPC_Python_tools_2022} and Stim \cite{Gidney_C_Stim}. 
All error bars in figures in this Section and elsewhere indicate regions in which the Bayes' factor is $\leq 1000$.

In Fig.~\ref{fig:phenom_data}(a), we present data for the $X$ memory experiments, examining the codes presented in Table~\ref{tab:mycodes}. We plot logical error probability (per round) for each code, without normalising for the number of logical qubits. We choose to run the memory experiment simulations for $r = 2d$ rounds, where $d$ is the code distance. Decoding is performed using a BP+OSD decoder (without any windowing)~\cite{roffe_decoding_2020,Roffe_LDPC_Python_tools_2022} and we fit a quadratic exponential to the data: $p_L = a\exp(-bp^2 -cp)$, for fitted parameters $a$, $b$, $c$. %

We perform an estimate of the threshold for the $X$ memory experiment -- the crossing-point of the three TT code curves in Fig.~\ref{fig:phenom_data}(a) -- which is at around $3.1\%$. We also determine pseudo-thresholds for the three codes -- the solution to the break-even equation $p_L (p) = p$, where $p$ is the physical error rate and $p_L$ is the per-round logical error rate in the $X$ memory experiment. (Note we define the pseudo-threshold differently here to in Section~\ref{sec:CLN_Sims}, where we define it as the solution to $p_L(p) = kp$. Under phenomenological noise, the solutions to this latter equation are far above threshold and hence have little meaning.) We find the following estimates for the pseudo-thresholds under this particular error channel. $[[72,6,6]]$ code: $2.0\%$; $[[180,12,8]]$ code: $2.7\%$; $[[432,12,12]]$ TT code: $2.9\%$.

We then test for single-shot decodability in the $Z$ memory experiment. We test for single-shot properties using an overlapping window method and a BP+OSD decoder~\cite{hillmann2024singleshot,scruby2024radial_codes,lin2025GB_codes,malcolm2025computingefficientlyqldpccodes,Skoric_2023,tan_2023_window} (see the LDPC Python package for an implementation~\cite{roffe_decoding_2020,Roffe_LDPC_Python_tools_2022}). The idea of an overlapping window method is that we perform decoding inside a ``window" of $w$ measurement rounds. The subset of the syndrome that was extracted in the rounds that fall within this window is passed to a BP+OSD decoder and a correction for that restricted syndrome is determined. We then consider a ``commit region" of the first $c$ rounds in the window. Any component of the correction in the window that lies in the commit region is committed to, while the component outside of this region is discarded. The window is then slid $c$ rounds into the future and decoding performed again, in the same way. We refer to this as a $(w,c)$ windowing strategy. Since we are decoding with restricted information (ignoring the check measurements outside of each window), this strategy does not in general perform as well as decoding with full syndrome information. In the surface code, for example, one needs to perform decoding over $O(d)$ rounds in order to maintain optimal performance~\cite{Skoric_2023}. The temporal overhead for decoding $d$ rounds of syndrome data can be considerable in practice and may scale super-linearly in the space-time volume of syndrome data~\cite{panteleev2021degenerate}. 
Single-shot decodable codes, meanwhile, tend to result in optimal performance when using an overlapping window decoder with a fixed window size much smaller than the code distance~\cite{hillmann2024singleshot,scruby2024radial_codes,lin2025GB_codes,malcolm2025computingefficientlyqldpccodes,Skoric_2023,tan_2023_window}. Asymptotically, they allow for a threshold to be maintained even for a window size of $w=1$~\cite{Brown_2016_SS_GCC}.

In Fig.~\ref{fig:phenom_data}(b), we present data for the $Z$ memory experiment simulations, examining codes from Table~\ref{tab:mycodes}, and employing a $(2,1)$ windowing strategy. Under this strategy, we observe a crossing point for $Z$ memory of $12\%$. The pseudo-thresholds of the three codes (for $Z$ memory) are: $[[72,6,6]]$ code: $9.8\%$; $[[180,12,8]]$ code: $10.6\%$; $[[432,12,12]]$ code: $10.8\%$. %

We also examine the $Z$ memory performance under different windowing strategies. In Fig.~\ref{fig:window_size_comp} we present results for a $(w,w-1)$ windowing strategy, where we vary $w$ between 1 and 7. We present results for a range of TT codes, fixing the total number of syndrome extraction rounds to be $N=60$ and the physical error rate to be $p=0.075$. We find that the logical error rates for $Z$ memory experiments for the TT codes rapidly plateau as we increase the window size, meaning that the performance of the decoder does not continue to increase as more syndrome information is taken into account during each round of decoding, indicating single-shot properties. %
We also include on the plot the results of decoding all syndrome data together (monolithic decode), illustrated as dashed lines. These theoretically represent the optimal performance from the decoder. However, we observe that the overlapping window decoder actually performs \textit{better} than the monolithic decode for the $[[432,12,12]]$ code. This may be due to the complexity of the decoding problem for that particular code. For the remaining codes, we observe that the decoder performance plateaus to the monolithic decode result at small window sizes.

\subsection{Performance of $X$ Logical Sector Under Overlapping Window Decoding}

\begin{figure}[t]
    \centering
    \subfloat[]{
    \includegraphics[width=0.99\linewidth]{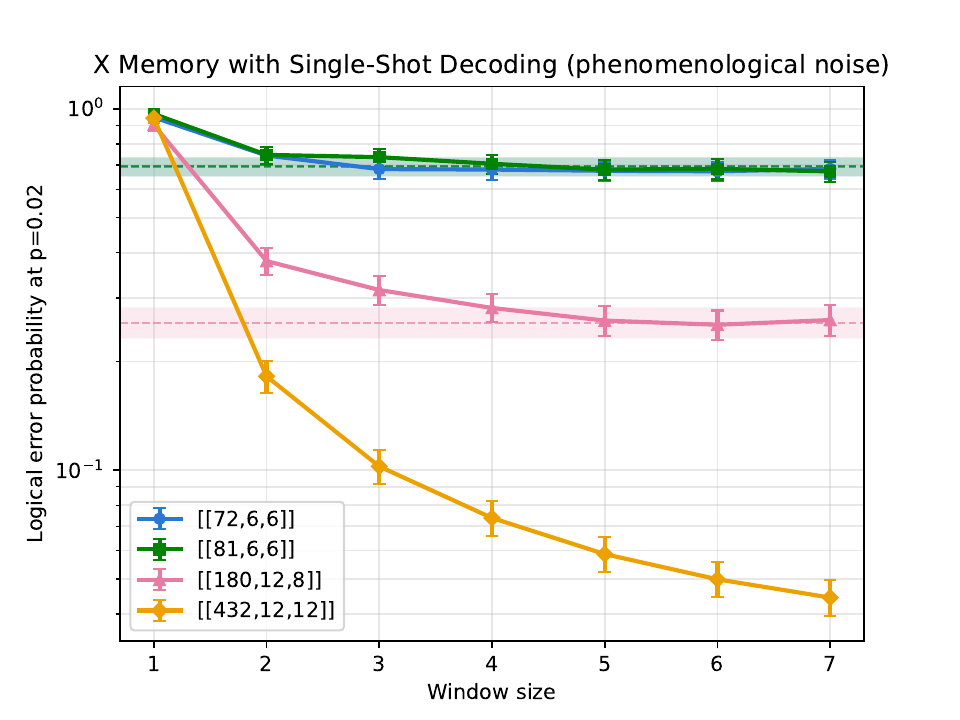}}\\
    \subfloat[]{\includegraphics[width=0.9\linewidth]{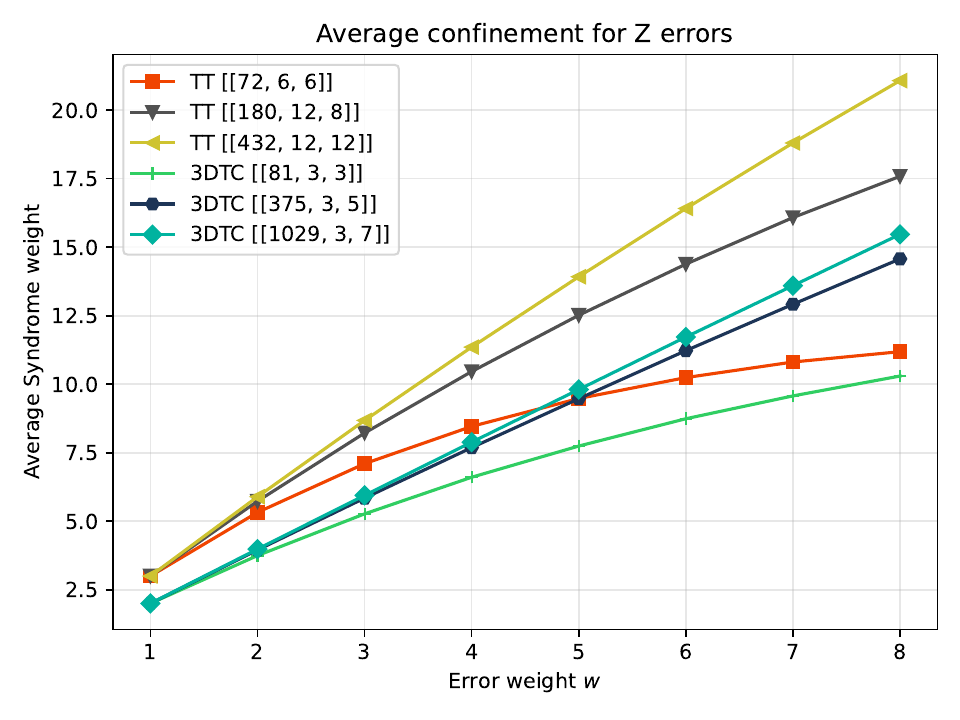}}
    \caption{(a) Decoding with an overlapping window decoder and BP+OSD of $X$ memory with TT codes. We use OSD-CS-30, 100 BP iterations and a min-sum scaling factor of 0.625. Experiments are run over $N=60$ error-correction rounds, and commit regions are set to size $w-1$, for window size $w$. (b) Average syndrome weight for errors of weight $w$, sampled over $2\times 10^4$ errors, for TT codes and 3D toric codes indicated.}
    \label{fig:X_single_shot}
\end{figure}

While the $Z$ logical sector of TT codes benefits from local metachecks, allowing for faulty syndrome information to be repaired, many other codes demonstrate single-shot decodability without the presence of this type of metacheck~\cite{lin2025GB_codes,scruby2024radial_codes,malcolm2025computingefficientlyqldpccodes}. To examine the properties of the TT codes' $X$-logical sector under a single-shot decoding strategy, we perform the same windowing analysis as performed for the $Z$-logical sector. We report results in Fig.~\ref{fig:X_single_shot}(a). As can be seen, small codes plateau rapidly to the logical error probability obtained from full-syndrome-history decoding as the window size is increased, but this behaviour differs for the $[[432,12,12]]$ code, which takes longer to plateau (the monolithic decode is not shown for this code and sits below the plot, at around $2.25\%$ with large error bars). This situation differs from the $Z$-logical sector, where all codes exhibit a plateau starting at a similar window size. However, we see that, for small codes at least, optimal results can be obtained when decoding \textit{both} bases for window sizes smaller than the distance of the code.

We observe strong effects resulting from the choice of BP+OSD decoder parameters, including when we use an overlapping window strategy. However, single-shot decoding is related to a decoder-independent property called confinement~\cite{quintavalle2021single}. Code families that exhibit good confinement have minimum syndrome weights that grow with the weights of the (reduced) error that produces them. These codes are single-shot against adversarial noise~\cite{quintavalle2021single}. We define the confinement profile of a code, which is the minimum syndrome weight over all errors with equal reduced weight. Roughly speaking, in codes with good confinement, errors cannot grow without triggering more syndromes. Under a stochastic noise model, we are more interested in the average syndrome weight incurred by errors of a given weight, which we may call the average confinement profile~\cite{scruby2024radial_codes} (see Appendix~\ref{app:confinement} for definitions).

To separate the effects of decoder choice from the properties of the code, we plot the average confinement profiles for TT and 3D toric codes in Fig.~\ref{fig:X_single_shot}(b). As can be seen, all codes tested demonstrate average confinement that grows with the error weight, even up past the distance of the code. We also find that the TT codes tend to have higher average confinement than the 3D toric codes, particularly at low error weights. This lends support to the claim that the TT codes will perform well under an overlapping window decoding strategy. In Appendix~\ref{app:confinement}, we additionally plot confinement profiles and average confinement for $X$ errors in 3D toric and TT codes.

Ultimately, we have demonstrated that small window sizes are sufficient in many cases to decode long memory experiments with equal fidelity to full-syndrome decoding, even in the unconfined logical sector. While we expect that the window size required to achieve the minimum logical error may grow with the code size, we nevertheless have demonstrated the utility of decoding with windowing to achieve a target logical error probability.

\section{Circuit Level Noise Simulations} \label{sec:CLN_Sims}

We conduct $X$ and $Z$ memory simulations, and verify the existence of $(w,c)$ single-shot decodability under a circuit level noise model for both the TT and 3D Toric codes. We utilise depth-13, interleaved syndrome measurement circuits detailed in Appendix~\ref{app:SM_circuits}. 
We define three sets of $n/3$ $Z$-measure qubits (for measuring $Z_a$, $Z_b$ and $Z_c$ checks), and for two of these, we alternate using them for $Z_b$ and $Z_c$ checks in adjacent rounds, again to minimise circuit-depth. We then use $n/3$ $X$-measure qubits. We omit resets on the $X$-measure qubits to minimise the circuit depth, however this likely has a minimal effect on memory performance~\cite{gehér2025resetresetquestion,harper2025characterisingfailuremechanismserrorcorrected}. Overall we use $4n/3$ measure qubits (for $n$ data qubits in the code). In Appendix~\ref{app:SM_circuits}, we prove that the circuits result in the required transformations of the stabilizer tableaux and discuss circuit optimisations we have explored. %

\begin{figure}[!t]
\subfloat[]{\includegraphics[width=\linewidth]{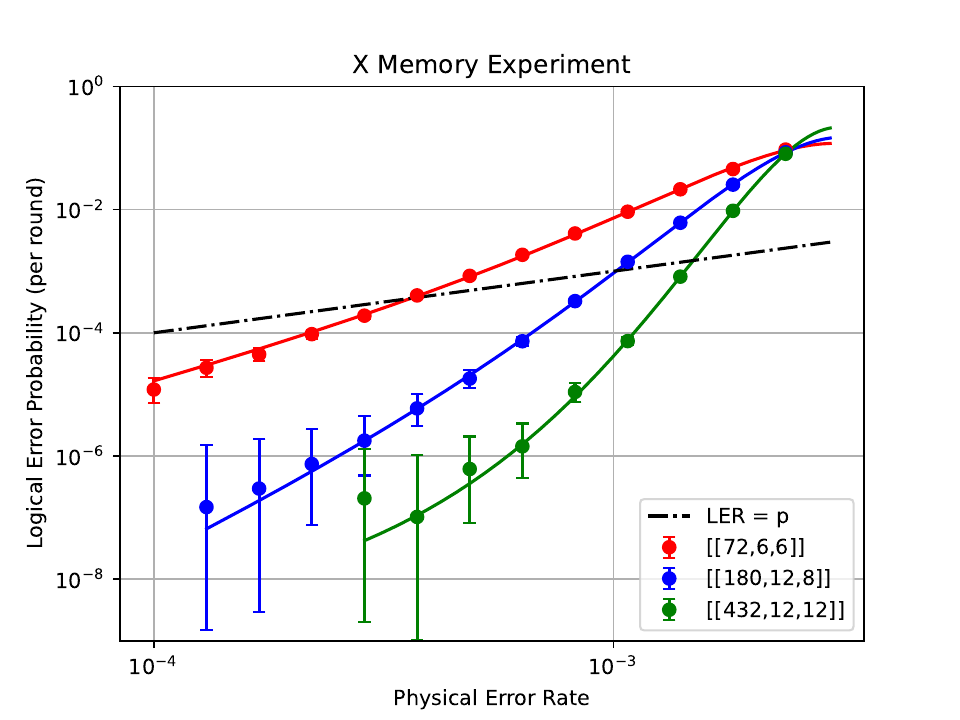}}\\
\subfloat[]{\includegraphics[width=\linewidth]{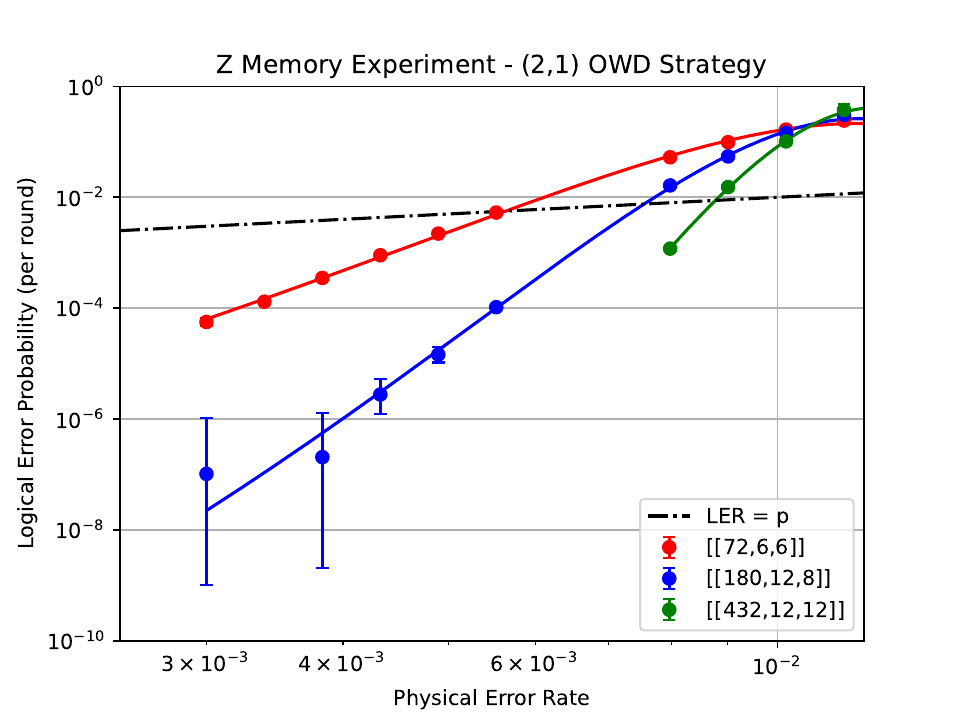}}
\caption{Performance of the TT codes under circuit level noise. (a) Logical error rate per round in $X$ memory experiments, decoded using 1000 iterations of min-sum BP with scaling factor 1.0 and with OSD-CS30. For $X$ memory, we do not use a windowed decoder and solve the decoding problem over the full syndrome history. A crossing between instances of TT codes is found at a physical error rate of $2.7 \times 10^{-3}$. (b) Logical error rate per round in $Z$ memory experiments. Decoding was performed using a $(2,1)$ overlapping window strategy (Sec. \ref{sec: single-shot}) using BP+OSD. We use 1000 iterations of BP, a min-sum scaling factor of 1.0, and OSD0. We find a crossing at $1.1 \times 10^{-2}$. For each plot, we show the LER = $p$ line (dashed). This shows when a code becomes effective as a whole. Below the dashed line, the code is protecting all $k$ logical qubits more effectively than an unencoded qubit. Note the slight difference to the pseudothreshold definition used in the text.}
\label{fig:decoding_data}
\end{figure}

We use a standard depolarizing, circuit-level noise model, where the total probability of error on a qubit is $p$ (this results in a uniform Pauli error probability $p/3$). All qubits undergo bit/phase flip errors after resets with probability $p$, and depolarising noise while idling. All qubits also undergo two-qubit (correlated) depolarising errors after each CNOT gate. Lastly, measure qubits undergo measurement flips (implemented via $X$ or $Z$ errors on measure qubits prior to measurement) with probability $p$. See Appendix~\ref{app:SM_circuits} for more details. The results are shown in Table \ref{Table: Toric_vs_TT_thresholds}. For toric code simulations, we use a modified version of the syndrome measurement circuit, omitting gates associated with $J_3$ terms, where $J \in \{A, B, C\}$, in the polynomials Eqns.(\ref{eqn: A poly}-\ref{eqn: C poly}). We omit noise channels on all idling layers in these circuits, retaining only gate errors, to obtain effectively CNOT-depth-6 circuits, optimal in depth for the 3D toric code. %

For memory experiments, we perform $r$ rounds of syndrome measurement, and compute the logical error rate per round ($p_L$) with the formula:
\begin{equation}
    p_L = 1 - (1 - P_L)^{\frac{1}{r}}
\end{equation}
where $P_L$ is the total number of logical errors over all shots. %
We fit trend lines using the \textit{ansatz}, $p_L = p^{d_{circ} / 2} e^{c_0 + c_1 p + c_2 p^2}$~\cite{bravyi2024high}, and numerically compute the coefficients using SciPy \cite{2020SciPy-NMeth}. The parameter $d_{circ}$ is the upper-bound on the circuit level distance obtained by finding the minimum weight error that triggers no detectors in the detector matrix \cite{derks2024designing}. The detector matrix is obtained using stimbposd \cite{Higgott_O_stimbposd} and the minimum weight error was found using the \texttt{codeDistance} package~\cite{webster2026distance}. 
Due to the asymmetry in the $X$ and $Z$ distances, the \textit{limiting} distance is the $Z$ distance. As such, in Table~\ref{Table: Toric_vs_TT_thresholds} we present crossing points for the $X$ memory experiments, which serve as an estimate for a threshold for the TT codes. In the following, we do not normalise for the number of logical qubits encoded, hence the error rates include correlated errors within a code patch.
We used a Belief-Propagation (BP) decoder with Ordered Statistics Decoding (OSD) post-processing as our decoder; deployed via the \texttt{ldpcV2} package~\cite{Roffe_LDPC_Python_tools_2022, roffe_decoding_2020}. In all simulations we use 1000 iterations of min-sum BP with a parallel schedule, and OSD in the combined-sweep (CS) configuration.
Each data point is sampled to a maximum of $7.5 \times 10^5$ shots with stim~\cite{Gidney_C_Stim} and plotted with error bars indicating binomial likelihood-ratio intervals computed using \texttt{sinter.fit\_binomial}. The bounds contain logical-error probabilities whose likelihood is within a factor of $10^3$ of the maximum-likelihood estimate and are subsequently transformed to per-round logical error rates. (We use BP+OSD as it is a standard benchmarking decoder for qLDPC codes, allowing direct comparison with existing and future results; we do not propose it as a real-time decoder for TT codes (see Section~\ref{sec:conclusion}).)

\subsection{Simulation Results}\label{subsec:simulation_results}
In $X$ memory experiments over $d_Z + 1$ rounds, we use BP+OSD-CS30 and a min-sum scaling factor of 1. As decoding an $X$ memory experiment constitutes a simpler decoding problem we employed a larger OSD order while remaining within a modest time budget. We observe a crossing between various small instances of TT codes at $2.7 \times 10^{-3}$. This is slightly higher than the crossing for the 3D toric code, which occurs at $2.6 \times 10^{-3}$. We also compute the pseudo-thresholds, the physical error rate below which the code suppresses more errors than an equivalent number of un-encoded qubits. This is the solution of the break-even equation $p_L(p) = kp$ \cite{bravyi2024high}, with $k$ the number of logical qubits, $p_L$ the block logical error rate and $p$ the physical error rate. We find that the pseudo-thresholds of TT codes are almost identical to those of 3DTCs. Our simulations suggest that the $[[180,12,8]]$ TT code has a logical error rate of approximately $10^{-7}$ at a physical error rate of $1.3 \times 10^{-4}$; meaning, it can preserve 12 logical qubits for over 1 million syndrome cycles, in the low noise regime. The [[72,6,6]] and [[180,12,8]] codes show an approximately constant log-log slope over the full range of error probabilities in the simulations. The data supports a power-law error suppression capability in TT codes, even in the low $p$ regime. However, the [[432,12,12]] code's log-log slope does not appear constant over the physical error range range. This may indicate the onset of error-floor-like behaviour in the code. However, the lowest $p$ estimates are based on only a few observed logical failures and consequently have broad uncertainties. These data points therefore do not resolve whether the apparent flattening represents a genuine change in the logical-error scaling or fluctuations arising from rare-event sampling. For the $[[432,12,12]]$ code, the above ansatz (see Sec.~\ref{sec:CLN_Sims}) overestimated the logical error suppression, a more accurate and conservative trend line was obtained by setting the prefactor to identity. Note that, if such an error floor is observed in these data, we expect that the effect will depend heavily on the decoder chosen, as more accurate decoders that mitigate the issues of BP trapping sets can drastically reduce the logical error rates of qLDPC codes and reduce or remove error floors seen in BP+OSD~\cite{gu2026scalableneuraldecoderspractical,beni2025tesseract,müller2025improvedbeliefpropagationsufficient}. 

For $Z$ memory experiments, which are limited by $d_X$, we evaluate the logical error rates over 13 rounds of syndrome measurement, with BP+OSD-CS0 and a min-sum scaling factor of 1. Note that $Z$-memory decoding is more complex, with higher-weight syndromes and BP failing more frequently. For this reason, we reduce the order of the OSD post-processing step to reduce runtime. For these simulations, we use a $(2,1)$ overlapping widow decoder to verify the existence of a crossing, while employing single-shot decoding. Under these circumstances we find a crossing at a physical error rate of $1.1 \times 10^{-2}$, and pseudo-thresholds at approximately 1\%. The error suppression in the $X$ error channel is thus far greater than in the $Z$ channel, even using windowing and a lower OSD order. %

\begin{figure}[h!]
    \centering
    \includegraphics[width=\linewidth]{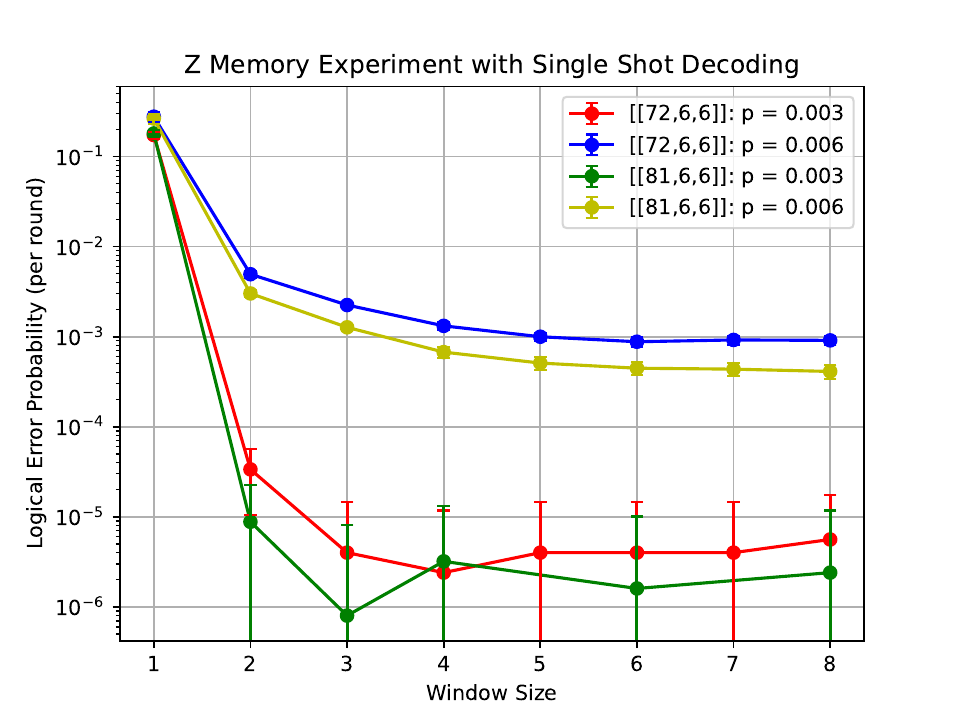}
    \caption{Logical error plateaus, for a commit size of 1 and varying window size, under a circuit level noise model for small TT codes, over 13 rounds, in the $X$ error channel. The $[[72,6,6]]$ and $[[81,6,6]]$ codes both exhibit single-shot decodability in the $X$ error channel, at error rates of $p=0.003$ and $p=0.006$. At $p=0.003$, the plateau begins around $w=3$, however at $p=0.006$ the plateau begins around $w=5$.}
    \label{fig:Z_CLN_window_size_comp}
\end{figure}

We perform $Z$ memory, single-shot experiments under a $(w,1)$ strategy for various values of $w$ and show that a logical error plateau is recovered even under a circuit level noise model. This suggests that TT codes retain single-shot capabilities in realistic settings. In Fig.\ref{fig:Z_CLN_window_size_comp} we explore sub-threshold noise regimes while varying the window size. We observe that the code performance rapidly plateaus as we increase the window size, albeit at a slower rate for larger $p$. That is, at a physical error rate of $p=6 \times 10^{-3}$ the LER plateau occurs at around $w=5$, whereas at a physical error rate of $p=3 \times 10^{-3}$ the LER plateau begins at around $w=3$. It is also interesting to note that the $[[72,6,6]]$ code has a higher logical error rate than the $[[81,6,6]]$ despite having the same $X$-distance and fewer qubits. We suspect this discrepancy is caused by structural differences in the Tanner graphs of the two codes. Of relevance could be the fact that the $[[81,6,6]]$ code has a toric layout, while the $[[72,6,6]]$ code does not. %
A thorough analysis is deferred to future work.

We also explore single-shot decoding in the $Z$ error channel. Fig.\ref{fig:X_CLN_window_size_comp} shows the change in LER as a function of the decoding window size in an $X$ memory experiment. Here we can see the LER plateaus at window sizes of 3-4 for all codes tested. %
These results support the usage of single-shot decoding in both sectors of the TT codes, which allows one to reduce the total decoding time by reducing the decoding volume passed to the decoder, per round.

We note that the presence of short cycles in TT code  Tanner graphs likely degrades the performance of the decoder in these experiments, affecting both $X$- and $Z$-memory experiments. The BP algorithm assumes a tree-like Tanner graph~\cite{yedidia2003understanding}, in which there are no cycles. However, quantum codes expressly violate this condition~\cite{babar2015fifteen}. The Z-checks in TT codes are high weight and have overlapping support, thereby increasing the number of short-cycles. This leads to BP failing frequently, and falling back to OSD to find a correction, increasing runtime. %
Although OSD improves decoding accuracy, its frequent use increases latency and reduces decoding throughput. We identify two possible approaches to mitigating these effects. First, one could alter the decoder. In particular, automorphism ensemble decoding, as presented in Ref.~\cite{koutsioumpas2025automorphism}, reduces the sensitivity of BP to short cycles by running multiple BP decoders in parallel on automorphically transformed representations. The most likely valid correction among the converged ensemble members is then selected. Other BP variants specifically designed to mitigate short-cycles  exist, such as BP with memory~\cite{kuo2022exploiting, müller2025improvedbeliefpropagationsufficient} and BP with decimation (freezing the most confident bits)~\cite{yao2024belief}. These methods reduce the need to fall back to OSD entirely. Alternatively, one can replace OSD with localised statistics decoding or ambiguity clustering~\cite{Hillmann_2025,wolanski2025ambiguityclusteringaccurateefficient}; one can use one of the many BP variants, such as RelayBP~\cite{müller2025improvedbeliefpropagationsufficient} or VibeLSD~\cite{koutsioumpas2025colourcodesreachsurface}; or one can combine multiple of these into a hierarchical decoder~\cite{basak2026hierarchicalquantumdecoders,zhao2026ultrahighratequantumerrorcorrection,bhardwaj2026highrateqldpcprocessors}.

Second, the overcomplete parity-check matrix could be replaced by a decoder-specific, row-space-equivalent representation with fewer checks, lower density, and fewer short cycles, this may decrease latency and increase BP accuracy. In classical coding, pruning-based neural BP has been used to remove unimportant check nodes from overcomplete parity-check matrices and substantially reduce the number of check-node evaluations~\cite{buchberger2021pruning}. Sparse, full-rank parity-check representations have also been constructed specifically to remove length-four cycles and provide graph connectivity favourable to BP~\cite{shen2025toward}. Related quantum work has applied dynamic check-node removal to mitigate trapping-set failures in QLDPC decoding~\cite{bhattacharyya2025decoding}. Provided that the retained checks span the original stabilizer row space, such a reduction preserves both the code space and the commutation relations. These approaches may improve the decoding throughput of TT codes and potentially their accuracy, but their resource requirements, and tail latency have not been evaluated for the codes considered here. 

\begin{figure}
    \centering
    \includegraphics[width=\linewidth]{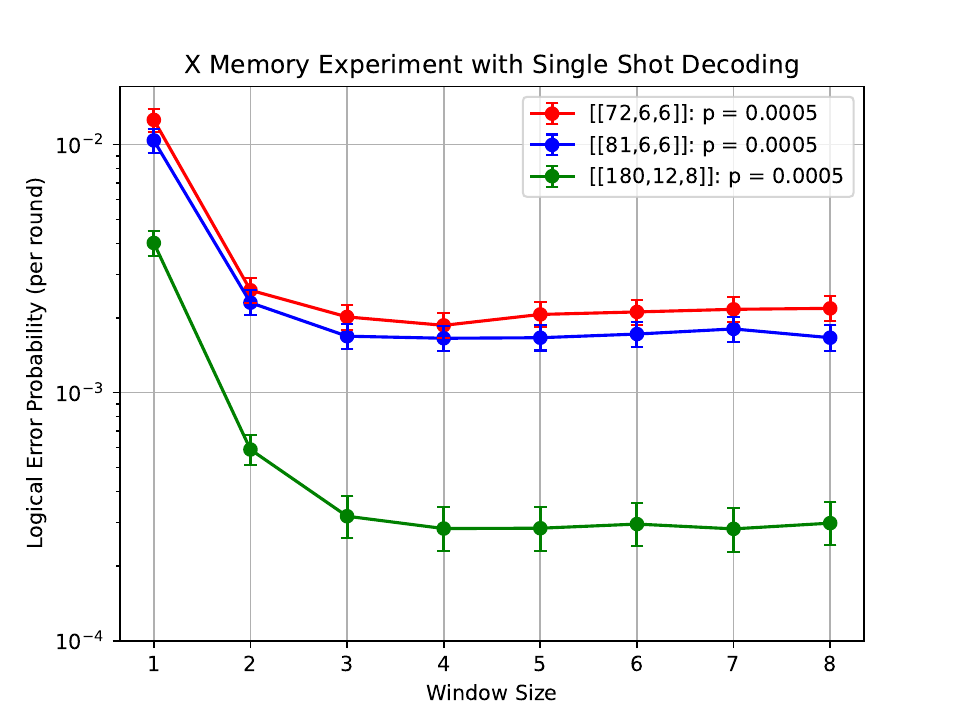}
    \caption{Logical error plateaus, for a commit size of 1 and varying window size, under a circuit level noise model for small TT codes, over $d+1$ rounds in the $Z$ error channel. The TT codes both exhibit \textit{practical} single-shot decodability, at error rates of $p=0.0005$. This singe-shot behaviour is not scalable, as the optimal window size will likely grow with the distance for larger codes. %
    }
    \label{fig:X_CLN_window_size_comp}
\end{figure}

\begin{table*}
\centering
\def\arraystretch{1.3}
\setlength{\tabcolsep}{14pt}  %
\begin{tabular}{|c|c|c|c|c|}
\hline
\textbf{Code Family} & 
\begin{tabular}{@{}c@{}} \textbf{Crossing} \\ \textbf{Point} \end{tabular} & 
$\mathbf{{[[n,k,d]]}}$ & 
\begin{tabular}{@{}c@{}} \textbf{Circuit} \\ \textbf{Distance} \end{tabular} & 
\begin{tabular}{@{}c@{}} \textbf{Pseudo-threshold} \\ $\mathbf{(\times 10 ^{-3})}$ \end{tabular} \\
\hline
\hline
\multirow{3}{*}{3D Toric} & \multirow{3}{*}{$2.6 \times 10^{-3}$} & $[[81,3,3]]$ & $\leq 3$ & $1.1$ \\ \cline{3-5}
 &                        & $[[375,3,5]]$ & $\leq 5$ & $2.0$ \\ \cline{3-5}
 &                        & $[[1029,3,7]]$ & $\leq 7$ & $2.2$ \\ \cline{3-5}
\cline{1-2}
\cline{1-5}
\multirow{3}{*}{Trivariate Tricycle} & \multirow{3}{*}{$2.7\times 10^{-3}$} & $[[72,6,6]]$ & $\leq 4$ & 0.9 \\ \cline{3-5}
 &                        & $[[180,12,8]]$ & $\leq 7$ & 1.7 \\ \cline{3-5}
 &                        & $[[432,12,12]]$ & $\leq 10 $ & 2.0 \\ 
\hline
\end{tabular}
\caption{Crossing, circuit distance and pseudo-thresholds for 3D Toric and TT codes, in $X$ memory experiments (the basis which limits code performance). The results show that the two codes are very close in terms of performance. The pseudo-threshold is obtained by solving $p_L(p) = kp$. The circuit distances are those computed by the method discussed in Sec.\ref{sec:CLN_Sims}.}
\label{Table: Toric_vs_TT_thresholds}
\end{table*}

\section{Logic Gates in Trivariate Tricycle Codes}\label{sec:logic_gates}

In this Section, we explore logic gates that can be applied within TT codes via automorphisms or transversal/constant-depth circuits of $CZ$/$CCZ$ gates. We begin by noting the structure of the logical qubits in TT codes (this applies to other 3-block codes as well). We will find that the logical qubits can fall into one of three sets, with logical Pauli operators for each set mutually commuting between the three sets. 

Let $X(P,Q,S)$, for polynomials $P$, $Q$ and $S$, represent an $X$ operator with support on qubits in the L block corresponding to all monomial terms in $P$, and similarly for $Q$ in the C block and $S$ in the R block (recall that qubits in each block can be labelled by monomials in $\mathcal{M}$). We similarly define $Z(\bar{P}, \bar{Q}, \bar{S})$ as a $Z$ operator with support on $\bar{P}$ qubits in the L block, $\bar{Q}$ qubits in the C block, and $\bar{S}$ qubits in the R block. We have the following: %
\begin{lem}\label{lem:operator_commutation}
    $X(P, Q, S)$ and $Z(\bar{P}, \bar{Q}, \bar{S})$ anti-commute if and only if $ 1\in P\bar{P}^\top + Q\bar{Q}^\top + S\bar{S}^\top$.
\end{lem}
The proof is analogous with a similar one in Ref.~\cite{bravyi2024high}, and we provide it in Appendix~\ref{app:TT_code_tables} for completeness.

Now suppose we can find polynomials $f$, $g$, $h$ and $j$ that obey:
\begin{align}
    &Af = Bf = Cf = 0 \label{eqn:f_poly_def}\\
    &Ch + Bj = Cg + Aj = Bg + Ah = 0.\label{eqn:hjg_poly_def}
\end{align}
Note that this implies that $f\in \ker A \cap \ker B \cap \ker C$ and $[g^\top,h^\top,j^\top]^\top \in \ker[0 \; C\; B]\cap \ker[C \; 0\; A] \cap  \ker[B \; A\; 0]$. According to Lemma~\ref{lem:Num_logical_qubits}, these sets cannot both be trivial (since $k>0$).

Then the following are all logical operators, where $\alpha\in \mathcal{M}$:
\begin{align}\label{eqn:logical_block_1}
    &\bar{X}_\alpha^{(1)} = X(\alpha f, 0, 0),\; \bar{Z}_\alpha^{(1)} = Z(\alpha h^\top, \alpha g^\top, 0),\nonumber\\
    &\bar{X}_\alpha^{(2)} = X(\alpha g, \alpha h, \alpha j),\; \bar{Z}_\alpha^{(2)} = Z(0, \alpha f^\top , 0), \\
    &\bar{X}_\alpha^{(3)} = X(0, 0, \alpha f),\; \bar{Z}^{(3)}_\alpha = Z(0, \alpha j^\top, \alpha h^\top).\nonumber
\end{align}
We note that these logical Pauli operators  belong to one of three sets, which we label with superscripts (and which are akin to the primed and unprimed blocks in Ref.~\cite{bravyi2024high}). To see that these are logical operators, note that $X(P,Q,S)$ commutes with all $Z$ stabilizers whenever $CQ + BS = CP + AS = BP + AQ = 0$ (due to Lemma~\ref{lem:operator_commutation}). Hence $X^{(1)}_\alpha$ and $X^{(3)}_\alpha$ clearly commute with all $Z$ stabilizers, because of Eqn.~\ref{eqn:f_poly_def}, while $X^{(2)}_\alpha$ does so, from Eqn.~\ref{eqn:hjg_poly_def}. $Z(P^\top, Q^\top, S^\top)$ commutes with all $X$ stabilizers whenever $AP + BQ + CS = 0$. Hence, it can be seen that $Z^{(i)}_\alpha$ commutes with all $X$ stabilizers, for $i=1,2,3$. 

It can be seen that all $\bar{X}^{(i)}_\alpha$ commute with all $\bar{Z}^{(j\neq i)}_\beta$. For example, $\bar{X}^{(2)}_\alpha$ and $\bar{Z}^{(1)}_\beta$ commute because $\alpha \beta^\top gh + \alpha \beta^\top hg + \alpha j 0^\top = 0 $. %
Note that the particular choice of logical groupings into the three sets is not unique. And indeed, the logical operators $\bar{X}_\alpha^{(i)}$ and $\bar{Z}_\beta^{(j)}$ will only be able to anti-commute if $fh\neq 0$. This is because $\bar{X}^{(i)}_\alpha$ and $\bar{Z}^{(i)}_\beta$ form an anti-commuting pair whenever $\alpha\beta^\top \in fh$. If $fh=0$, we will not be able to find any conjugate logical operators within the sets chosen above. In such cases, we can find two more choices for the logical sets: one which assumes $fg\neq 0$ and another which assumes $fj\neq 0$. If $fg\neq 0$, we have:
\begin{align}\label{eqn:logical_block_2}
    &\widetilde{X}_\alpha^{(1)} = X(0, \alpha f, 0),\; \widetilde{Z}_\alpha^{(1)} = Z(\alpha h^\top, \alpha g^\top, 0),\nonumber\\
    &\widetilde{X}_\alpha^{(2)} = X(\alpha g, \alpha h, \alpha j),\; \widetilde{Z}_\alpha^{(2)} = Z(\alpha f^\top , 0, 0),\\
    &\widetilde{X}_\alpha^{(3)} = X(0, 0, \alpha f),\; \widetilde{Z}^{(3)}_\alpha = Z(\alpha j^\top, 0, \alpha g^\top),\nonumber
\end{align}
while if $fj\neq 0$ we have:
\begin{align}\label{eqn:logical_block_3}
    &\widehat{X}_\alpha^{(1)} = X(\alpha f, 0, 0),\; \widehat{Z}_\alpha^{(1)} = Z(\alpha j^\top, 0, \alpha g^\top),\nonumber\\
    &\widehat{X}_\alpha^{(2)} = X(\alpha g, \alpha h, \alpha j),\; \widehat{Z}_\alpha^{(2)} = Z(0 , 0, \alpha f^\top),\\
    &\widehat{X}_\alpha^{(3)} = X(0, \alpha f, 0),\; \widehat{Z}^{(3)}_\alpha = Z(0, \alpha j^\top, \alpha h^\top).\nonumber
\end{align}

Not all codes we find admit a complete decomposition into one of the choices of logical sets given above. All codes we investigate have some choice of $f$ and $[g^\top,h^\top,j^\top]^\top$ that obey Equations~\ref{eqn:f_poly_def} and \ref{eqn:hjg_poly_def} respectively, but it may be the case that none of the corresponding logical sets (Equations~\ref{eqn:logical_block_1}, \ref{eqn:logical_block_2} or \ref{eqn:logical_block_3}) cover all logical qubits in the code. In Table~\ref{tab:codes_2} of  Appendix~\ref{app:TT_code_tables}, we report whether each code found admits a decomposition of a full set of $k$ logicals into one of the logical sets above (Equations~\ref{eqn:logical_block_1}, \ref{eqn:logical_block_2} or \ref{eqn:logical_block_3}).

\subsection{Shift Automorphisms}

Owing to the translational invariance of the TT codes, any shift automorphism that maps L, C, and R qubits $\alpha \mapsto \beta\alpha$ is a valid logical operation. This is a ``shift" automorphism because it shifts the entire lattice (Fig.~\ref{fig:checks}), mapping each qubit $\alpha$ in each block to a new location, $\beta \alpha$. Consider an $X$-check corresponding to monomial $\alpha\in \mathcal{M}$. This can be written as $X(A\alpha, B\alpha, C\alpha)$, since it acts non-trivially on those data qubits in the L block connected to it by terms in $A$, and similarly for data qubits in the C and R blocks, with polynomials $B$ and $C$, respectively. The shift automorphism transforms this $X$ stabilizer generator as:
\begin{align}
    X(A\alpha,B\alpha,C\alpha) \mapsto X(A\beta\alpha, B\beta\alpha, C\beta\alpha).
\end{align}
Similarly, observing the rows of $H_Z$, we may see how $Z_a$-, $Z_b$-, and $Z_c$-stabilizer generators are written. These are transformed by the shift automorphism as:
\begin{align}
    Z_a(0,C^\top\alpha,B^\top\alpha) &\mapsto Z_a(0, C^\top\beta\alpha,B^\top\beta\alpha),\nonumber\\
    Z_b(C^\top\alpha,0,A^\top\alpha) &\mapsto Z_b(C^\top\beta\alpha,0,A^\top\beta\alpha),\\
    Z_c(B^\top\alpha, A^\top\alpha,0) &\mapsto Z_c(B^\top\beta\alpha, A^\top\beta\alpha, 0).\nonumber
\end{align}
Notice that the automorphism simply permutes checks within a class ($X$, $Z_a$, $Z_b$ or $Z_c$), and hence preserves the code space.

If we can find a decomposition of the logical operators into one of the three logical sets introduced above (Equations~\ref{eqn:logical_block_1}, \ref{eqn:logical_block_2} or \ref{eqn:logical_block_3}), then the shift automorphisms can be understood as transforming between the logical operators within a set. For example, the $\alpha\mapsto \beta\alpha$ automorphism maps $\bar{X}_\alpha^{(i)} \mapsto \bar{X}_{\beta\alpha}^{(i)}$ and $\bar{Z}_\alpha^{(i)} \mapsto \bar{Z}_{\beta\alpha}^{(i)}$ for $i=1,2,3$. 
However, these automorphisms cannot map between operators in different sets.

To implement such shift automorphisms via a low-depth circuit, we can perform SWAP gates between data qubits and $Z$-measure qubits in the following pattern:
\begin{align}
    &q(\text{L},\alpha) \leftrightarrow q_Z(a, \beta_1\alpha) \leftrightarrow q(\text{L},\beta_2\beta_1\alpha)\\
    &q(\text{C},\alpha) \leftrightarrow q_Z(b,\beta_1\alpha) \leftrightarrow q(\text{C},\beta_2\beta_1\alpha)\\
    &q(\text{R},\alpha) \leftrightarrow q_Z(c,\beta_1\alpha) \leftrightarrow q(\text{R},\beta_2\beta_1\alpha),
\end{align}
where we label data qubits $q(\text{B},\alpha)$ for block label B and monomial label $\alpha$, and we label $Z$-measure qubits $q_Z(b,\alpha)$, where $b$ is the $Z$-measure block label. 
This series of SWAP gates permutes $Z$-measure qubits within each block ($a$, $b$ and $c$) and performs a shift of the data qubits $\alpha\mapsto \beta_2\beta_1\alpha$ within each of the data qubit blocks (L, C and R).

If we wish to avoid adding additional connections between qubits, we must alter this pattern of SWAP gates slightly. Firstly, note that, from the necessary structure of the syndrome measurement circuits, we can apply gates between $q(\text{L},\alpha)$ and $q_X(A_i^\top \alpha)$ (for any $i$) and similarly between $q(\text{C},\alpha)$ and $q_X(B_i^\top \alpha)$ and $q(\text{R},\alpha)$ and $q_X(C_i^\top \alpha)$, where we once again label $X$-measure qubits with only a monomial label. We may also perform gates between $q(\text{L},\alpha)$ and $q_Z(b,C_i\alpha)$ or $q_Z(c,B_i\alpha)$ and so on. Suppose we wish to perform a shift automorphism of the data qubits given by $\alpha\mapsto A_j A_i^\top \alpha$ (for all data qubit blocks). We can perform the following depth-2 circuit of SWAP gates using existing connections to implement this:
\begin{align}
    &q(\text{L},\alpha) \leftrightarrow q_X(A_i^\top\alpha) \leftrightarrow q(\text{L},A_jA_i^\top \alpha)\\
    &q(\text{C},\alpha) \leftrightarrow q_Z(c,A_j \alpha) \leftrightarrow q(\text{C},A_i^\top A_j \alpha)\\
    &q(\text{R},\alpha) \leftrightarrow q_Z(b, A_j\alpha) \leftrightarrow q(\text{R},A_i^\top A_j \alpha).
\end{align}
One can easily find a similar, depth-2 pattern of SWAP gates for shifts by $B_jB_i^\top$ and $C_jC_i^\top$. If we assume the Tanner graph is connected, these shift automorphisms generate the entire group of shifts $\alpha\mapsto \beta$.

\subsection{Logical $CZ$ gates between code blocks}

We now focus on logical gates that can be implemented \textit{between} identical code blocks. Recent qLDPC codes have been found that allow for certain Clifford gates to be performed within or between code blocks \cite{malcolm2025computingefficientlyqldpccodes}. For BB codes, Ref.~\cite{breuckmann2024cupsgatesicohomology} explains methods for verifying that a code possesses $CZ$ gates between code blocks via the cup product (see below for an introduction), while in Ref. \cite{eberhardt2024logicaloperatorsfoldtransversalgates}, the authors present fold-transversal gates that exist in BB codes. In this section, we introduce an approach to finding $CZ$ gates in multi-block codes, including BB and TT codes, that is distinct from these previous approaches.

Consider the following depth-1 circuit of $CZ$ gates, where $\beta\in\mathcal{M}$ is arbitrary:
\begin{align}
\begin{split}
    \overline{CZ}_{\text{CR},\beta} = \prod_{\alpha\in\mathcal{M}} CZ[ q(\text{C},\alpha) ,q(\text{R},\beta \alpha^{\top})] \\
    CZ[q(\text{R},\alpha), q(\text{C},\beta \alpha^\top)],
\end{split}
\end{align}
where the first (second) argument of $CZ[\cdot,\cdot]$ refers to a data qubit in the first (second) code block. Similarly to above, $q(\text{C},\alpha)$ refers to the C data qubit labelled by monomial $\alpha$. Let us consider how an $X$ stabilizer from the first code block is transformed by this circuit:
\begin{align}
\begin{split}
    \lbrace X(\alpha A, \alpha B, \alpha C)\rbrace_1 \mapsto \lbrace X(\alpha A, \alpha B, \alpha C)\rbrace_1 \\
    \lbrace Z(0, \beta \alpha^\top C^\top, \beta\alpha^\top B^\top)\rbrace_2,
\end{split}
\end{align}
where we use curly brackets and a subscript to indicate the code block on which the Pauli operator acts. The $Z$-type operator is a stabilizer of the second code block: it corresponds to the $\beta\alpha^\top$-labelled $Z_a$-check. The $X$ stabilizers of the second code block are similarly transformed. Hence, the above transversal gate preserves the code space. We have two other sets of transversal gates formed in the same way:
\begin{align}
    \overline{CZ}_{\text{LR},\beta} &= \prod_\alpha CZ[ q(\text{L},\alpha) ,q(\text{R},\beta\alpha^{\top})] \nonumber \\
    & \qquad \qquad CZ[q(\text{R},\alpha), q(\text{L},\beta\alpha^\top)]\\
    \overline{CZ}_{\text{LC},\beta} &= \prod_\alpha CZ[ q(\text{L},\alpha) ,q(\text{C},\beta\alpha^{\top})] \nonumber\\
    &\qquad \qquad CZ[q(\text{C},\alpha), q(\text{L},\beta\alpha^\top)].
\end{align}
We note that gates of the form $\overline{CZ}_{\text{LR},\beta}$ are similarly transversal logical gates of BB codes.

If we can find a decomposition of the logicals into three commuting sets (Equations~\ref{eqn:logical_block_1}--\ref{eqn:logical_block_3}), then the actions of two of these gates on the logical operators becomes particularly simple. Consider, for example, the $X$-logicals from Equation~\ref{eqn:logical_block_1} and the gates $\overline{CZ}_{\text{LC},\beta}$ and $\overline{CZ}_{\text{CR},\beta}$:
\begin{align}
    \overline{CZ}_{\text{LC},\beta}:\lbrace \bar{X}^{(1)}_\alpha\rbrace_1 &\mapsto \lbrace \bar{X}^{(1)}_\alpha\rbrace_1 \lbrace \bar{Z}^{(2)}_{\beta\alpha^\top}\rbrace_2\\
    \lbrace \bar{X}^{(2)}_\alpha\rbrace_1 &\mapsto \lbrace \bar{X}^{(2)}_\alpha\rbrace_1 \lbrace \bar{Z}^{(1)}_{\beta\alpha^\top}\rbrace_2\nonumber\\
    \lbrace \bar{X}^{(3)}_\alpha\rbrace_1 &\mapsto \lbrace \bar{X}^{(3)}_\alpha\rbrace_1;\nonumber\\
    \overline{CZ}_{\text{CR},\beta}: \lbrace \bar{X}_\alpha^{(3)}\rbrace_1 &\mapsto \lbrace \bar{X}_\alpha^{(3)}\rbrace_1 \lbrace \bar{Z}_{\beta \alpha^\top}^{(2)}\rbrace_2\nonumber \\
    \lbrace \bar{X}_\alpha^{(2)}\rbrace_1 &\mapsto \lbrace \bar{X}_\alpha^{(2)}\rbrace_1 \lbrace \bar{Z}_{\beta \alpha^\top}^{(3)}\rbrace_2\\
    \lbrace \bar{X}^{(1)}_\alpha\rbrace_1 &\mapsto \lbrace \bar{X}^{(1)}_\alpha\rbrace_1.\nonumber
\end{align}
The transformation of the $X$-logicals from block two is equivalent, owing to the symmetry of the $CZ$ gates. Hence, these gates act as logical $\overline{CZ}$ gates on certain logical qubits. For example, $\overline{CZ}_{\text{LC},\beta}$ is a logical $\overline{CZ}$ between qubit $\alpha$ in the first set and qubit $\beta\alpha^\top$ in the second set (for all $\alpha$) in both blocks of code, while acting as the identity on all qubits in the third set. The third type of $\overline{CZ}$ gate does not have such a simple representation as the above, generically performing some circuit of logical $CZ$ gates between all logical qubits. We note that TT codes also possess transversal CNOT gates owing to the fact that they are CSS codes. We also note that we numerically observe these $CZ$ gates to have non-trivial logical action even in the case in which the logical qubits cannot be fully decomposed into one of the sets from Equations~\ref{eqn:logical_block_1}--\ref{eqn:logical_block_3}.

\subsection{Constant-depth $CCZ$ circuits}\label{sec:CCZ_gates}

\begin{table*}
\begin{center}
\def\arraystretch{1.5}
\[\begin{array}{|c|c|c|c|c|}
\hline
    \begin{array}{c}\textbf{Code} \\ \textbf{parameters}\end{array} & \mathbf{\frac{n(d)}{n_{\text{TC}}(d)}} & \mathbf{\ell, m, p} & \mathbf{d_Z, d_X} & \textbf{Polynomials}\\ %
    \hline \hline
    \mathbf{[[21, 3, 3]]} & 7/27 & 7, 1, 1 & 3, 6 & \begin{array}{c}
         A = 1 + x, \;
         B = 1 + x^5,\;
         C = 1 + x^4
    \end{array}\\
    \hline
    [[36, 3, 3]] & 4/9& 3, 2, 2 & 3, 8 & \begin{array}{c}
         A = 1 + xyz, \;
         B = 1 + x^2z,\;
         C = 1 + x^2y
    \end{array}\\% & 6\\
    \hline
    \mathbf{[[36, 3, 4]]} & 3/16 & 12, 1, 1 & 4, 8 & \begin{array}{c}
         A = 1 + x^3, \;
         B = 1 + x^{11},\;
         C = 1 + x^7
    \end{array}\\
    \hline
    [[48, 3, 4]] & 1/4 & 4, 2, 2 & 4, 8 & \begin{array}{c}
         A = 1 + x, \;
         B = 1 + xz, \;
         C = 1 + xy
    \end{array}\\% & 6\\
    \hline
    [[54, 3, 4]] & 9/32 & 3, 3, 2 & 4, 9 & \begin{array}{c}
         A = 1 + yz,\;
         B = 1 + xz,\;
         C = 1 + xyz
    \end{array}\\% & 6\\
    \hline
    [[60, 3, 4]] & 5/16 & 5, 2, 2 & 4, 12 & \begin{array}{c}
         A = 1 + xz,\;
         B = 1 + xy,\;
         C = 1 + xyz
    \end{array}\\% & 6\\
    \hline
    \mathbf{[[81, 3, 5]]} & 27/125 & 9, 3, 1 & 5, 15 & \begin{array}{c}
         A = 1 + x, \;
         B = 1 + x^{8}y,\;
         C = 1 + x^5 y^2
    \end{array}\\
    \hline
    [[90, 3, 5]] & 6/25 & 5, 3, 2 & 5, 15 & \begin{array}{c}
         A = 1 + x,\;
         B = 1 + xy,\;
         C = 1 + x^2y^2z
    \end{array}\\% & 6\\
    \hline
    \mathbf{[[114, 3, 6]]} & \sim 17.6\% & 38, 1, 1 & 6, 19 & \begin{array}{c}
         A = 1 + x, \;
         B = 1 + x^{31},\;
         C = 1 + x^{27}
    \end{array}\\
    \hline
    \mathbf{[[210, 3, 7]]} & \sim 20.4\% & 70,1,1 & 7, \leq 28 & \begin{array}{c}
         A = 1 + x, \;
         B = 1 + x^{16},\;
         C = 1 + x^{25}
    \end{array}\\
    \hline
\end{array}\]
\end{center}
\caption{Examples of $(2,2,2)$ TT codes with $CCZ$ gates that have non-trivial logical action. In the second column we give the ratio of $n$ for the TT code with distance $d$ to that of the 3D toric code with equivalent distance. Other columns are equivalent to those from Table~\ref{tab:mycodes}. We bold the best code parameters at each distance $d$ between $3$ and $7$.\label{tab:222_CCZ_codes}}
\end{table*}

\begin{figure}
    \centering
    \includegraphics[width=\linewidth]{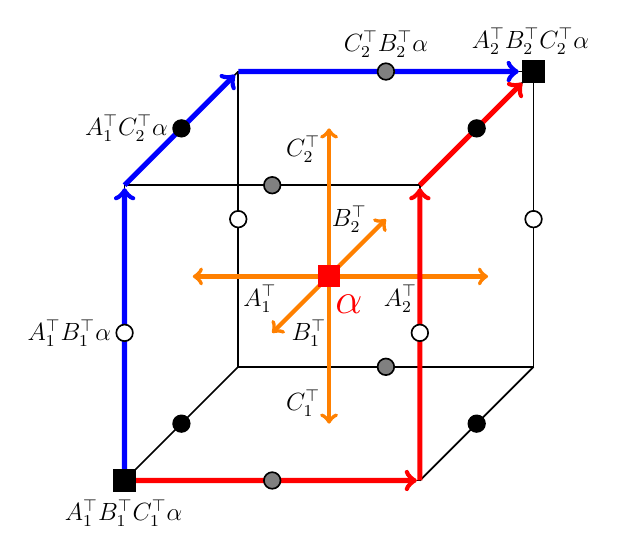}
    \caption{Prescription for applying $CCZ$ gates within each cube (labelled by meta-check $\alpha\in\mathcal{M}$) in order to implement a logical non-Clifford gate in a TT code with polynomials $A = A_1 + A_2$, $B= B_1 + B_2$ and $C = C_1 + C_2$. Two example paths that correspond to $CCZ$ gates are shown in blue and red.}
    \label{fig:CCZ_cube}
\end{figure}

We finally investigate the possibility of fault-tolerant non-Clifford gates in TT codes. Fault-tolerant logical CCZ gates have been implemented on certain QLDPC codes using a construction derived from the cohomological cup product ~\cite{zhu2024noncliffordparallelizablefaulttolerantlogical,Chen2023,Wang2024,breuckmann2024cupsgatesicohomology}. This technique has recently been applied to hypergraph product codes and homological products of LDPC codes, showing favourable asymptotic scaling \cite{zhu2025topological, zhu2025transversal}.

We show below that we can find families of TT codes with specific check structures, which admit a logical $\overline{CCZ}$ gates which can be enacted with a finite-depth circuit. This follows from the ability to define a cup product on these codes. We begin by introducing the following lemma, showing that certain families of polynomials allow us to define cup products on the resulting codes:
\begin{lem}\label{lem:cup_product}
    A TT code with all three polynomials ($A$, $B$, and $C$) either weight-2 or of the form $P = \sum_{k=1}^{\text{ord}(g)}g^kf$, where $g\in\mathcal{M}$ is such that $\text{ord}(g)$ is even, and $f$ is an arbitrary (non-zero) polynomial, admits a cup product.
\end{lem}
\noindent
We prove this Lemma in Appendix~\ref{app:CCZ_gates}. 

We begin this section by focusing on polynomials of weight-2; i.e., $A=A_1 + A_2$, $B=B_1+B_2$, $C=C_1+C_2$. The resulting codes are similar to 3D toric codes but with certain twisted boundary conditions. The codes, however, still admit a finite-depth implementation of a $\overline{CCZ}$ gate. Crucially, this gate is implementable via a depth-2 physical CCZ circuit. We call such a code a (2,2,2) TT code, with the numbers referring to the weights of the polynomials A, B and C, respectively. %

To explain the circuit that implements the logical $\overline{CCZ}$ gate, consider a meta-check labelled by monomial $\alpha$ (recall the meta-check labelling introduced in Section~\ref{sec:3D_layout}). In the 3D toric code and (2,2,2) TT codes, this corresponds to a cube in the lattice (note this lattice may have highly irregular boundary conditions). A logical gate can be achieved by a constant-depth circuit of $CCZ$ gates illustrated in Fig.~\ref{fig:CCZ_cube} (for all cubes in the lattice).
For the cube shown, we first draw an oriented path from $X$-check $A_1^\top B_1^\top C_1^\top \alpha$ to $X$-check $A_2^\top B_2^\top C_2^\top \alpha$ (two examples of such paths are shown in the figure). A CCZ gate is applied for each arrow between the three qubits it intersects, with the first such qubit being in the first code block, the second being in the second block and the third being in the third. For example, the blue path in Fig.~\ref{fig:CCZ_cube} corresponds to the $CCZ$ gate:
\begin{align}
    CCZ[q(\text{R}, A_1^\top B_1^\top \alpha), q(\text{C}, A_1^\top C_2^\top \alpha), q(\text{L}, C_2^\top B_2^\top \alpha)].
\end{align}

In \Cref{tab:222_CCZ_codes}, we enumerate (2,2,2) codes found by performing searches over low dimension sizes, $\ell$, $m$ and $p$. Each of these codes admits a logical $\overline{CCZ}$ circuit. All codes found have $k=3$, which can be understood from the fact that they are equivalent to 3DTCs on a lattice with twisted boundary conditions (see Ref.~\cite{voss2025multivariatebicyclecodes} for an analysis of BB codes on 2D lattices with similarly twisted boundaries). In Table~\ref{tab:222_CCZ_codes}, we compare the qubit overhead of the corresponding code found to that of the 3DTC with equivalent distance. We also numerically verify that, for all codes found, the $CCZ$ circuit from Fig.~\ref{fig:CCZ_cube} implements a logical transformation of the same form as for the 3DTC. Let us label the three logical qubits of a (2,2,2) TT code $a$, $b$ and $c$ respectively. In the following, we will let $\overline{CCZ}(a,b,c)$ refer to a logical $CCZ$ gate acting on logical qubit $a$ from the first code block, $b$ from the second, and $c$ from the third. We find that the logical action of the $CCZ$ circuit is then given by:
\begin{align}
\begin{split}
    \overline{CCZ}(a,b,c)\overline{CCZ}(a,c,b)\overline{CCZ}(b,a,c)\\
    \overline{CCZ}(b,c,a)\overline{CCZ}(c,a,b)\overline{CCZ}(c,b,a).
\end{split}
\end{align}
That is, for all codes, the $CCZ$ circuit implements 6 logical $\overline{CCZ}$ gates. %

We note that the $X$ and $Z$ distances for the codes found are more balanced than they are for the 3DTC, with $d_X < d_Z^2$. %
We numerically verified that the $CCZ$ circuit preserves the stabilizer groups of the codes found (see Appendix~\ref{app:CCZ_gates}). The resulting codes improve on the overheads of the 3D toric code dramatically, using up to $5.7\times$ fewer qubits for an equivalent-distance encoding. Interestingly, we observe that the best codes at each size (bolded in Table~\ref{tab:222_CCZ_codes}) arise from univariate polynomials.

Lemma~\ref{lem:cup_product} allows us to find polynomials of larger weights that can be used to define higher-rate codes that admit logical $\overline{CCZ}$ circuits. For example, polynomials of the form $(1+g)(1+\alpha)$ with $g,\alpha\in\mathcal{M}$ and $g^2=1$, admit a cup product, from Lemma~\ref{lem:cup_product}. We can choose $A$ to be of this form, while $B$ and $C$ can be chosen to be weight-2, to keep the overall check weights relatively low. %
We refer to these as (4,2,2) codes. The resulting codes have weight-8 $X$-checks and weight-6 or weight-4 $Z$-checks. The $CCZ$ circuit that preserves the stabilizer for such codes is given in Appendix~\ref{app:CCZ_gates}.

However, while these higher-weight polynomial codes can be found to have large rates, the $Z$-distances of all such codes are upper-bounded by $2$. This ultimately arises from the factorisation of the polynomials producing low-weight classical codewords that lift to logicals of the TT code~\cite{tiew2026copycupgatestensorproducts} (see Appendix~\ref{app:inv_free_ansatz}). We therefore introduce further ans\"{a}tze that build on Lemma~\ref{lem:cup_product} that can achieve codes of high rate and distances $d>2$. %

\begin{table*}
    \begin{center}
    \def\arraystretch{1.5}
    \[\begin{array}{|c|c|c|c|c|c|c|c|}
    \hline
        \begin{array}{c}\textbf{Code}\\
        \textbf{parameters}\end{array} & \textbf{Rate} & \mathbf{d_Z} & \mathbf{d_X} & \ell, m, p & \mathbf{A} & \mathbf{B} & \mathbf{C}\\
    \hline
    \hline
    [[36, 3, 4]] & 1/12 & 4 & 6 & 12, 1, 1 & \begin{array}{c}A_\text{in} = \langle x^3\rangle\\A_\text{out} = x \langle x^3\rangle\\A_\text{free} = x^2(1+x^3)\end{array} & 1+x & 1+x^5\\
    \hline
    [[60, 12, 3]] & 1/5 & 3 & 5 & 20, 1, 1 & \begin{array}{c}A_\text{in} = \langle x^5\rangle\\A_\text{out} = x \langle x^5\rangle\\A_\text{free} = x^3(1+x^{15}) + x^{14}(1+x^5)\end{array} & 1+x^8 & 1+x^{16}\\
    \hline
    [[96, 24, 4]] & 1/4 & 4 & 4 & 4, 4, 2 & \begin{array}{l}A_\text{in} = (1+z)(1+y)\\A_\text{out} = x \langle y\rangle\end{array} & \begin{array}{c}B_\text{in} = (1+x)(1+x^2y^2)\\B_\text{out} = y \langle x\rangle\end{array} & 1+y\\
    \hline
    [[96,12,4]] & 1/8 & 4 & 8 & 4,4,2 &  \begin{array}{c} A_\text{in} = (1+z)(1+y) \\
    A_\text{out} = x\langle y\rangle 
    \end{array} & 1+x & 1 + xy^2z\\ 
    \hline
    [[144, 24, 3]] & 1/6 & 3 & 6 & 4, 4, 3 & \begin{array}{c}A_\text{in} = (1+y)(1+x^2)\\A_\text{out} = z \langle y\rangle\end{array} & 1+z & 1+x^2y^2z\\
    \hline
    [[192, 3, 8]] & 1/64 & 8 & 8 & 8, 8, 1 & \begin{array}{c} A_\text{in} = \langle y\rangle\\A_\text{out} = x \langle y\rangle\\A_\text{free} = x^2(1+y)\end{array} & 1+y & 1+x\\
    \hline
    [[432, 72, 3]] & 1/6 & 3 & 6 & 4, 4, 9 & \begin{array}{c}A_\text{in} = (1+y)(1+x^2)\\A_\text{out} = z^3 \langle y\rangle\end{array} & 1+z^3 & 1+x^2y^2z^6\\
    \hline
    \end{array}\]
    \end{center}
\caption{Examples of higher-rate TT codes that admit a $CCZ$ gate via the cup product formalism. The $CCZ$ gate in each case acts non-trivially on the logical space. We write $\langle g\rangle \coloneqq \sum_{k=1}^{\text{ord}\, g} g^{k}$ for the sum over the cyclic subgroup generated by $g$. We also partition each polynomial of weight $>2$ into in, out and free sets (see Appendix~\ref{app:CCZ_gates}), which defines the physical $CCZ$ circuit.\label{tab:high_rate_CCZ_codes}}
\end{table*}

We present several found high-rate codes with distances $d>2$ in Table~\ref{tab:high_rate_CCZ_codes}. For all codes, we test the logical action of the $CCZ$ circuit, which is always a collection of $\overline{CCZ}$ logical gates between three code blocks~\cite{breuckmann2024cupsgatesicohomology}. %
We check that the logical action of these gates is non-trivial, that is, that it applies at least one non-trivial logical $\overline{CCZ}$ gate, as explained in Appendix~\ref{app:CCZ_gates}. The highest-rate resulting code obtains rate $1/4$ and $d_Z = d_X = 4$. We find other codes of distances up to $8$.

We note that these higher-rate constructions are not strictly fault-tolerant, since the code distances cannot be increased without increasing the check-weights of the codes. This is explained in Appendix~\ref{app:inv_free_ansatz}. Additionally, the check weights for these finite-sized instances are very high. For example, the $[[96,12,4]]$ code has maximum check weight $12$, while the $[[192,3,8]]$ code has maximum check weight $22$. The codes do, however, demonstrate dramatically higher encoding rates than the codes of Table~\ref{tab:222_CCZ_codes}. We leave further improvements to these codes to future work.

For the distance-2 codes found (see Appendix~\ref{app:CCZ_gates}), we comment that it is often the case that the full set of 
minimum-weight logicals for the codes often includes several logicals that are higher weight than 2. This allows us to generate higher-distance codes without $CCZ$ gates but with additional inter-block Clifford gates via gauge fixing certain logicals, as explained in Appendix~\ref{app:CCZ_gates}.

\section{Conclusion}\label{sec:conclusion}

In this work, we have investigated a new code construction, the trivariate tricycle codes, which combine several properties useful for future fault-tolerant quantum computation: high thresholds under circuit-level noise, single-shot decodability and fault-tolerant logical gates. The construction includes codes with encoding rates and relative distances that outperform topological codes based on 2D- or 3D-local connectivity. By way of example, our $[[432, 12, 12]]$ code outperforms the 3D toric code encoding 12 logical qubits at $d_Z = 12$ by a factor of 48.  Furthermore, the $X$- and $Z$-distances of many codes found are naturally more balanced for the TT codes than they are for the 3DTC, meaning we can obtain the benefits of a 3D code without such a large distance asymmetry. While the code parameters do not reach those attained by bivariate bicycle codes, we note that the inherent distance asymmetry means the $X$-distances outperform those of other qLDPC codes (for similar $n$ and $k$). While the distance asymmetry could be eliminated by distance balancing techniques~\cite{hastings2016weightreductionquantumcodes}, it could indeed be useful for code platforms that have a biased noise profile~\cite{Ruiz_2025}, making either $Z$ or $X$ errors more common. However, preserving the noise bias at the circuit level generally requires bias-preserving CNOT gates~\cite{Puri_2020,ye2026biaspreservingcatcatcnotgate}. It may be difficult to achieve bias-preserving gates on a hardware that also allows for the connectivity requirements of the TT codes. Currently, the distance asymmetry of the TT codes, while it tends to be smaller than the corresponding asymmetry in 3D toric codes, may limit the utility of TT codes in some platforms. Future work could look into distance balancing or adaptations of the TT codes to different hardware platforms with biased noise.

Beyond noise bias, the practicality of TT codes on a given platform will depend on its ability to provide the long-range connectivity the codes require. TT codes are most naturally suited to platforms with reconfigurable connectivity, such as neutral atom arrays~\cite{Bluvstein_2023,Xu_2024,Wang2024,zhao2026ultrahighratequantumerrorcorrection,bhardwaj2026highrateqldpcprocessors} and trapped ions~\cite{ye2025quantumerrorcorrectionlong}. There, non-local connectivity is achieved by physically moving qubits, so the cost of long-range connections appears mainly as extra syndrome-extraction cycle time. Each CNOT layer of our syndrome extraction circuits (Appendix~\ref{app:SM_circuits}) involves monomial shifts, uniform permutations of entire blocks of qubits. These are well-suited to parallel, block-wise qubit transport. Future work could develop hardware-specific layouts for TT codes on these platforms.

The TT codes possess meta-checks in one basis, with a meta-check distance equal to $d_Z$, and good confinement properties under $X$ error channels. This allows for partial single-shot decodability, which we demonstrate using circuit-level noise simulations and an overlapping window decoding strategy. The $Z$ memory performance of the codes is particularly high, with pseudo-thresholds around $1\%$ under circuit-level noise. For small-to-medium sized TT codes, we see that we can equally use an overlapping window strategy for $X$ memory decoding while achieving optimal results with small window sizes. We demonstrate the high thresholds and pseudo-thresholds of our codes using the syndrome extraction circuit of Appendix~\ref{app:SM_circuits}. The high performance is quite surprising given the large depth of this circuit. We note that this circuit was found analytically but we have numerically explored circuit optimisations associated with commuting CNOT order permutations. It is possible, however, that further improvements could be made by attempting to maximise the circuit-level distance over a larger search space.  %

While the ability to single-shot decode (with an overlapping window strategy) helps to overcome the large time overheads of decoding, since it allows for only decoding over a constant number of rounds of syndrome data, the in-window decoder must still be chosen so that the time taken to decode a window's worth of syndrome data is less than the time taken to generate said data. Otherwise a backlog problem is encountered~\cite{Skoric_2023,Terhal_quantum_memories}. It is unlikely that BP+OSD will be sufficiently fast for real-time decoding, even if overlapping window decoding is used. This is because the OSD sub-routine scales as $O(N^3)$, with $N$ the number of fault locations in the window~\cite{panteleev2021degenerate}. There are a large number of directions one could look when trying to speed up decoding: One can replace the OSD post-processing step with localised statistics decoding or ambiguity clustering, which offer better average runtimes~\cite{Hillmann_2025,wolanski2025ambiguityclusteringaccurateefficient}; one can use one of the many BP variants, such as RelayBP~\cite{müller2025improvedbeliefpropagationsufficient,maurer2025realtimedecodinggrosscode}, automorphism ensemble decoders~\cite{koutsioumpas2025automorphism} or VibeLSD~\cite{koutsioumpas2025colourcodesreachsurface}; or one can combine several of these techniques into a hierarchical decoder~\cite{basak2026hierarchicalquantumdecoders,zhao2026ultrahighratequantumerrorcorrection,bhardwaj2026highrateqldpcprocessors} (see also Ref.~\cite{Maan_2026} for real-time \textit{correlated} decoding). The requirements of the decoder will be highly dependent on the hardware used as well as where the TT codes sit within a broader architecture. However, it is likely that low-latency, high-accuracy decoders will be necessary to make use of TT codes within any architecture. Future work could integrate these recent improvements in qLDPC code decoding with an overlapping window strategy to provide suitable options for real-time decoding with TT codes.  

Furthermore, future work could investigate the possibility of performing correlated decoding~\cite{Maan_2026} within overlapping windows. This would require decoding $X$ syndrome information within windows as well as $Z$ syndrome information. However, we have shown that overlapping window decoding can be optimal for \textit{both} $X$ and $Z$ memory in small TT codes, and hence this strategy could be promising for correlated decoding within a window. Note also that correlated decoding is required when performing transversal logic gates between code blocks, such as is discussed in Section~\ref{sec:logic_gates}. Future work should consider how recent advances in correlated decoding across transversal logical operations~\cite{Cain_PRL_correlated_decoding,Zhou_2025,turner2025scalabledecodingprotocolsfast,Sahay_2025} can be incorporated into an architecture that includes TT codes.

We investigated several proposals for fault-tolerant Clifford and non-Clifford gates in TT codes. We introduced shift automorphisms and transversal $CZ$ gates between code blocks, identifying a logical basis in which many of these gates have a simple logical action. And we identified several polynomial constructions that allowed for a constant-depth $CCZ$ circuit to be defined which preserves the code space. %

We presented several codes in \Cref{tab:222_CCZ_codes} with logical $\overline{CCZ}$ gates (over three code blocks) implementable by constant-depth circuits. These codes outperform the regular 3DTC with the same $Z$-distance by up to a factor of $5.7$, and have lower asymmetry than the 3DTC between $d_Z$ and $d_X$. Meanwhile, in Table~\ref{tab:high_rate_CCZ_codes}, we present several codes with logical $\overline{CCZ}$'s that have encoding rates of up to $1/4$. This comes at the cost of relatively small distances and high check weights. Future work could search for improved ans\"{a}tze that circumvent these limitations.

The TT codes found have rates and relative distances that are less favourable than several bivariate bicycle codes, such as the $[[144,12,12]]$ gross code~\cite{bravyi2024high}. As such, the positive aspects of the TT code construction --- the single-shot decodability, the presence of metachecks, and the fault-tolerant logical gates --- need to be weighed against the potential trade-off of qubit overhead. That being said, despite this larger overhead compared to BB codes, TT codes may still be highly useful when combined with other codes into a fault-tolerant architecture. For example, future work could examine whether these codes could assist in single-shot code state preparation~\cite{hong2024singleshotpreparationhypergraphproduct}, using a dimension-jump scheme~\cite{li2026transversaldimensionjumpproduct} to prepare the code states of two-block group algebra codes with reduced time overhead (extending the hypergraph product code dimension-jump construction of Ref.~\onlinecite{hong2024singleshotpreparationhypergraphproduct}). TT codes may also be useful within magic state factories~\cite{Menon_2026}, as a part of an architecture with code-switching-based fault-tolerant logic~\cite{li2026transversaldimensionjumpproduct,tan2025singleshotuniversalityquantumldpc}, or as a transient merged code within a surgery procedure~\cite{berent2026liftedsurgeryfastprocessing}. We leave further potential applications of TT codes to future work.

\textit{Note added:} While this work was being completed, related work appeared in Ref.~\cite{lin2025abelianmulticyclecodessingleshot}. There the authors present multi-block constructions for group algebra codes, which include a construction similar to equations~\ref{eqn:PCM_MZ} -- \ref{eqn:PCM_X}. They do not, however, study the trivariate tricycle codes, or their logic gates, presented here.

\section{Acknowledgements}

The authors would like to thank Nikolas Breuckmann, Riley Chien, Oscar Higgott, Timo Hillmann, Rebecca Radebold, Tom Scruby, Mackenzie Shaw, and Mark Webster for enlightening discussions. AJ was supported by the Engineering and Physical Sciences Research Council [grant number EP/S021582/1]. CM and DEB were supported by the Engineering and Physical Sciences Research Council [grant number EP/Y004620/1]. CM acknowledges support from the Intelligence Advanced Research Projects Activity (IARPA), under the Entangled Logical Qubits program through Cooperative Agreement Number W911NF-23-2-0223.

\appendix

\section{TT code examples, properties, and logical sets}\label{app:TT_code_tables}

In this Appendix, we present more TT code examples, along with their properties. In \Cref{tab:codes_2}, we present several codes obtained by performing a random search through weight-3 polynomials, setting the first term of each to 1, without loss of generality (see Lemma~\ref{lem:code_equivalence}). Along with the code parameters, dimensions, and polynomials, we state whether there exists one of the logical sets from Equations~\ref{eqn:logical_block_1}--\ref{eqn:logical_block_3} that covers all logical qubits in the code. This is based on an exhaustive search through the three logical sets, and through all possible combinations of polynomials $f$, and $[g,h,j]$, and monomials $\alpha\in\mathcal{M}$ that feature in each of the logical sets. In the case that we cannot find a logical set that covers all logical qubits, we state the maximum number of logical Pauli $\bar{X}$ and $\bar{Z}$ operators that can be covered by a single logical set (i.e., a single choice of $f$ and $[g,h,j]$, and for all choices of $\alpha\in\mathcal{M}$). Note also that, in the cases in which we cannot find a single set that covers all logicals, we still observe that we can cover all logicals by varying $f$ and $[g,h,j]$ in the construction of the sets. That is, for each logical Pauli $\bar{X}$ or $\bar{Z}$, there exists a choice of $f$, $g$, $h$, $j$, and $\alpha$ such that the logical operator is of one of the forms in Equations~\ref{eqn:logical_block_1}--\ref{eqn:logical_block_3}.

\begin{table*}
    \begin{center}
        \def\arraystretch{1.3}
        \[\begin{array}{|c|c|c|c|c|c|c|}
        \hline
        \mathbf{[[n, k, d]]}  & \mathbf{\ell, m, p} & \mathbf{d_X} & \mathbf{A} & \mathbf{B} & \mathbf{C} & \begin{array}{c} \textbf{Complete} \\ \textbf{logical sets?}\end{array} \\
        \hline \hline 
        [[36,6,4]] &  3,2,2 & 8 & 1+ x + x^2z & 1 + xy  + x^2y & 1 + xyz + x^2 & \begin{array}{c} \text{\ding{55}}\\ \text{(max. logicals 6)}\end{array} \\
        \hline
        [[72,6,6]] &  4,3,2 & 12 & 1+ y + xy^2 & 1 + yz  + x^2y^2 & 1 + xy^2z + x^2y & \begin{array}{c} \text{\ding{55}}\\ \text{(max. logicals 6)}\end{array}\\
        \hline
        [[81,6,6]] &  3,3,3 & 12 & 1+ x + xy & 1 + y  + yz & 1 + z + x & \text{\ding{51}}\\
        \hline
        [[126,6,8]] &  7,3,2 & \leq 22 & 1+ y^2 + x^4y & 1 + xy  + x^4y^2z & 1 + x^2yz + x^2y^2 & \text{\ding{51}} \\
        \hline
        [[135,12,6]] &  5,3,3 & 14 & 1+ x + x^3y^2 & 1 + xz^2  + x^2yz & 1 + xy^2z + x^2yz & \text{\ding{51}}\\
        \hline
        [[180,12,8]] &  5,4,3 & \leq 20 & 1+ x^2y^3z + x^4y & 1 + x^3  + x^4z^2 & 1 + x^3y^3 + x^4yz^2 & \text{\ding{51}}\\ %
        \hline
        [[288,6,10]] &  8,4,3 & \leq 45 & 1+ yz + x^7yz^2 & 1 + x^2yz  + x^5z^2 & 1 + x^2y^3z + x^6y^2z^2 & \begin{array}{c} \text{\ding{55}} \\ \text{(max. logicals 6)}\end{array}\\
        \hline
        [[432,12,12]] &  6,6,4 & \leq 36 & 1+ xyz^3 + x^3y^4z^2 & 1 + x^3yz^2  + x^3y^2z^3 & 1 + x^4y^3z^3 + x^5z^2 & \begin{array}{c} \text{\ding{55}} \\ \text{(max. logicals 12)}\end{array}\\
        \hline
        [[588,9,\leq 14]] &  7,7,4 & \leq 73 & 1+ yz + x^3y^2z & 1 + xy^4z  + x^4y^4z^2 & 1 + x^2y^2z + x^2y^4 & \text{\ding{51}} \\
        \hline
        [[648,6,\leq 16]] &  6,6,6 & \leq 72 & 1+ xz^4 + x^3yz^4 & 1 + xy^4z  + x^5yz^5 & 1 + x^2yz^2 + x^2y^2z^3 & \text{\ding{51}} \\
        \hline
        [[648,12,\leq 14]] &  6,6,6 & \leq 72 & 1+ x + x^4z & 1 + y  +xyz^4 & 1 + z + x^3y^2z & \begin{array}{c} \text{\ding{55}} \\ \text{(max. logicals 12)}\end{array} \\
        \hline
        [[735,18,\leq 14]] &  7,7,5 & \leq 80 & 1+ y^2 + x^6y^5 & 1 + x^2yz  + x^2y^4z^3 & 1 + x^4y^5z^2 + x^5y & \text{\ding{51}} \\
        \hline
        [[840,9, \leq 16]] &  8,7,5 & ?? & 1+ y^3z + xyz^2 & 1 + y^5z  + x^3y^4z & 1 + xy^3z^4 + x^7yz & \text{\ding{51}} \\
        \hline
        [[1029,18, \leq 16]] &  7,7,7 & ?? & 1+ y^2z^3 + x^5y^5z & 1 + x^2y^5z  + x^4z^5 & 1 + x^4y^6 + x^5y^4z^6 & \text{\ding{51}} \\
        \hline
        \end{array}\]
    \end{center}
    \caption{Example trivariate tricycle codes and their properties. First column: parameters $[[n,k,d]]$; second column: dimensions of $x$, $y$ and $z$, respectively; third column: $X$-distance found; fourth - sixth columns: polynomials defining the code; seventh column: we report whether one can find a logical set (see Section~\ref{sec:logic_gates}) that covers all logical qubits in the code. Upper bounds for distances are found using three different methods. Distm4ri heuristic method~\cite{pryadko2024distm4ri} involved $10^6$ steps. \label{tab:codes_2}}
\end{table*}

We also state (without proofs, since they are analogous to those found in Ref.~\cite{bravyi2024high}) Lemmas relating to the connectivity of the Tanner graphs and 3D-toric layout (see Section~\ref{sec:3D_layout}) of TT codes:
\begin{lem}\label{lem:connectivity_tanner_graph}
A TT code based on matrices $A$, $B$ and $C$ has a connected Tanner graph if $\lbrace A_i A_j^\top | i,j \in \lbrace 1,2,3\rbrace \rbrace \cup \lbrace B_i B_j^\top | i,j \in \lbrace 1,2,3\rbrace \rbrace \cup \lbrace C_i C_j^\top | i,j \in \lbrace 1,2,3\rbrace \rbrace$ generate $\mathcal{M}$.
\end{lem}
\begin{lem}\label{lem:3D_toric_layout}
    A TT code based on matrices $A$, $B$ and $C$, and dimensions $\ell, m, p$, has a 3D toric layout if there exists $i,j,g,h,k,o\in \{1,2,3\}$ such that:

    (i) $\langle A_i A_j^\top , B_g B_h^\top, C_k C_o^\top\rangle = \mathcal{M}$ and

    (ii) $\text{ord}(A_iA_j^\top)\text{ord}(B_gB_h^\top) \text{ord}(C_kC_o^\top) = \ell m p$.
\end{lem}
All the TT codes we consider in this paper have connected Tanner graphs but most do not have a 3D toric layout. Some exceptions include the $[[81,6,6]]$ and $[[648, 12, 14]]$ codes from Table~\ref{tab:codes_2}.

We finally prove Lemma~\ref{lem:operator_commutation}, important for the construction of the logical sets, which we repeat here:
\setcounter{lem}{3}
\begin{lem}
    $X(P, Q, S)$ and $Z(\bar{P}, \bar{Q}, \bar{S})$ anti-commute if and only if $ 1\in P\bar{P}^\top + Q\bar{Q}^\top + S\bar{S}^\top$.
\end{lem}
\begin{proof} 
One can write:
\begin{align}
    P=\sum_{\alpha \in \mathcal{M}} p_\alpha \alpha, \quad \bar{P}=\sum_{\alpha \in \mathcal{M}} \bar{p}_\alpha \alpha, \nonumber \\
    \quad Q=\sum_{\alpha \in \mathcal{M}} q_\alpha \alpha, \quad \bar{Q}=\sum_{\alpha \in \mathcal{M}} \bar{q}_\alpha \alpha,\\ \quad S=\sum_{\alpha \in \mathcal{M}} s_\alpha \alpha, \quad \bar{S}=\sum_{\alpha \in \mathcal{M}} \bar{s}_\alpha \alpha, \nonumber
\end{align}
where $p_\alpha, \bar{p}_\alpha, q_\alpha, \bar{q}_\alpha, s_\alpha, \bar{s}_\alpha, \in \mathbb{F}_2$ are coefficients. Pauli operators $Z(\bar{P},\bar{Q},\bar{S})$ and $X(P,Q,S)$
overlap on a qubit $q(L, \alpha)$ iff $p_\alpha \bar{p}_\alpha = 1$, $q(C, \alpha)$ iff $q_\alpha \bar{q}_\alpha = 1$ and $q(R, \alpha)$ iff $s_\alpha \bar{s}_\alpha = 1$. 
This means:
\begin{align}
    P \bar{P}^\top=\sum_{\alpha \in \mathcal{M}} p_\alpha \bar{p}_\alpha 1+\ldots \text {,} \nonumber \\
    Q \bar{Q}^\top=\sum_{\alpha \in \mathcal{M}} q_\alpha \bar{q}_\alpha 1 +\ldots \text {,} \\
    S\bar{S}^\top=\sum_{\alpha \in \mathcal{M}} s_\alpha \bar{s}_\alpha 1 + \ldots \text {,} \nonumber
\end{align}
By the matrix structure of the monomials the transpose operation is equal to the inverse, and the existence of 1 indicates operators are acting on the same qubit. The dots represent all monomials different from 1, i.e. non-overlap qubits. Then by linearity, 
\begin{equation}
    P\bar{P} + Q\bar{Q} + S\bar{S} = \sum_{\alpha \in \mathcal{M}} (p_\alpha \bar{p}_\alpha + q_\alpha \bar{q}_\alpha + s_\alpha\bar{s}_\alpha)1 + ...
\end{equation}
Thus proving the statement. 
\end{proof}

\section{Additional Simulation Results}

\subsection{Comparisons between TT and 3D toric codes}

\begin{figure}
    \centering
    \subfloat[]{\includegraphics[width=\linewidth]{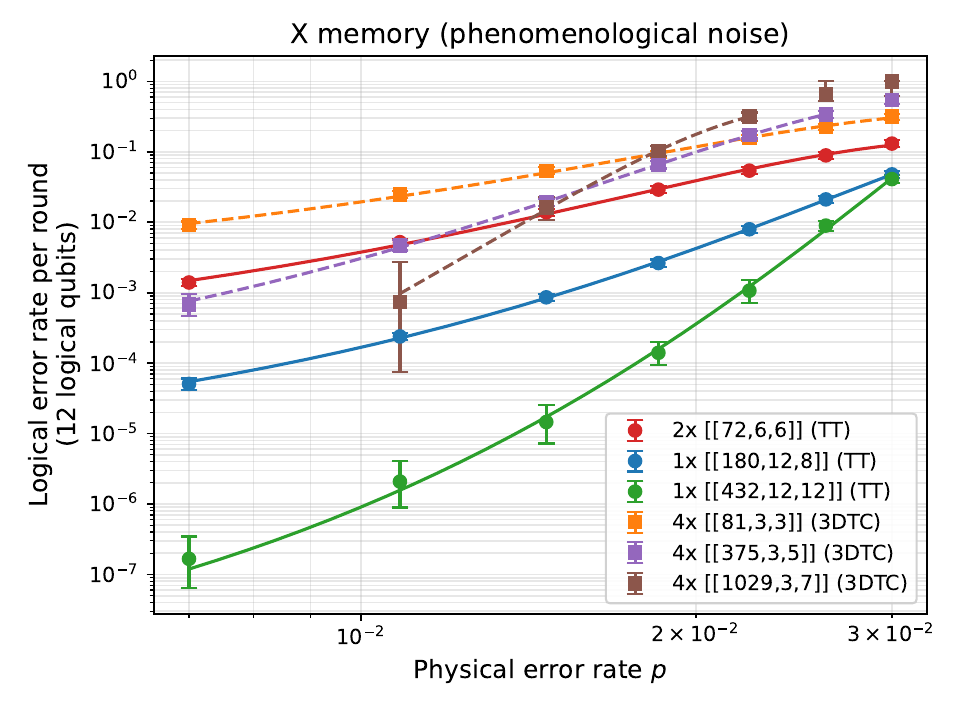}}\\
    \subfloat[]{\includegraphics[width=\linewidth]{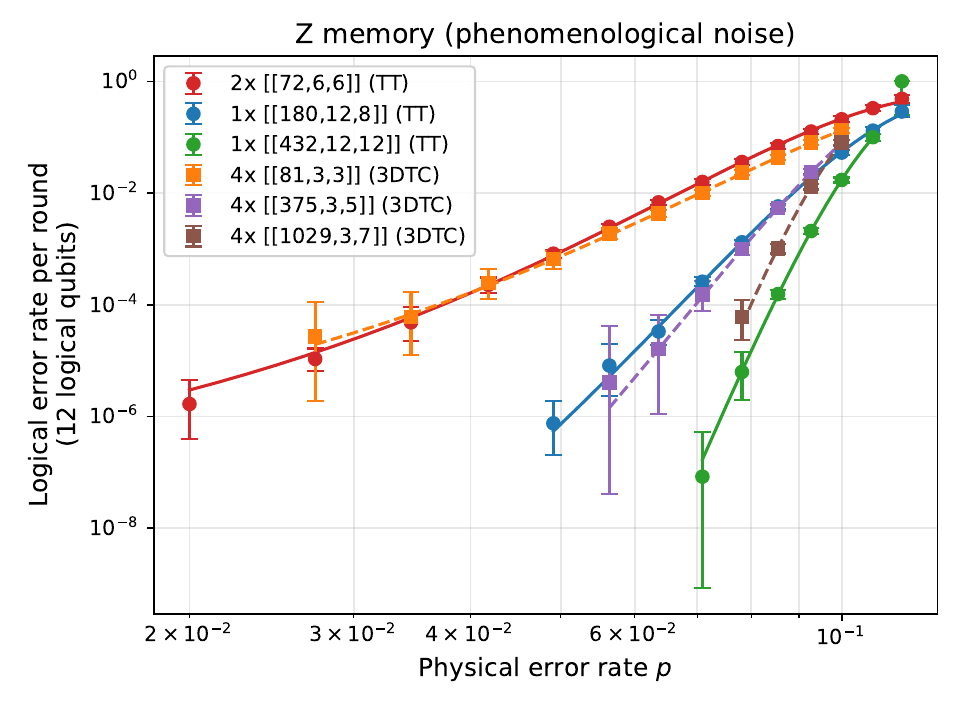}}
    \caption{(a) Comparison of TT code X memory results with those of the 3D toric codes. The decoder used is BP+OSD-CS30 in all cases, with 100 BP iterations and a min-sum scaling factor of $0.625$, and with no windowing used. (Dashed) lines are drawn between only those data points included in fits, and we exclude high-logical-error data points from fits. (b) Comparison of Z memory experiments for TT and 3D toric codes. We use a $(2,1)$ windowing strategy in all cases, with BP+OSD0, $150$ BP iterations and min-sum scaling factor of $0.625$.}
    \label{fig:TT_vs_3DTC_X}
\end{figure}

Here, we provide phenomenological noise simulations comparing TT codes with 3D toric codes at similar distances. We present our data in Fig.~\ref{fig:TT_vs_3DTC_X}. The BP+OSD parameters used for decoding are presented in the figure caption. We compare the $[[72,6,6]]$, $[[180,12,8]]$, and $[[432,12,12]]$ TT codes to the $[[81,3,3]]$, $[[375,3,5]]$, and $[[1029,3,7]]$ 3D toric codes. In the figure, we plot the logical error probability of $N$ code blocks of the same size, such that $Nk = 12$. That is, we report the logical error probability of two copies of $[[72,6,6]]$ code, and four copies of 3DTC ($N = 1$ for the other two TT codes plotted). This allows for comparison between codes that encode differing numbers of logical qubits. Given the block error rate $p_L$, we rescale to the error rate of $N$ blocks using the formula:
\begin{align}
    P_N = 1 - (1-p_L)^N,
\end{align}
i.e., the probability that any of the $N$ blocks incur a logical error. We fit each code's data points to an independent (quadratic exponential) function: $A\exp(-Bp^2 - Cp)$, for arbitrary parameters $A,B,C$. We only include data points in the fits for which the logical error probability (rescaled) is below $0.5$, as above this point, we see saturation behaviour. In Fig.~\ref{fig:TT_vs_3DTC_X}, we only plot the fitted curves between those data points included in the data used for generating the fit.  

It should be noted that the 3D toric codes used for comparison have far worse encoding rates than the TT codes. Note that a fair comparison between codes is made by comparing the footprint required to encoding some number of logical qubits (in this case 12) at a fixed physical noise rate $p$, while obtaining a target total logical error rate $p_L$. By this metric, the 3D toric codes are significantly less qubit-efficient than the TT codes.
For example, consider $X$ memory at $p\approx 0.0108$. At this physical noise rate, (two copies of) the $[[72,6,6]]$ code performs the same as (four copies of) the $[[375,3,5]]$ 3DTC. But the TT code's footprint is $10.4\times$ smaller than the 3DTC's. Meanwhile, at the same $p$, the $[[180,12,8]]$ TT code performs similarly to four copies of the $[[1029,3,7]]$ 3D toric code, despite having an encoding rate $22\times$ smaller than that 3D toric code's. Finally, the $[[432,12,12]]$ code at the same $p$ has a logical error rate over $2$ orders of magnitude lower than the $[[1029,3,7]]$ code ($2.1\times 10^{-6}$ vs $7.4\times 10^{-4}$), while still using around $9.5\times$ fewer qubits. Interestingly, when comparing $Z$ memory, the three 3DTCs follow the three TT codes studied quite closely in performance. (Note that these simulations are sensitive to $X$ errors, and each 3DTC has $d_X = d_Z^2$, while the TT codes studied have smaller distance asymmetry.) 
Hence, the $[[180,12,8]]$ code performs similarly to four copies of the $[[375,3,5]]$ code despite having a smaller $X$-distance and a much lower footprint. Similarly, the $[[432,12,12]]$ code performs better than the $[[1029,3,7]]$ code despite its smaller $X$-distance. Notice also that the $X$-memory crossing points for the 3D toric codes are lower than the TT codes' crossing points, indicating an overall lower threshold (we do not see clear threshold-like behaviour at the system sizes tested, however). 

\subsection{Comparing Performance with Different Decoder Parameters}

\begin{figure}
    \centering
    \subfloat[]{\includegraphics[width=\linewidth]{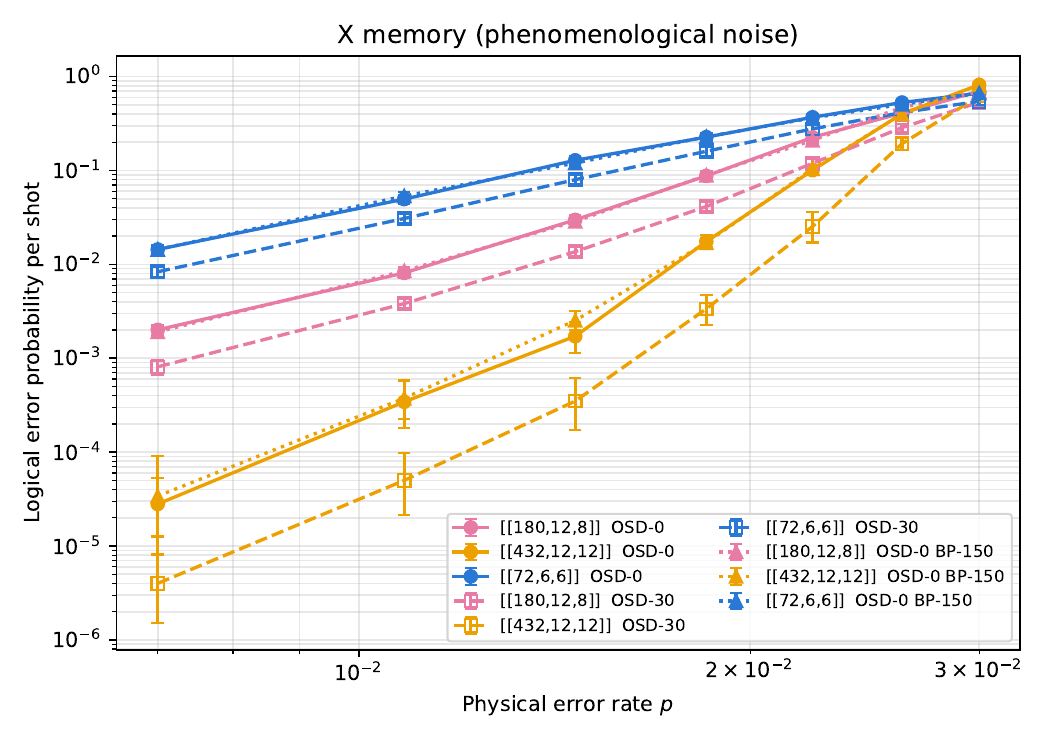}}\\
    \subfloat[]{\includegraphics[width=\linewidth]{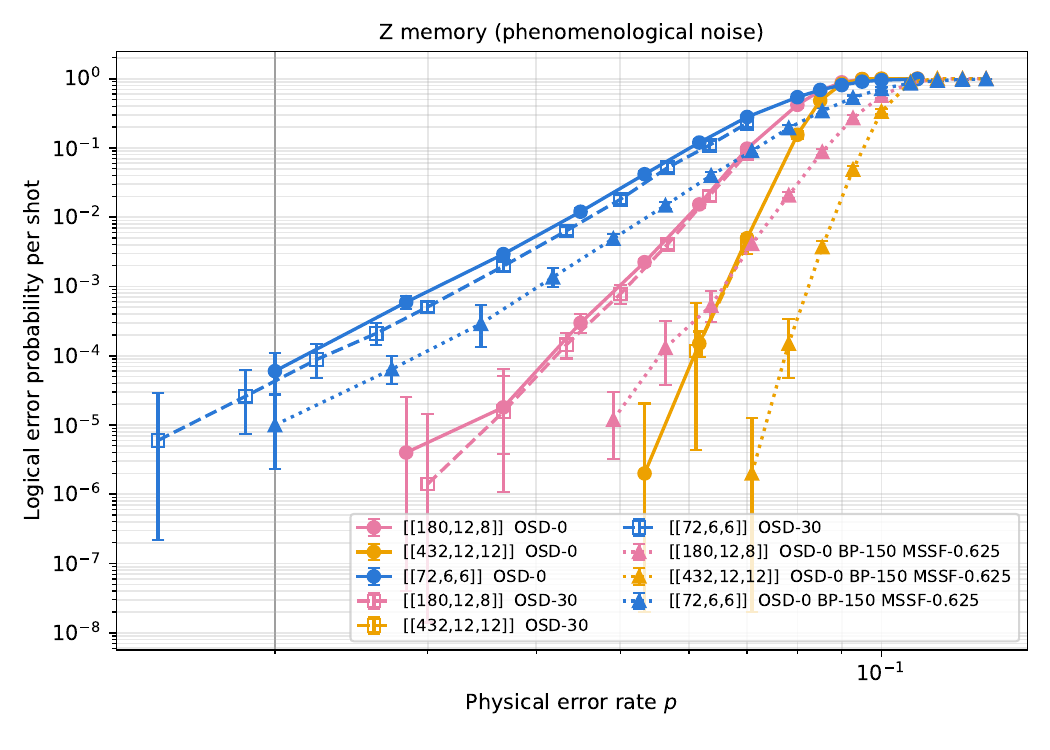}}
    \caption{Comparison of various decoder parameter choices. (a) $X$ memory simulations using full-syndrome decoding with OSD-CS. For OSD order $0$ we use $100$ and $150$ BP iterations, and with OSD order $30$ we use $150$ BP iterations. We set the min-sum scaling factor (MSSF) to $0.625$ for all curves. (b) BP+OSD-CS decoding under a $(2,1)$ windowing strategy for $Z$ memory. Circles indicate results for BP+OSD0, using 100 BP iterations and a MSSF of 1.0. Square data points indicate the same BP parameters and an OSD order of 30. Triangle data points indicate OSD order 0, 150 BP iterations, and a MSSF of 0.625.}
    \label{fig:Z_param_comparison}
\end{figure}

In Fig.~\ref{fig:Z_param_comparison}, we compare the performance of TT codes with a BP+OSD decoder using different decoding parameters. For $X$ memory simulations (Fig.~\ref{fig:Z_param_comparison}(a)), we vary the number of BP iterations and OSD-CS order, from a starting point of OSD0 and 100 iterations (we use a min-sum scaling factor of $0.625$ for all curves. We observe that increasing the number of BP iterations does not influence the logical error rate at all, while increasing OSD order to 30 does improve performance. For $Z$ memory simulations, we observe no effect from increasing the OSD order, but we see that the BP iterations and min-sum scaling factor does make a difference to the logical error rate.

\section{Confinement Profiles and Average Confinement}\label{app:confinement}

\begin{figure}
    \centering
    \subfloat[]{\includegraphics[width=\linewidth]{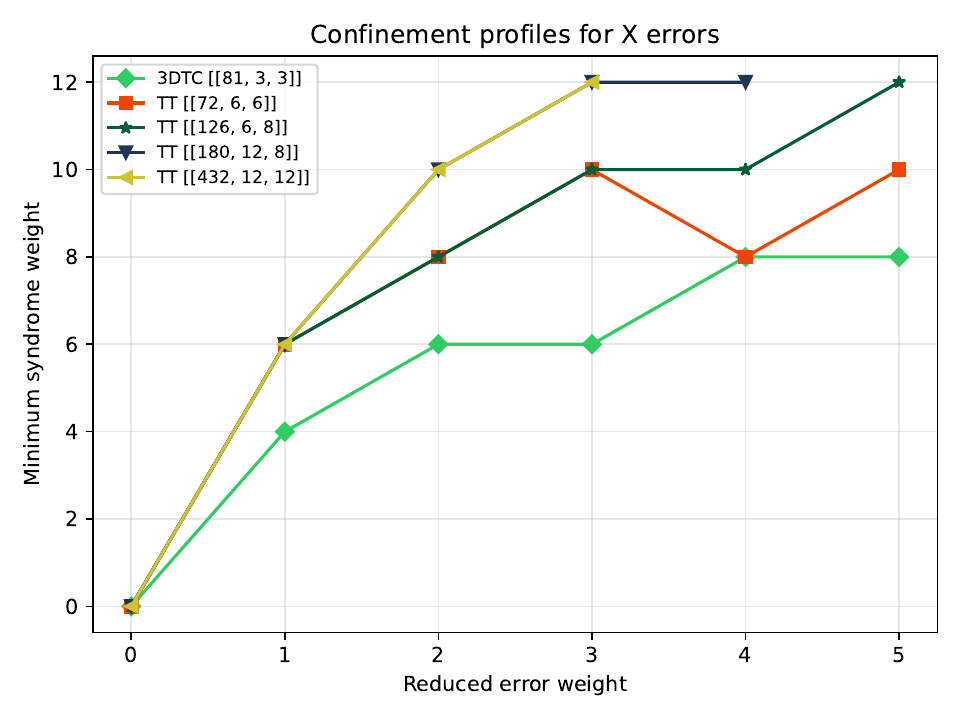}}\\
    \subfloat[]{\includegraphics[width=\linewidth]{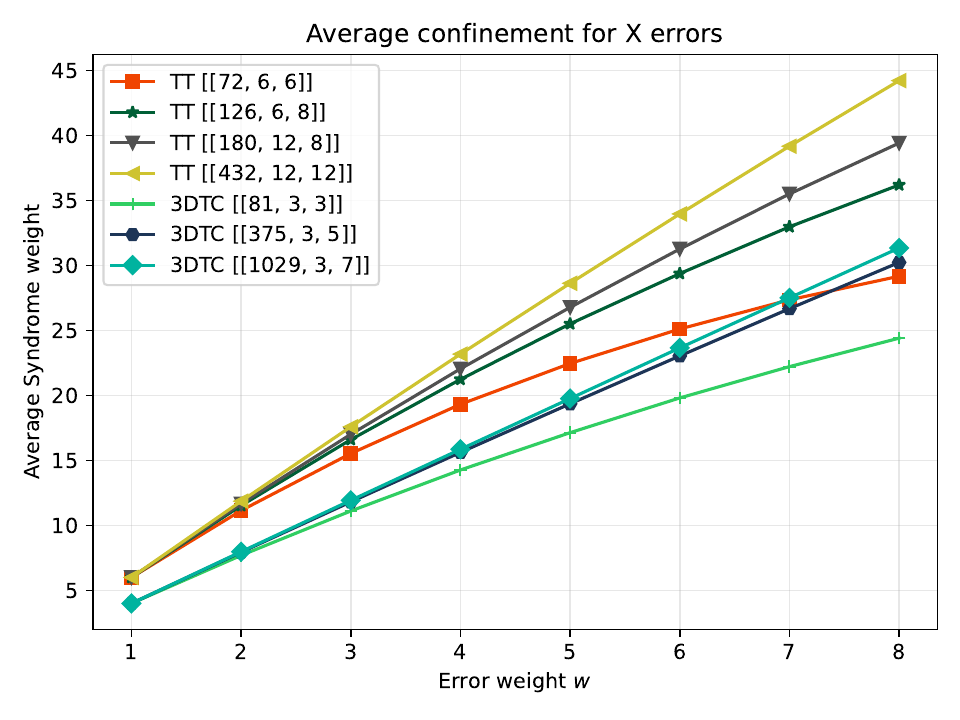}}
    \caption{(a) Confinement profile for TT and 3D toric codes under the $X$ error channel. (b) Average confinement for TT and 3D toric codes under the $X$ error channel. Averages are estimated from $30,000$ samples of errors of weight $w$.}
    \label{fig:confinement_x_errors}
\end{figure}

Single-shot decodability, at least under adversarial noise, is related to the \textit{confinement} of the code, which is an asymptotic property of code families~\cite{Bombin_2015,quintavalle2021single}. For finite-sized instances of codes, other works have investigated the confinement profile of code instances~\cite{scruby2024radial_codes,mian2026multivariatemulticyclecodescomplete}. In this work, we define the confinement profile for errors of Pauli type $P$ to be the function~\cite{quintavalle2021single}: $c_P(|E|) = \min_{e\in \mathcal{E}_P,\, |e|_\text{red} = |E|} |\sigma(e)|$, where $\sigma(e)$ is the syndrome set of error $e$, $\mathcal{E}_P$ is the set of errors of type $P$ (either $X$ or $Z$), and $|e|_\text{red}$ is the reduced weight of error $e$, which is the minimum weight of an error $e'$ in $\mathcal{E}$ such that $\sigma(e') = \sigma(e)$. TT codes have $X$ confinement profiles $c_X$ that grow rapidly with the reduced weight of $E$. We demonstrate this behaviour in Fig.~\ref{fig:confinement_x_errors}(a), comparing the confinement profiles of various TT codes to that of the 3D toric codes (note that all 3DTCs have the same confinement profiles for $X$ errors up to weight $\lfloor d_X/2 \rfloor$, and hence we only plot the confinement profile of the $[[81,3,3]]$ code). We see that the TT codes confinement profiles for low reduced error weights is generally higher than that of the 3DTC, indicating a guarantee of high syndrome weight as the error weight increases. 

We performed the calculation of confinement profiles by enumerating all errors up to a given maximum weight, evaluating their syndromes and then matching the ``coset-leaders" to each syndrome pattern (the coset-leader is the lowest-weight error consistent with that syndrome). The confinement profile plot is then the plot of the ordered coset-leader weights versus the corresponding syndrome weights. Note that this is a different method than and evaluates a slightly different quantity to the main other method used in the literature~\cite{mian2026multivariatemulticyclecodescomplete,scruby2024radial_codes,pryadko2024distm4ri}.

Another interesting quantity is the \textit{average} confinement~\cite{scruby2024radial_codes} for errors of type $P$: $c^\text{ave}_P(w) = \text{ave}_{e\in \mathcal{E}_P,\, |e| = w} (|\sigma(e)|)$. This is the average weight of syndrome $\sigma(e)$ over all errors $e$ of a given weight, $w$. We find that for low $w$, the average confinement grows roughly linearly with $w$. In Fig.~\ref{fig:confinement_x_errors}(b), we plot the average $X$-error confinement for TT and 3D toric codes. The results are qualitatively similar to the $Z$-error average confinement: below a weight of $7$, all TT codes have better average confinement (larger average syndrome weight for errors at weight $w$) than all 3D toric codes, even though the 3DTCs investigated have many more qubits, and hence typical errors of weight $w$ would likely be spread more sparsely through the graph. We attribute this greater average confinement in $X$ and $Z$ errors for TT codes to the larger degree of connectivity of the Tanner graphs of TT codes.

\section{Numerical Simulation Details and Syndrome measurement circuits}\label{app:SM_circuits}

\begin{figure*}
    \centering
    \includegraphics[width=\linewidth]{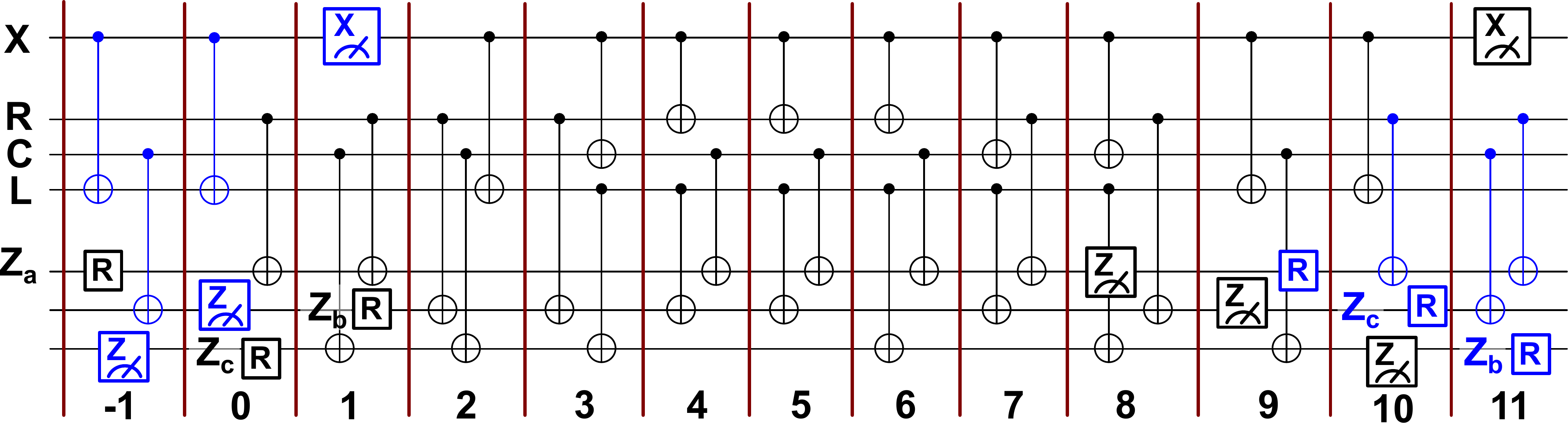}
    \caption{Syndrome measurement circuit for a trivariate tricycle code, with weight-9 $X$-checks and weight-6 $Z$-checks. Gate layers are indicated with numbers below the circuit and delineated with red lines between gates. Lines indicate blocks of data qubits (labelled R, C, and L) or measure qubits (labelled $X$, $Z_a$, $Z_b$, or $Z_c$). Blue gates, resets and measurements indicate circuit components from the previous or next syndrome measurement round. The final two qubit lines switch roles between $Z_b$- and $Z_c$-measure qubits each round, as indicated.}
    \label{fig:SM_circuit}
\end{figure*}

\begin{table}[t]
    \centering
    \renewcommand{\arraystretch}{1.5}
    \begin{tabular}{ccc}
        \toprule
        Gate & Error channel & \makecell{Noise \\ strength} \\
        \midrule
        CNOT gates & 2-qubit depolarising & $p$ \\
        \addlinespace
        \makecell{Measure in $Z$ \\ ($X$) basis} & \makecell{Classical outcome \\ flip} & $p$ \\
        Reset in $Z$ basis & Bit-flip & $p$\\
        \makecell{Idle during \\ gates/measurements} & 1-qubit depolarising & $p$\\
        \bottomrule
    \end{tabular}
    \caption{Uniform noise model applied in circuit-level simulations}
    \label{tab:noise_model}
\end{table}

In this Appendix, we provide further details of our numerical simulations and illustrate our syndrome measurement circuits, proving that they perform the required check measurements and preserve the logical operators of the code. 

For phenomenological noise simulations, we construct Stim circuits using multi-Pauli product measurements (only in the basis of the memory experiment) which suffer classical outcome flips with probability $p$, each one preceded by single-qubit depolarising noise channels of strength $p$. The circuit begins with an initial round of data qubit resets (noiseless) and terminates with data-qubit measurements (noiseless), both in the basis of the memory experiment. 

For circuit-level noise simulations, we construct stim circuits based on the syndrome measurement circuits detailed below, and supply them with a uniform depolarising noise model specified in \Cref{tab:noise_model}. Table~\ref{tab:dem_sizes} provides sizes of the Stim Detector Error Models (DEMs) of the memory circuits, used in Sec.~\ref{sec:CLN_Sims}. There is a significant size asymmetry between the $X$ and $Z$ DEMs, resulting from $H_Z$ having triple the row count of $H_X$. Additionally, we analysed if exchanging definitions of the polynomial terms (e.g. $A_1 \leftrightarrow A_3$) which in turn reorders the CNOT gates in the syndrome circuit, affected the circuit-level distance. We exhausted all permutations and found no increase in the circuit distance, this suggests that the current ordering is optimal.

\begin{table}[ht]
\centering
\renewcommand{\arraystretch}{1.3}
\begin{tabular}{cccc}
\toprule
Memory & Code & Rounds & DEM Size \\
\midrule

\multirow{3}{*}{$X$ Memory}
    & $[[72,6,6]]$ & 7 & $192 \ \times \ 3696$ \\
    & $[[180,12,8]]$ & 9 & $600 \ \times \ 11,880$ \\
    & $[[432,12,12]]$ & 13 & $2016 \ \times \ 41,184$ \\
\midrule
\multirow{3}{*}{$Z$ Memory}
    & $[[72,6,6]]$ & 13 & $1008 \ \times \ 9360$ \\
    & $[[180,12,8]]$ & 13 & $2520 \ \times \ 23,400$ \\
    & $[[432,12,12]]$ & 13 & $6048 \ \times \ 56,160$ \\

\bottomrule
\end{tabular}
\caption{DEM sizes for TT codes used in the circuit-level noise simulations in Sec.~\ref{sec:CLN_Sims}.}
\label{tab:dem_sizes}
\end{table}

In each case, we only include detectors of the same Pauli type as the memory experiment (i.e., we perform uncorrelated decoding). The Stim circuit's detector error model is converted into a parity-check matrix and list of prior probabilities~\cite{Roffe_LDPC_Python_tools_2022}. For decoding, we use belief propagation (BP) and ordered statistics decoding (OSD)~\cite{panteleev2021degenerate,roffe_decoding_2020,Roffe_LDPC_Python_tools_2022}. We use $\geq 100$ BP iterations, a minimum-sum BP method (with scaling factor of either $1.0$ or $0.625$), and a parallel schedule. For OSD, we use either OSD0 or the combination sweep (CS) method with the order set to 30. The specific decoder parameters used for each simulation can be found in the captions of the figures in which the results are presented.

The syndrome extraction circuit is summarised in Fig.~\ref{fig:SM_circuit}. We label blocks of qubits (represented by wires in the figure) by $X$ (for $X$-check measure qubits), R, C, L (for the right, center, and left blocks of data qubits), $Z_a$, $Z_b$ and $Z_c$ (for the three blocks of $Z$-check measure qubits, corresponding to the three blocks of rows in $H_Z$). Each CNOT gate in the diagram represents a collection of CNOT gates, all of which can be performed in parallel, between data and measure qubits in the corresponding blocks (see below). %

Note that these circuits have the complication that the $Z_b$ and $Z_c$ blocks of measure qubits switch labels every round -- this is indicated after the resets in Fig.~\ref{fig:SM_circuit}. In order to not introduce an extra idling layer, we omit explicit resets on the $X$-measure qubits, instead performing a Pauli frame update dependent on the outcomes of measurements of the $X$-measure qubits (alternatively, one could assume having an additional $n/3$ $X$-measure qubits available). The lack of resets can be accommodated by updating the definitions of detectors (see, e.g., Ref.~\cite{gehér2025resetresetquestion}). This should not drastically affect quantum memory performance, though it decreases the time-like distance of certain logical operations, like (generalised) lattice surgery (see, e.g., Ref.~\cite{gehér2025resetresetquestion,harper2025characterisingfailuremechanismserrorcorrected}). For simplicity, in our simulations, we perform a noiseless reset rather than update the definitions of detectors. As drawn, the circuits have depth 13. The first reset in a layer is on the $Z_a$-measure qubits and the final measurement is on the $X$-measure qubits. If we perform resets on the $X$-measure qubits, the circuits become depth-14, however the 2-qubit gate depth remains the same. 

The ordering of the matrices in the gates is important and noted in the equations below. We prove that the syndrome measurement circuits perform as required by observing how they update the $X$ and $Z$ stabilizer tableaux (see also Ref.~\cite{bravyi2024high}). Let us consider the $X$-tableau first. We require the syndrome measurement circuit to update this tableau in the following way:
\begin{align}
    \begin{pmatrix}
        I & 0       & 0     & 0      & 0 & 0 & 0\\
        0 & A  & B & C & 0 & 0 & 0
    \end{pmatrix} \longrightarrow \begin{pmatrix}
        I & A  & B & C & 0 & 0 & 0\\
        0 & A  & B & C & 0 & 0 & 0
    \end{pmatrix}
\end{align}
where the columns of the tableau are labelled $X$, $L$, $C$, $R$, $Za$, $Zb$, and $Zc$, in order. 

For the $Z$-tableau, we require the following transformation:
\begin{align}
\begin{pmatrix}
    0 & 0         & 0 & 0               & I     & 0     & 0  \\
    0 & 0   & 0       & 0               & 0     & I     & 0 \\
    0 & 0    & 0 & 0                    & 0     & 0     & I\\
    0 & 0         & C^\top & B^\top    & 0     & 0     & 0\\
    0 & C^\top   & 0       & A^\top   & 0     & 0     & 0\\
    0 & B^\top    & A^\top & 0         & 0     & 0     & 0
\end{pmatrix} \longrightarrow
    \begin{pmatrix}
        0 & 0         & C^\top & B^\top    & I     & 0     & 0  \\
        0 & C^\top   & 0       & A^\top   & 0     & I     & 0 \\
        0 & B^\top    & A^\top & 0         & 0     & 0     & I   \\
        0 & 0         & C^\top & B^\top    & 0     & 0     & 0\\
        0 & C^\top   & 0       & A^\top   & 0     & 0     & 0\\
        0 & B^\top    & A^\top & 0         & 0     & 0     & 0
    \end{pmatrix}
\end{align}
where the columns are labelled similarly to above. We show that the circuit below performs both of these transformations.

For the $X$-tableau, we will omit the first column since this is unchanged throughout the circuit. The notation CX$_{B_1}(R,Za)$ means we perform a CNOT gate between qubit $i$ in the $R$ block of data qubits, and $Za[B_1[i]]$, in the block of $Z$ measure qubits, for $i=1,\ldots, n/3$, where $B_1[i] = j$ if and only if row $i$ and column $j$ of $B_1$ is $1$.
\begin{widetext}
    \begin{align}
        &\text{Round 0:}
        \quad \begin{pmatrix}
            0       & 0     & 0      & 0              & 0 & 0\\
            A  & B & C & 0              & 0 & 0
        \end{pmatrix} \xrightarrow{\text{CX}_{B_1}(R,Za)} \begin{pmatrix}
            0       & 0     & 0      &      0         & 0 & 0\\
            A  & B & C & CB_1  & 0 & 0
        \end{pmatrix}\\ \nonumber
        \\
        &\text{Round 1:}\xrightarrow[\text{CX}_{A_1}(C,Zc)]{\text{CX}_{B_3}(R,Za)}
        \begin{pmatrix}
            0       & 0     & 0      &      0                   & 0 & 0\\
            A  & B & C & C(B_1+B_3)  & 0 & BA_1
        \end{pmatrix}\\ \nonumber
        \\
        &\text{Round 2:}\xrightarrow[\begin{subarray}{c}\text{CX}_{A_3}(C,Zc)\\ \text{CX}_{A_2}(X,L)\end{subarray}]{\text{CX}_{A_1}(R,Zb)} \begin{pmatrix}
            A_2 & 0     & 0      &      0                  & 0 & 0\\
            A   & B & C & C(B_1+B_3) & CA_1 & B(A_1+A_3)
        \end{pmatrix}\\ \nonumber
        \\
        &\text{Round 3:}\xrightarrow[\begin{subarray}{c}\text{CX}_{B_2}(X,C)\\ \text{CX}_{B_1}(L,Zc)\end{subarray}]{\text{CX}_{A_3}(R,Zb)} \begin{pmatrix}
            A_2 & B_2 & 0      &      0                  & 0 & B_1A_2\\
            A   & B   & C & C(B_1+B_3) & C(A_1+A_3) & B(A_1+A_3) + B_1A
        \end{pmatrix}\\ \nonumber
        \\
        &\text{Round 4:}\xrightarrow[\begin{subarray}{c}\text{CX}_{C_1}(C,Za)\\ \text{CX}_{C_1}(L,Zb)\end{subarray}]{\text{CX}_{C_2}(X,R)} \begin{pmatrix}
            A_2 & B_2 & C_2 & C_1B_2 & C_1A_2 & B_1 A_2\\
            A   & B   & C   & C(B_1+B_3) + C_1B & C(A_1+A_3) + C_1A & B(A_1+A_3) + B_1A
        \end{pmatrix}\\ \nonumber
        \\
        &\text{Round 5:}\xrightarrow[\begin{subarray}{c}\text{CX}_{C_2}(C,Za)\\ \text{CX}_{C_2}(L,Zb)\end{subarray}]{\text{CX}_{C_1}(X,R)} \begin{pmatrix}
            A_2 & B_2 & C_1+C_2 & (C_1+C_2)B_2 & (C_1+C_2)A_2 & B_1 A_2\\
            A   & B   & C   & CB_2 + C_3B & CA_2 + C_3A & B(A_1+A_3) + B_1A
        \end{pmatrix}\\ \nonumber
        \\
        &\text{Round 6:}\xrightarrow[\begin{subarray}{c}\text{CX}_{C_3}(C,Za)\\ \text{CX}_{B_2}(L,Zc)\end{subarray}]{\text{CX}_{C_3}(X,R)} \begin{pmatrix}
            A_2 & B_2 & C & CB_2 & (C_1+C_2)A_2 & (B_1 +B_2) A_2\\
            A   & B   & C   & CB_2 & CA_2 + C_3A & BA_2 + B_3A
        \end{pmatrix}\\ \nonumber
        \\
        &\text{Round 7:}\xrightarrow[\begin{subarray}{c}\text{CX}_{B_1}(X,C)\\ \text{CX}_{C_3}(L,Zb)\end{subarray}]{\text{CX}_{B_2}(R,Za)} \begin{pmatrix}
            A_2 & B_1+B_2 & C & 0 & CA_2 & (B_1 +B_2) A_2\\
            A   & B   & C & 0 & CA_2 & BA_2 + B_3A
        \end{pmatrix}\\ \nonumber
        \\
        &\text{Round 8:}\xrightarrow[\begin{subarray}{c}\text{CX}_{B_3}(X,C)\\ \text{CX}_{B_3}(L,Zc)\end{subarray}]{\text{CX}_{A_2}(R,Zb)} \begin{pmatrix}
            A_2 & B & C & 0 & 0 & B A_2\\
            A   & B & C & 0 & 0 & BA_2
        \end{pmatrix}\\ \nonumber
        \\
        &\text{Round 9:}\xrightarrow[\text{CX}_{A_1}(X,L)]{\text{CX}_{A_2}(C,Zc)} \begin{pmatrix}
            A_1+A_2 & B & C & 0 & 0 & 0\\
            A   & B & C & 0 & 0 & 0
        \end{pmatrix}\\ \nonumber
        \\
        &\text{Round 10:}\xrightarrow{\text{CX}_{A_3}(X,L)}\begin{pmatrix}
            A & B & C & 0 & 0 & 0\\
            A   & B & C & 0 & 0 & 0
        \end{pmatrix}
    \end{align}
    The $Z$-tableau works similarly. We will omit the last three columns since they are unchanged throughout the circuit.
    \begin{align}
        &\text{Round 0:}
        \quad \begin{pmatrix}
            0       & 0         & 0       & 0   \\
            0       & 0         & 0       & 0   \\
            0       & 0         & 0       & 0   \\
            0       & 0         & C^\top & B^\top\\
            0       & C^\top   & 0       & A^\top\\
            0       & B^\top    & A^\top & 0
        \end{pmatrix} \xrightarrow{\text{CX}_{B_1}(R,Za)} \begin{pmatrix}
            0       & 0         & 0       & B^\top_1   \\
            0       & 0         & 0       & 0   \\
            0       & 0         & 0       & 0   \\
            0       & 0         & C^\top & B^\top\\
            0       & C^\top   & 0       & A^\top\\
            0       & B^\top    & A^\top & 0
        \end{pmatrix}\\ \nonumber
        \\
        &\text{Round 1:}\xrightarrow[\text{CX}_{A_1}(C,Zc)]{\text{CX}_{B_3}(R,Za)}
        \begin{pmatrix}
            0       & 0         & 0           & B^\top_1+B^\top_3  \\
            0       & 0         & 0           & 0   \\
            0       & 0         & A^\top_1   & 0   \\
            0       & 0         & C^\top     & B^\top\\
            0       & C^\top   & 0           & A^\top\\
            0       & B^\top    & A^\top     & 0
        \end{pmatrix}\\ \nonumber
        \\ 
        &\text{Round 2:}\xrightarrow[\begin{subarray}{c}\text{CX}_{A_3}(C,Zc)\\ \text{CX}_{A_2}(X,L)\end{subarray}]{\text{CX}_{A_1}(R,Zb)} \begin{pmatrix}
            0                & 0         & 0                     & B^\top_1+B^\top_3  \\
            0                & 0         & 0                     & A^\top_1  \\
            0                & 0         & A^\top_1+A^\top_3   & 0   \\
            0                & 0         & C^\top               & B^\top\\
            A^\top_2C^\top & C^\top   & 0                     & A^\top\\
            A^\top_2B^\top  & B^\top    & A^\top               & 0
        \end{pmatrix}\\ \nonumber
        \\ 
        &\text{Round 3:}\xrightarrow[\begin{subarray}{c}\text{CX}_{B_2}(X,C)\\ \text{CX}_{B_1}(L,Zc)\end{subarray}]{\text{CX}_{A_3}(R,Zb)} \begin{pmatrix}
            0                                 & 0         & 0                     & B^\top_1+B^\top_3  \\
            0                                 & 0         & 0                     & A^\top_1+A^\top_3  \\
            B^\top_2(A^\top_1+A^\top_3)     & B^\top_1  & A^\top_1+A^\top_3   & 0   \\
            B^\top_2C^\top                   & 0         & C^\top               & B^\top\\
            A^\top_2C^\top                  & C^\top   & 0                     & A^\top\\
            A^\top_2B^\top + A^\top B^\top_2 & B^\top    & A^\top               & 0
        \end{pmatrix}\\ \nonumber
        \\ 
        &\text{Round 4:}\xrightarrow[\begin{subarray}{c}\text{CX}_{C_1}(C,Za)\\ \text{CX}_{C_1}(L,Zb)\end{subarray}]{\text{CX}_{C_2}(X,R)} \begin{pmatrix}
            C^\top_2(B^\top_1+B^\top_3)      & 0         & C^\top_1           & B^\top_1+B^\top_3  \\
            C^\top_2(A^\top_1+A^\top_3)    & C^\top_1 & 0                   & A^\top_1+A^\top_3 \\
            B^\top_2(A^\top_1+A^\top_3)     & B^\top_1  & A^\top_1+A^\top_3 & 0   \\
            B^\top_2C^\top + B^\top C^\top_2 & 0         & C^\top             & B^\top\\
            A^\top_2C^\top+A^\top C^\top_2 & C^\top   & 0                   & A^\top\\
            A^\top_2B^\top + A^\top B^\top_2 & B^\top    & A^\top             & 0
        \end{pmatrix}\\ \nonumber
        \\
        &\text{Round 5:}\xrightarrow[\begin{subarray}{c}\text{CX}_{C_2}(C,Za)\\ \text{CX}_{C_2}(L,Zb)\end{subarray}]{\text{CX}_{C_1}(X,R)} \begin{pmatrix}
            (C^\top_1+C^\top_2)(B^\top_1+B^\top_3)  & 0                   & C^\top_1+C^\top_2 & B^\top_1+B^\top_3  \\
            (C^\top_1+C^\top_2)(A^\top_1+A^\top_3)& C^\top_1+C^\top_2 & 0                   & A^\top_1+A^\top_3 \\
            B^\top_2(A^\top_1+A^\top_3)             & B^\top_1  & A^\top_1+A^\top_3 & 0   \\
            B^\top_2C^\top + B^\top(C^\top_1+C^\top_2) & 0         & C^\top             & B^\top\\
            A^\top_2C^\top+A^\top(C^\top_1+C^\top_2) & C^\top   & 0                   & A^\top\\
            A^\top_2B^\top + A^\top B^\top_2 & B^\top    & A^\top             & 0
        \end{pmatrix}\\{\nonumber}
        \\
        &\text{Round 6:}\xrightarrow[\begin{subarray}{c}\text{CX}_{C_3}(C,Za)\\ \text{CX}_{B_2}(L,Zc)\end{subarray}]{\text{CX}_{C_3}(X,R)} \begin{pmatrix}
            C^\top(B^\top_1+B^\top_3)  & 0                   & C^\top             & B^\top_1+B^\top_3  \\
            C^\top(A^\top_1+A^\top_3)& C^\top_1+C^\top_2 & 0                   & A^\top_1+A^\top_3 \\
            B^\top_2(A^\top_1+A^\top_3)& B^\top_1 +B^\top_2 & A^\top_1+A^\top_3 & 0   \\
            B^\top_2C^\top + B^\top C^\top   & 0         & C^\top             & B^\top\\
            (A^\top_1+A^\top_3)C^\top   & C^\top   & 0                   & A^\top\\
            A^\top_2B^\top + A^\top B^\top_2 & B^\top    & A^\top             & 0
        \end{pmatrix}
    \end{align}
    
    \begin{align}
        &\text{Round 7:}\xrightarrow[\begin{subarray}{c}\text{CX}_{B_1}(X,C)\\ \text{CX}_{C_3}(L,Zb)\end{subarray}]{\text{CX}_{B_2}(R,Za)} \begin{pmatrix}
            C^\top B^\top_3                         & 0                   & C^\top             & B^\top  \\
            C^\top(A^\top_1+A^\top_3)            & C^\top & 0                   & A^\top_1+A^\top_3 \\
            (B^\top_1+B^\top_2)(A^\top_1+A^\top_3)& B^\top_1 +B^\top_2  & A^\top_1+A^\top_3 & 0   \\
            B^\top_3C^\top   & 0         & C^\top             & B^\top\\
            (A^\top_1+A^\top_3)C^\top   & C^\top   & 0                   & A^\top\\
            A^\top_2B^\top + A^\top(B^\top_1+B^\top_2) & B^\top    & A^\top             & 0
        \end{pmatrix}\\ \nonumber
        \\
        &\text{Round 8:}\xrightarrow[\begin{subarray}{c}\text{CX}_{B_3}(X,C)\\ \text{CX}_{B_3}(L,Zc)\end{subarray}]{\text{CX}_{A_2}(R,Zb)} \begin{pmatrix}
            0                                       & 0                   & C^\top             & B^\top  \\
            C^\top(A^\top_1+A^\top_3)            & C^\top & 0                   & A^\top \\
            B^\top(A^\top_1+A^\top_3)& B^\top  & A^\top_1+A^\top_3 & 0   \\
            0                                       & 0         & C^\top             & B^\top\\
            (A^\top_1+A^\top_3)C^\top            & C^\top   & 0                   & A^\top\\
            (A^\top_1+ A^\top_3)B^\top            & B^\top    & A^\top             & 0
        \end{pmatrix}\\ \nonumber
        \\
        &\text{Round 9:}\xrightarrow[\text{CX}_{A_1}(X,L)]{\text{CX}_{A_2}(C,Zc)} \begin{pmatrix}
            0                & 0         & C^\top & B^\top  \\
            C^\top A^\top_3 & C^\top   & 0       & A^\top \\
            B^\top A^\top_3  & B^\top    & A^\top & 0   \\
            0                & 0         & C^\top & B^\top\\
            A^\top_3C^\top & C^\top   & 0       & A^\top\\
            A^\top_3B^\top  & B^\top    & A^\top & 0
        \end{pmatrix}\\ \nonumber
        \\
        &\text{Round 10:}\xrightarrow{\text{CX}_{A_3}(X,L)}\begin{pmatrix}
            0                & 0         & C^\top & B^\top  \\
            0                & C^\top   & 0       & A^\top \\
            0                & B^\top    & A^\top & 0   \\
            0                & 0         & C^\top & B^\top\\
            0                & C^\top   & 0       & A^\top\\
            0                & B^\top    & A^\top & 0
        \end{pmatrix}
    \end{align}
\end{widetext}

Now let us consider the logical operators. We will consider a $Z$-logical first. As in Ref.~\cite{bravyi2024high}, we write this as $(u,w,v)$ with $u$, $w$ and $v$ restrictions of the operator to the $L$, $C$ and $R$ qubits, respectively. The commutation of $(u,w,v)$ with the $X$ parity check matrix means:
\begin{align}
    \begin{bmatrix}
        A & B & C
    \end{bmatrix} \begin{bmatrix}
        u\\
        w\\
        v
    \end{bmatrix} = A u + B w + C v = 0.
\end{align}

Meanwhile, let us consider the action of the SM circuit on the logical. Since this logical has $Z$ support, the CNOT gates controlled on data qubits have no action on the logical operator. Only the CNOTs targeted on data qubits have action; these are those involved in the $X$-check measurements. The support of the updated logical on the $X$-measurement qubits is given by:
\begin{align}
    t = A u + B w + C v = 0.
\end{align}

For an $X$ logical operator, $(u,w,v)$, the vectors obey:
\begin{align}
    C^\top w + B^\top v = C^\top u + A^\top v = B^\top u + A^\top w = 0.
\end{align}
The supports of the logical on qubits $Z_a$, $Z_b$ and $Z_c$ are given by $t_a$, $t_b$ and $t_c$ respectively:
\begin{align}
    t_a & = (C^\top_1 + C^\top_2 + C^\top_3)w + (B^\top_1+B^\top_2+B^\top_3) v = 0\\
    t_b & = C^\top u + A^\top v = 0\\
    t_c & = B^\top u + A^\top w = 0 .
\end{align}
Therefore the logical operators are preserved by the circuit.

\section{Details of logical $CCZ$ gates}\label{app:CCZ_gates}

\subsection{Review of Cup Product Conditions}

Here we review some of the ingredients entering into the cup product construction of Ref.~\cite{breuckmann2024cupsgatesicohomology}, from which we derive the non-Clifford gates acting on (4,2,2) codes and others in this paper. 

Consider a classical group algebra code, for abelian $G$:
\begin{align}
    C^0\xrightarrow{A}C^1,
\end{align}
where $C^0 = C^1 = \mathbb{F}_2[G]$.
Here, $\mathbb{F}_2[G]$ denotes the formal sums of the form $\sum_{g\in G} a_g g$, for $a_g\in \lbrace 0, 1\rbrace$, and $A$ is an element of $\mathbb{F}_2[G]$. For our purposes, we can assume $G = \mathbb{Z}_\ell \times \mathbb{Z}_m \times \mathbb{Z}_p$ and $A$ is a polynomial in the generators of the cyclic subgroups of $G$, as is appropriate for a TT code.

We define a \textit{pre-orientation} on this code by partitioning the terms of $A$ into three disjoint sets: $A = A_\text{in} + A_\text{out} + A_\text{free}$. Let $X_0$ and $X_1$ represent bases for $C^0$, the checks (on which $A$ acts), and $C^1,$ the physical bits of the code, respectively. We then define the cup product, a bi-linear map on the chain complex, via: $\alpha \cup \alpha = \alpha$ for $\alpha \in X_0$; $\alpha \cup \beta = \beta$ for $\alpha \in X_0$ and $\beta \in A_\text{out}\alpha$; $\beta \cup \alpha = \beta$ for $\alpha \in X_0$, $\beta \in A_\text{in}\alpha$; zero on all other basis inputs. We furthermore define the integral on $C^1$ to be $\int_1 f = |f|$ (mod 2), for $f\in C^1$. The integrated Leibniz rule for the 3-fold cup product states that for all $\alpha_i\in C^0$:
\begin{align}
\begin{split}
    \int_1 (A\alpha_1\cup \alpha_2)\cup \alpha_3 + (\alpha_1\cup A\alpha_2)\cup \alpha_3  \\
    + (\alpha_1\cup\alpha_2)\cup A \alpha_3 = 0\quad \text{(mod }2).
\end{split}
\end{align}
Note that the multi-fold cup product is not necessarily associative, and hence we retain brackets in the above expression. Not all choices of pre-orientation obey this rule, but it is crucial to being able to define a cohomology operation (and hence a logical gate) from the (balanced) product of two or more of these classical codes.

By using the properties of the cup product, we can write the integrated Leibniz rule as:
\begin{align}\label{eqn:int_Leibniz_rule_2}
\begin{split}
    |A\alpha_1\cap A_\text{in}\alpha_2 \cap A_\text{in}\alpha_3| + |A_\text{out}\alpha_1\cap A\alpha_2 \cap A_\text{in}\alpha_3| \\
    + \delta_{\alpha_1,\alpha_2}|A_\text{out}\alpha_2 \cap A\alpha_3| = 0\quad \text{(mod } 2),
\end{split}
\end{align}
where $\delta_{\alpha_1,\alpha_2}=1$ if $\alpha_1=\alpha_2$, and is zero otherwise.

We now consider when the integrated Leibniz rule is satisfied. First let $\alpha_1=\alpha_2=\alpha_3$. Then the first term in Equation~\ref{eqn:int_Leibniz_rule_2} is simply $|A_\text{in}\alpha_1|$, the second term is zero and the final term is $|A_\text{out}\alpha_1|$. This produces the first cup-product condition (Equation~\ref{eqn:cup_cond_1}). Now suppose $\alpha_1=\alpha_3\neq\alpha_2$. Only the first term is potentially non-zero, and we find:
\begin{align}\label{eqn:in_in_intersection}
    |A_\text{in}\alpha_1 \cap A_\text{in}\alpha_2| = 0 \quad (\text{mod }2).
\end{align}
Now suppose $\alpha_1 = \alpha_2\neq \alpha_3$. We find, using Equation~\ref{eqn:in_in_intersection}:
\begin{align}
    &|A_\text{out}\alpha_2 \cap A_\text{in}\alpha_3| + |A_\text{out}\alpha_2 \cap A \alpha_3|\nonumber \\
    &= |A_\text{out}\alpha_2 \cap A_\text{free} \alpha_3| + |A_\text{out}\alpha_2 \cap A_\text{out} \alpha_3| \\
    &= 0\quad (\text{mod }2). 
\end{align}
If $\alpha_1 \neq \alpha_2 = \alpha_3$ (again, using Equation~\ref{eqn:in_in_intersection}) we have:
\begin{align}
    &|A\alpha_1 \cap A_\text{in}\alpha_2| + |A_\text{out}\alpha_1 \cap A_\text{in}\alpha_2|\nonumber \\
    &= |A_\text{free}\alpha_2 \cap A_\text{in}\alpha_3| = 0 \quad (\text{mod }2).
\end{align}
Finally, if all $\alpha_i$ are distinct, we have:
\begin{align}
\begin{split}
    |A\alpha_1 \cap A_\text{in}\alpha_2 \cap A_\text{in}\alpha_3| + |A_\text{out}\alpha_1 \cap A \alpha_2 \cap A_\text{in}\alpha_3| \\
    = 0\quad (\text{mod} 2).
\end{split}
\end{align}
We therefore obtain the conditions summarised below:
\begin{align}
    &|A_\text{in} (\alpha_1)| + |A_\text{out} (\alpha_1)| = 0 \label{eqn:cup_cond_1}\\
    &|A_\text{in} (\alpha_1) \cap A_\text{in} (\alpha_2)| = 0 \label{eqn:cup_cond_in}\\
    &|A_\text{out} (\alpha_1)\cap A_\text{out}(\alpha_2)| + |A_\text{free} (\alpha_1)\cap A_\text{out}(\alpha_2)| = 0 \label{eqn:cup_cond_3}\\
    &|A_\text{free} (\alpha_1)\cap A_\text{in}(\alpha_2)| = 0 \label{eqn:cup_cond_2}\\
    &|A (\alpha_1)\cap A_\text{in}(\alpha_2)\cap A_\text{in}(\alpha_3)| + \nonumber \\
    &\qquad \qquad |A_\text{out} (\alpha_1)\cap A(\alpha_2)\cap A_\text{in}(\alpha_3)| = 0, \label{eqn:cup_cond_4}
\end{align}
where all $\alpha_i$ are distinct.

The formulas for the cup product of the balanced product of three group algebra codes can be found in Ref.~\cite{breuckmann2024cupsgatesicohomology}. However, we point out that it is only non-zero for inputs of the form: $\alpha_\text{L} \cup \beta_\text{C} \cup \gamma_\text{R}$ (or for other permutations of the three data qubit blocks). For our purposes, only the form of the $CCZ$ circuits implementing the logical non-Clifford gate will be important. We introduce these below.

\subsection{Proof of Lemma~\ref{lem:cup_product}}

We now prove Lemma~\ref{lem:cup_product}, which we restate here:
\setcounter{lem}{4}
\begin{lem}
A TT code with all three polynomials ($A$, $B$, and $C$) either weight-2 or of the form $P = \sum_{k=1}^{\text{ord}\, g}g^k f$, where $g\in\mathcal{M}$ is such that $\text{ord}\, g$ is even, and $f$ is an arbitrary (non-zero) polynomial, admits a cup product.
\end{lem}
\begin{proof}
    It suffices to show that the polynomials of weight 2 or of the form of $P$ satisfy the cup product conditions, since $A$, $B$ and $C$ must satisfy these conditions individually. To see that a weight-2 polynomial, $P_1+P_2$, satisfies the conditions, simply let $P_\text{in} \coloneqq P_1$ and $P_\text{out} \coloneqq P_2$. All conditions are thereby easily verified.

    Let $G = \sum_{k=1}^{\text{ord}\, g} g^k$. We first divide the terms of $f$ into in, out and free terms: $f=f_\text{in} + f_\text{out} + f_\text{free}$. If $|f|>1$, we will choose $\lfloor |f|/2\rfloor$ terms to belong to $f_\text{in}$, arbitrarily, and $\lfloor |f|/2\rfloor$ terms to belong to $f_\text{out}$, with any remaining term belonging to $f_\text{free}$. Note that this is not a unique choice and in fact we can choose any assignment of terms in $f$ to the three sets. If $|f|=1$ we will choose $f = f_\text{in}$. We will then choose $P_\text{in} = Gf_\text{in} \equiv G(f_\text{in}^{(1)}+f_\text{in}^{(2)} +\ldots)$, and similarly for $P_\text{out}$ and $P_\text{free}$. The fact that $\text{ord}\, g$ is even ensures that $|P_\text{in}\alpha| = |P_\text{out}\alpha|$ (mod 2). 

    We begin by showing that $|Gh_1\,\alpha \cap Gh_2\,\beta| = 0$ (mod 2) when $\alpha\neq \beta$, where $h_1$ and $h_2$ are arbitrary polynomials. Let us suppose that $Gh_1\, \alpha\cap G h_2\, \beta \neq \emptyset$, so that $g^k h_1^{(i)}\alpha = g^l h_2^{(j)}\beta$, for some integers $k,l,i,j$, where $h_1^{(i)}$ is the $i^\text{th}$ term in $h_1$, etc. But this implies the following relations also hold:
    \begin{align*}
        g^{k+1}h_1^{(i)}\alpha &= g^{l+1}h_2^{(j)}\beta\\
        g^{k+2}h_1^{(i)}\alpha &= g^{l+2}h_2^{(j)}\beta\\
        &\ldots\\
        g^{k+\text{ord}\,g-1}h_1^{(i)}\alpha &= g^{l+\text{ord}\,g-1}h_2^{(j)}\beta.
    \end{align*}
    Since $\text{ord}\,g$ is even, and all the above are distinct, there are an even number of such relations. If we can find another relation, $g^{k'}h_1^{(i')}\alpha = g^{l'}h_2^{(j')}\beta$, this will produce a set of relations either distinct from the above entirely, or it will reproduce the set of above relations. This is because, if we reproduce one of the above relations, then $g^{k+M}h_1^{(i)}\alpha = g^{l+M}h_2^{(j)}\beta = g^{k'}h_1^{(i')}\alpha = g^{l'}h_2^{(j')}\beta$, for some integer $M$, and it is easy to see that all relations will be reproduced. We conclude therefore, $|Gh_1\alpha\cap Gh_2\beta| = K\text{ord}\, g = 0$ (mod 2), where $K$ is an integer.

    From the above, we conclude that $|P_\text{in}\alpha\cap P_\text{in}\beta|= |P_\text{in}\alpha\cap P_\text{free}\beta| = |P_\text{out}\alpha\cap P_\text{out}\beta| = |P_\text{free}\alpha\cap P_\text{out}\beta| = 0$ (mod 2), for $\alpha\neq \beta$. To show that Equation~\ref{eqn:cup_cond_4} holds, we note that $P_\text{in}\alpha_2\cap P_\text{in}\alpha_3 = G h$, for some (potentially zero) polynomial $h$, and similarly for $P_\text{out}\cap P_\text{out}$ and $P_\text{in}\cap P_\text{free}$. We can use this, along with the above, to show that Equation~\ref{eqn:cup_cond_4} holds.
\end{proof}

\subsection{Improved Ansatz for Group Algebra Codes}\label{app:inv_free_ansatz}

The classical group algebra codes obtained from any of the polynomials of Lemma~\ref{lem:cup_product} with weight larger than 2 result in a quantum code with distance-2. The reason is that these polynomials admit a factorisation of the form:
\begin{align}
    P = (1+g^{(\text{ord}\, g)/2})\sum_{k=1}^{(\text{ord}\, g)/2} g^k f.
\end{align}
In the above, $g^{(\text{ord}\, g)/2}$ is what is known as an involution: a group element that squares to the identity~\cite{tiew2026copycupgatestensorproducts}. This results in a length-2 codeword of the classical code, given by $(1+g^{(\text{ord}\, g)/2})$, since this polynomial squares to give zero, and hence it is in the kernel of $P$. Such a codeword lifts to a $Z$-logical operator of the TT code with $P$ as one of its defining polynomials, meaning that the resulting code has distance $\leq 2$ (see Ref.~\onlinecite{tiew2026copycupgatestensorproducts}). 

However, here we provide an ansatz that raises the distance floor above $d>2$:
\begin{align}
    P_\text{in} &= \sum_{k=1}^{\text{ord}\, g}g^k f \\
    P_\text{out} &= \sum_{k=1}^{\text{ord}\, h}h^k p,
\end{align}
where $g\neq h$ are such that $\text{ord}\, g \equiv \text{ord}\, h \equiv 0$ mod $2$, and $f$ and $p$ are arbitrary (non-zero) polynomials. We set $P_\text{free} = 0$. 

Given this ansatz, Equations~\ref{eqn:cup_cond_1} holds automatically, and Equations~\ref{eqn:cup_cond_in}--\ref{eqn:cup_cond_4} reduce to $|P_\text{in}\alpha\cap P_\text{in}\beta| \equiv |P_\text{out}\alpha\cap P_\text{out}\beta| \equiv 0$ mod $2$ for all $\alpha\neq \beta$, which can be seen to hold from Lemma~\ref{lem:cup_product}. So long as $g\neq h$, the nilpotent elements $(1+g^{(\text{ord}\, g)/2})$ and $(1+h^{(\text{ord}\, h)/2})$ are not by themselves codewords of $P = P_\text{in} + P_\text{out}$: only their product is annihilated by $P$. Hence, in this fashion we can obtain codes with distances $d\leq 4$. 

We also present a further ansatz that contains free elements. Let $G = \sum_{k=1}^{\text{ord}\, g} g^k$. Then we take:
\begin{align}
    P_\text{in} &= Gf\\
    P_\text{out} &= Gh\\
    P_\text{free} &= G' p
\end{align}
where $f$, $h$ and $p$ are arbitrary, non-zero polynomials that share no common terms, and $G'$ is a sum of an even number of elements of the form $g^k$. Of course, if $G' = G$, we reproduce the ansatz from Lemma~\ref{lem:cup_product}, so instead we take the elements of $G'$ to form an even-cardinality proper subset of the elements of $G$. We then see that $|P_\text{in}\alpha \cap P_\text{free}\beta| \equiv |P_\text{out}\alpha\cap P_\text{free}\beta| \equiv 0$ mod $2$ ($\alpha\neq \beta$) because $Gf \alpha$ overlaps $G'p\beta$ either on the entirety of its elements (since $G'\subset G$, abusing notation) or not at all. Thus the even number of elements of $G'$ ensures the even overlap condition holds. This observation along with the above considerations is sufficient to show that Equations~\ref{eqn:cup_cond_1}--\ref{eqn:cup_cond_4} are all satisfied by this ansatz too. Now note that $G' G = 0$, and hence $G$ is a valid codeword of $P = P_\text{in} + P_\text{out} + P_\text{free}$, meaning that the distance is upper bounded by $|G|$.

\subsection{Circuit implementing logical non-Clifford gates}\label{app:ccz_circuit}

We now explain the constant-depth $CCZ$ circuit that preserves the stabilizer for codes that admit a cup product. We apply the following gates to the data qubits:

\begin{align}\label{eqn:ccz_circuit}
\mathcal{C}_{CCZ} = \prod_{\alpha\in\mathcal{M}}\ \ \prod_{i,j,j',k,l} CCZ\left(q(\text{Q}_1,\alpha),q(\text{Q}_2,\beta),q(\text{Q}_3,\gamma)\right),
\end{align}
where
\begin{align}
    \beta &= \mathcal{Q}_{2,\text{out}}^{(i)}\left(\mathcal{Q}_{1,\text{in}}^{(j)}\right)^\top \alpha, \label{eqn:ccz_hop_1}\\
    \gamma &= \mathcal{Q}_{3,\text{out}}^{(l)}\left(\mathcal{Q}_{2,\text{in}}^{(k)}\right)^\top\mathcal{Q}_{1,\text{in}}^{(j)}\left(\mathcal{Q}_{1,\text{in}}^{(j')}\right)^\top \beta, \label{eqn:ccz_hop_2}
\end{align}
for all choices of data qubit blocks Q$_1$, Q$_2$, Q$_3$, which are all distinct elements from $\lbrace \text{L}, \text{C}, \text{R}\rbrace$. In the above, we use $\mathcal{Q}_i$ to refer to the polynomial from block Q$_i$ in $H_X$. I.e., if Q$_i = \text{L}$, $\mathcal{Q}_i = A$, if Q$_i = \text{C}$, $\mathcal{Q}_i = B$, and if Q$_i = \text{R}$, $\mathcal{Q}_i = C$. The superscripts label the individual terms of the relevant in- or out-set, and the indices $i,j,j',k,l$ each run over all terms of $\mathcal{Q}_{2,\text{out}}$, $\mathcal{Q}_{1,\text{in}}$, $\mathcal{Q}_{1,\text{in}}$, $\mathcal{Q}_{2,\text{in}}$ and $\mathcal{Q}_{3,\text{out}}$ respectively.

\subsection{Details of numerics used to verify stabilizer preservation and logical action of $CCZ$ circuits}

Here, we detail the methods used to verify that the codes of \Cref{sec:CCZ_gates} admit a logical $CCZ$ circuit, and that the logical action of this circuit is non-trivial. For each code, we make an assignment of terms in $A$, $B$ and $C$ into in, out and free subsets, as detailed above. We then build a description of the $CCZ$ circuit from Equation~\ref{eqn:ccz_circuit}: for each gate in the circuit, we include in the circuit description a tuple of the form $(q_1,q_2,q_3)$, for row vectors $q_i$ corresponding to qubits in the support of the gate, and where the three slots of the tuple correspond to the three blocks of code. Then for every $X$-check in every code block $i$, we check that the $CCZ$ circuit maps it to a $CZ$-type operator that conjugates any logical or $X$-stabilizer in block $j\neq i$ to a $Z$-stabilizer of block $k\neq j,i$. We describe this procedure in Algorithm~\ref{alg:ccz_circuit_check}. %

Let us determine this criterion more precisely. Take $X$-stabilizer $S_X$ of block $i$ (with the other two blocks labelled $j$ and $k$) and let $CZ_S$ be the circuit of $CZ$ gates that obeys $\mathcal{C}_{CCZ}S_X\mathcal{C}_{CCZ} = S_X CZ_S$. Then define $M_S$ to be the $n\times n$ binary matrix whose $(a,b)$ entries count (modulo 2) the number of $CZ$ gates in $CZ_S$ acting between the qubits indexed $b$ in block $k$ and $a$ in block $j$. Then $CZ_S$ acts as the identity on the code spaces of blocks $j$ and $k$ if and only if:
\begin{align}\label{eqn:CZ_rule}
    c_j^\top M_S c_k \equiv 0 \quad (\text{mod }2), \qquad \forall c_j,c_k \in \ker(H_Z).
\end{align}
That is, every $X$-stabilizer or $X$-logical of block $k$ is mapped under $CZ_S$ to a $Z$-type operator in block $j$ that commutes with every $X$-stabilizer/logical on that block and, hence, is itself a $j$-block $Z$-stabilizer. Let $K = [H_X^\top, L_X^\top]^\top$ be the matrix whose rows span $\ker(H_Z)$. Then Equation~\ref{eqn:CZ_rule} is equivalent to $K M_S K^\top \equiv 0$ (mod $2$). Algorithm~\ref{alg:ccz_circuit_check} tests this condition. We note in the algorithm that rows$(A)$ corresponds to the set of rows of matrix $A$, while rs$(A)$ corresponds to the row space of $A$ (i.e., the span of all rows of $A$ over the field $\mathbb{F}_2$). 
We also note that it suffices to pick one stabilizer (labelled $\texttt{row}_1$ in Algorithm~\ref{alg:ccz_circuit_check}) of the $i^\text{th}$ code block arbitrarily. This is because all other $i$-block $X$-stabilizers are related by a monomial. This monomial can be absorbed into the product over $\alpha\in \mathcal{M}$ in the circuit (Equation~\ref{eqn:ccz_circuit}) and hence if the circuit conjugates $\texttt{row}_1$ to a code-space-preserving $CZ$ circuit, it does so for all choices of $\texttt{row}_1$. This is simply a result of the translation invariance of the code.

We next check which logical qubits are acted upon non-trivially by the logical gate implemented by the $CCZ$ circuit. We begin by producing a minimum-weight set of $Z$-logicals, $L_Z$, using the QDistRnd algorithm~\cite{Pryadko2022}. We then produce a set of $X$-logicals, $L_X$, that are conjugate to those minimum-weight $Z$-logicals: $L_X L_Z^\top = \mathds{1}_k$. We then determine the logical action of the $CCZ$ circuit on this set of $X$-logicals by determining the ``intersection number" for each triple of logicals, which counts the number of $CCZ$ gates in circuit~\ref{eqn:ccz_circuit} that intersect $l_1$ in code block 1, $l_2$ in block 2 and $l_3$ in block 3, where $l_i\in \text{rows}(L_X)$. If and only if this intersection number $I(l_1,l_2,l_3) = 1$ (mod 2), the $CCZ$ circuit implements the logical gate $\overline{CCZ}(l_1,l_2,l_3)$ between the three code blocks. %
The tensor $I(l_1,l_2,l_3)$ will depend on the choice of logical $X$ basis, but we verify that it is non-zero, indicating that the $CCZ$ circuit acts non-trivially on the logical space.

\begin{algorithm}
\caption{Check $CCZ$ Circuit Preserves Code's Stabilizer Group}
\label{alg:ccz_circuit_check}

\KwIn{Matrices $H_X$, $H_Z$, $L_X$; List, $C$, of $CCZ$ gates in circuit~\ref{eqn:ccz_circuit} (each is a length-3 tuple of qubit indices corresponding to the support of that gate in the three code blocks)}

\KwOut{\textbf{true} if $CCZ$ circuit is valid, \textbf{false} otherwise}

$n \gets$ number of columns in $H_X$\;
$K \gets [H_X^\top, L_X^\top]^\top$\;

\For{$i$ in $(1,2,3)$}{
$(j,k) \gets$ the two blocks other than $i$\;
$\texttt{row} \gets $ First row of $H_X$\;
Ind($\texttt{row}$) $\gets$ List of indices corresponding to non-zero elements of $\texttt{row}$\;
$M \gets n\times n$ matrix of zeros;

\For{each gate $(q_1,q_2,q_3)$ in $C$}{
        \If{$q_i \in$ \textnormal{Ind}($\texttt{row}$)}{
            $M[q_j, q_k] \gets M[q_j, q_k] + 1$\;
        }
    }
    $M \gets M \bmod 2$\;
    \If{$K M K^\top \neq 0 \pmod 2$}{
        \Return{\textbf{false}}\;
}}

\Return{\textbf{true}}\;

\end{algorithm}

\subsection{Gauge-fixing to obtain higher-distance codes with logical $CZ$ gates}

\begin{table*}
    \begin{center}
    \def\arraystretch{1.5}
    \[\begin{array}{|c|c|c|c|c|}
    \hline
        \textbf{Code parameters} & \ell, m, p & \begin{array}{c}\textbf{Z-logical} \\ \textbf{weights}\end{array} & \mathbf{d_X} & \textbf{Polynomials} \\
        \hline \hline
        [[36, 6, 2]] & 3, 2, 2 & 3,3,3,3,3,2 & 3 %
        & \begin{array}{c}
            A = (1+z)(1+x)\\
            B = 1 + x\\
            C= 1+xyz
        \end{array} \\[2.0em]
        \hline
        [[48, 6, 2]] & 4, 2, 2 & 5,5,4,4,4,2 & 4 %
        & \begin{array}{c} 
        A = (1+x^2yz)(1+xz)\\
        B = 1 + x^3\\
        C = 1 + x^3yz
        \end{array} \\[2.0em]
        \hline
        [[54, 9, 2]] & 3, 3, 2 & \begin{array}{c}5,5,3,3,\\3,3,3,2,2\end{array} & 4 %
        & \begin{array}{c} 
        A = (1+z)(1+y)\\
        B = 1 + xyz\\
        C = 1 + x^2y^2
        \end{array} \\[2.0em]
        \hline
        [[108, 6, 2]] & 4, 3, 3 & 7,6,6,6,6,2 & 6 %
        & \begin{array}{c} 
        A = (1+x^2)(1+xz)\\
        B = 1 + x^2y^2\\
        C = 1 + x^2y^2z^2
        \end{array} \\[2.0em]
        \hline
        [[108, 18, 2]] & 4, 3, 3 & \begin{array}{c}\lbrace 7\rbrace \times 5, \lbrace 6\rbrace \times 2,\\ \lbrace 2\rbrace \times 11 \end{array} & 4 %
        & \begin{array}{c} 
        A = (1+x^2)(1+xy^2z^2)\\
        B = 1 + x^2y^2z\\
        C = 1 + x^2y^2z
        \end{array} \\
        \hline
    \end{array}\]
    \end{center}
\caption{Examples of (4,2,2) codes with $CCZ$ gates that have non-trivial logical action. Third column: Rather than simply reporting the $Z$-distance (which is 2 in all cases), we report the minimum weights found for all independent $Z$-logical operators. Fifth column: The $A$ polynomial is chosen to have weight-4, and to be of the form $(1+g)(1+a)$, where $g^2=1$ and $a\neq 1$ is arbitrary. %
\label{tab:422_codes}
}
\end{table*}

In Table~\ref{tab:422_codes}, we provide examples of $(4,2,2)$ codes with distance 2 (all of the form given by Lemma~\ref{lem:cup_product}). These codes all admit a cup product and hence a constant-depth circuit of the form given in Equation~\ref{eqn:ccz_circuit} which acts non-trivially on the code space. The codes listed have several weight-2 $Z$-logicals but $X$-distances larger than 2.

We find that gauge fixing the distance-2 logicals by adding weight-2 $Z$-stabilizers to the code results in the $\overline{CCZ}$ gate's logical action becoming trivial. However, the $\overline{CCZ}$ gates do still act non-trivially on the logical qubits with higher-weight logicals. Hence, gauge fixing the distance-2 logical qubit(s) to the $\ket{0}$ state results in an \textit{error-correcting} code that loses its non-Clifford gate but has an \textit{additional} $CZ$ gate, on top of those transversal $CZ$ gates found in Section~\ref{sec:logic_gates}. %
This $\overline{CZ}$ gate need not have a trivial logical action across two code blocks. We provide an example of this procedure, focusing on the $[[36, 6, 2]]$ code from \Cref{tab:422_codes}. Observe that the $CCZ$ circuit implements (among others) the following logical gates (for one choice of basis with $Z$-logical weights given in \Cref{tab:422_codes}): $\overline{CCZ}(5,1,3)\overline{CCZ}(5,3,1)\overline{CCZ}(5,2,4)\overline{CCZ}(5,4,2)$. In this circuit, logical qubit $5$ has $Z$ logical with minimum-weight 2, while all others have $Z$ logicals with minimum-weight 3. All other logical gates involve qubit 5 from code blocks 2 and/or 3. The steps to obtain a series of $\overline{CZ}$ gates on a gauge-fixed code are:
\begin{enumerate}
    \item Gauge fix the second and third blocks of $[[36, 6, 2]]$ code so that the distance-2 qubit (qubit 5) is in the $\ket{0}$ state. This results in a $[[36, 5, 3]]$ code. Note that this gauge fixing involves measuring $\bar{Z}_5$, which commutes with the $CCZ$ circuit.
    \item Consider the gate $\mathcal{C}_{CCZ} \bar{X}_5\otimes \bar{\mathds{1}} \otimes \bar{\mathds{1}} \mathcal{C}_{CCZ}$, where $\mathcal{C}_{CCZ}$ is the $CCZ$ circuit obtained from the cup product (given in Appendix~\ref{app:ccz_circuit}), and $\bar{X}_5\otimes \bar{\mathds{1}} \otimes \bar{\mathds{1}}$ is the logical $X$ operator for qubit 5 in the first code block. The logical action of this gate between code blocks 2 and 3 is: $\overline{CZ}(1,3)\overline{CZ}(3,1)\overline{CZ}(2,4)\overline{CZ}(4,2)$.
    \item Delete the first code block and truncate circuit $\mathcal{C}_{CCZ} \bar{X}_5\otimes \bar{\mathds{1}} \otimes \bar{\mathds{1}} \mathcal{C}_{CCZ}$ to a circuit of $CZ$s acting on code blocks 2 and 3. 
\end{enumerate}
In this way, one can obtain a larger number of Clifford gates in gauge-fixed versions of the codes from \Cref{tab:422_codes}, which will have higher distances.

\bibliography{ref}

\end{document}